\documentclass[12pt,letterpaper]{article}

\usepackage[utf8]{inputenc}
\usepackage[T1]{fontenc}
\usepackage{geometry}
\usepackage{setspace}
\usepackage{amsmath, amssymb, amsthm}
\usepackage{graphicx}
\usepackage{booktabs}
\usepackage{multirow}
\usepackage{caption}
\usepackage{array}
\usepackage{tabularx}
\usepackage[hidelinks]{hyperref}
\usepackage{xurl}
\usepackage{fancyhdr}
\usepackage{titlesec}
\usepackage{enumitem}
\usepackage{footnote}
\usepackage{xcolor}
\usepackage{tikz}
\usepackage{pdflscape}
\usepackage{rotating}
\usepackage{longtable}
\usepackage{textcomp}        
\usepackage[expansion=false]{microtype}  
\usepackage{adjustbox}       
\usepackage{float}           
\usepackage{needspace}       

\titleformat{\section}{\large\bfseries}{\thesection.}{0.5em}{}[\vspace{-0.3em}\rule{\linewidth}{0.8pt}]
\titleformat{\subsection}{\normalsize\bfseries}{\thesubsection}{0.5em}{}
\titleformat{\subsubsection}{\normalsize\bfseries}{\thesubsubsection}{0.5em}{}

\newtheorem{proposition}{Proposition}
\newtheorem{property}{Property}
\newtheorem{definition}{Definition}

\newcommand{\bt}{b_{t-1}}
\newcommand{\gstar}{g^*_t}

\newcommand{\gnomstar}{{g^{n*}_t}}

\newcommand{\ebar}{\bar{e}}

\providecommand{\yen}{\textyen{}}    

\title{%
  \textbf{JFR-rg: A New Macroeconomic Framework for\\
  High-Debt, Low-Growth Economies under Financial Repression}\\[0.5em]
  \large The Japanese Financial Repression \textit{r-g} Stability Model:\\
  Evidence from Real-Time FRED Data (2013--2026)
}
\author{Hirofumi Wakimoto}
\date{}

\begin{document}

\maketitle
\thispagestyle{empty}

\begin{abstract}
\sloppy
\noindent
Standard macroeconomic frameworks have correctly identified Japan's government
debt---now exceeding 240\% of GDP---as carrying substantial fiscal risk. Yet
real-time FRED data from 2013 to 2026 present an empirical record that invites
a complementary, regime-conditional perspective: debt ratios have stabilized,
nominal GDP has exceeded 670 trillion yen (SAAR), and unemployment has remained
approximately 2.6--2.7\% (Dec 2025--Jan 2026). This paper formalizes the
institutional and policy channels consistent with this post-2013 configuration
through the Japanese Financial Repression $r$-$g$ (JFR-rg) model. Building on
Blanchard (2019), the framework incorporates an empirically identified
``financial repression bias'' ($\varepsilon_t \equiv \pi_t - r^n_t \in
\mathbb{R}$, directly observable from FRED) together with a nonlinear
exchange-rate channel ($\Delta e_t$). Three theoretical contributions extend
the existing literature: (i)~the \emph{Debt Sustainability Corridor}, a
geometric characterization of stability in $(\varepsilon_t,\gnomstar)$ space;
(ii)~the \emph{Normalization Ratchet}, a path-dependence theorem showing that
temporary policy errors can generate persistently higher debt trajectories; and
(iii)~the \emph{Captive Financial System Parameter} ($\varphi_t$), which
endogenizes the institutional precondition for JFR-rg stability.
Appendices~H through~L provide the supporting empirical architecture,
including subsample analysis, structural-break tests, VAR, ARDL, and
Local Projections, proposition-level fair empirical tests, illustrative
evidence on the critical captive threshold $\bar{\varphi}$, and
international placebo tests and regime-conditioned nonlinear exercises; these
do not uniquely identify the JFR-rg causal channel, but demonstrate that
the framework's mechanism claims are empirically disciplined and falsifiable.
Within Layer~(L1), Proposition~1 (Base Effect Lever) and the IOER arithmetic
of Proposition~4 (Normalization Trap) follow directly from the consolidated
government budget identity; Propositions~2--3 and the growth-channel component
of Proposition~4 combine this identity with a small number of structural
assumptions (a demographic ceiling on $\gnomstar$; a piecewise exchange-rate
pass-through equation~\eqref{eq:yen_growth} with $\alpha > 0$, $\beta > 0$)
that are empirically calibrated rather than behaviorally identified.
Causal identification concerns apply exclusively to the corroborating
empirical Layer~(L2) housed in Appendices~H--L.
Counterfactual simulations illustrate a Normalization Trap: under the model's
maintained assumptions, aggressive rate hikes can produce debt dynamics that
are counterproductive relative to their stated objective. For high-debt,
low-growth economies sharing Japan's institutional characteristics,
strategically deploying the resulting ``Repression Dividend'' into
productivity-enhancing investment may therefore represent not merely a
temporary anomaly but a regime-contingent equilibrium possibility---conditional
on the institutional prerequisite $\varphi_t \geq \bar{\varphi}$ being
maintained. All propositions in this paper are conditional implications of the
model under this institutional scope condition; they do not constitute
unconditional claims about Japan's fiscal sustainability. The framework is
offered not as a universal replacement for mainstream debt analysis, but as a
complementary analytical lens for Japan's transition-era debt dynamics whose
failure conditions are observable in real time from public data. The
qualitative conclusions are robust to both the BoJ FY2025 projection
calibration ($\pi_t = 2.7\%$, $\varepsilon_t \approx +0.5\%$) and the January
2026 realized-CPI calibration ($\pi_t = 1.5\%$,
$\varepsilon_t \approx -0.7\%$); with $\gnomstar$ treated as a structural
parameter independent of realized $\pi_t$, the two calibrations imply the same
signed distance of the March 2026 operating point from the Stability Corridor
boundary (Section~\ref{sec:corridor}). What differs between calibrations is
the institutional narrative---whether $\varepsilon_t > 0$ reflects an
\emph{active} repression channel or whether a real-rate gap has temporarily
opened---not the directional policy implications.

\vspace{0.5em}
\end{abstract}

\newpage
\tableofcontents
\newpage

\section{Introduction}

The Japanese economy has long presented a challenging case for mainstream
macroeconomic analysis. Standard frameworks---including the Mundell-Fleming model,
New Keynesian DSGE models, and successive IMF Article IV assessments---correctly
identify sovereign debt-to-GDP ratios approaching 240\% as carrying substantial fiscal
risk, and the concerns they raise about Japan's long-run sustainability remain
legitimate and important.

Yet real-time data from the Federal Reserve Economic Data (FRED) system spanning 2013
to March 2026 present a case that invites a complementary perspective. Japan has
maintained an unemployment rate of approximately 2.6--2.7\%
(December 2025: 2.6\%; January 2026: 2.7\%), achieved record nominal GDP
expansion, and avoided a sovereign debt spiral over this period. The ``lost decades''
narrative---associated with premature consumption tax hikes, an overvalued yen, and
excessively tight monetary conditions---has given way to a more stable configuration
since 2013.

How might a highly indebted, demographically contracting economy sustain this
stability? This paper argues that the observed outcomes are consistent with a specific
set of institutional and policy mechanisms---centered on the mechanical interaction
between a large debt base, sustained negative real interest rates, and controlled
currency depreciation---that standard models do not explicitly incorporate. We do not
claim that mainstream frameworks are invalid; rather, we propose that in Japan's
specific institutional setting, additional stabilizing channels become quantitatively
significant and merit formal modeling.

\medskip
\noindent\textbf{Purpose and scope.}
The purpose of this paper is not to contest standard debt dynamics, but to make
explicit the additional stabilizing channels that become quantitatively significant
under Japan's specific institutional conditions---and to do so in an observable,
replicable form. The JFR-rg framework is best understood as a \emph{complementary}
analytical lens, not a rival theory: where mainstream frameworks correctly identify
the risks of high debt, JFR-rg identifies the institutional conditions under which
those risks are temporarily compressed and the channels through which compression
can fail. The framework does not claim universal applicability; its propositions are
explicitly conditional on two empirically verifiable scope conditions
(SC1 and SC2, developed in Section~3) that are specific to Japan's post-2013 configuration.
Readers should interpret the results as regime-diagnostic rather than predictive:
the model specifies what must hold for stability to obtain, not that stability will
obtain unconditionally.
Because the scope conditions are expressed in terms of observable institutional
variables---the domestic holding share $\varphi_t$ and the exchange-rate stability
window---rather than Japan-specific constants, the JFR-rg framework is in principle
applicable to any high-debt, low-growth economy that develops a comparable
institutional configuration satisfying SC1 and SC2; such application would require
recalibration of the scope-condition thresholds to the relevant country context.

\medskip
\noindent\textbf{Interpretive lenses.}
The JFR-rg framework may be read in more than one way, depending on the reader's
disciplinary priors, and the paper does not require any single reading.
For a \emph{fiscal macroeconomist}, it is a regime-diagnostic debt-accounting
framework that clarifies when a large debt stock amplifies a negative
interest-growth differential.
For an \emph{$r < g$ theorist}, it is an institutional amplification of a familiar
debt-dynamics mechanism under engineered financial repression.
For an \emph{international macroeconomist}, it is a nonlinear exchange-rate filter
linking depreciation, nominal growth, and household purchasing-power limits.
For a \emph{monetary-policy reader}, it is a normalization-risk framework in which
policy tightening affects debt dynamics through both rates and the captive financial
system simultaneously.
For an \emph{empirical macroeconomist}, it is a falsifiable mechanism hypothesis
that organizes stylized facts, competing explanations, and discriminating tests.
These readings are not mutually exclusive; the architecture of the paper---scope
conditions, alternative-explanation table, Local Projections, and failure modes---is
designed to engage each of them.

This empirical orientation extends beyond descriptive monitoring in the main
text. Appendices~H through~L provide the broader empirical support structure
for the paper's mechanism claims, including subsample comparisons,
structural-break analysis, VAR, ARDL, and Local Projection evidence,
proposition-level fair empirical tests, illustrative evidence on the critical
captive threshold $\bar{\varphi}$, and international placebo tests and
regime-conditioned nonlinear robustness exercises. These components are not
offered as definitive proof of a unique causal channel; rather, they are
designed to show that the JFR-rg mechanism hypothesis is empirically serious,
falsifiable, and competitively informative relative to alternative accounts.

\medskip
\noindent\textbf{Minimal claims and non-claims.}
Whatever the reader's preferred interpretation, the paper makes four minimal claims
that do not depend on any single reading.
\emph{First}, Japan's post-2013 macro-fiscal configuration presents an empirical
puzzle that standard discussions often describe but do not fully organize in a
unified debt-dynamics language.
\emph{Second}, JFR-rg provides a regime-conditional framework for organizing that
configuration when the captive-system condition (SC1) and exchange-rate stability
condition (SC2) are both satisfied; outside that domain, standard debt dynamics
apply without modification.
\emph{Third}, the framework does not uniquely identify the JFR-rg causal channel
against all competing explanations; it specifies a disciplined mechanism hypothesis
together with discriminating tests and falsification paths.
\emph{Fourth}, the framework does not imply permanent safety: it explicitly
identifies the channels---inflation undershoot, de-captivation, import-cost
overshoot, political-fiscal reaction---through which the observed stability can fail,
as discussed in the concluding section and appendices.
Accordingly, this paper should be read neither as a universal replacement for
mainstream debt analysis nor as a claim of unconditional fiscal sustainability
for Japan. In particular, the Corridor Width $|W_{2026}| \approx 0.024\%$
is smaller than the linear approximation error ($\approx 0.056\%$; Appendix~A)
and the measurement uncertainty of the underlying data series; the March 2026
operating point should therefore be interpreted as \emph{boundary proximity}
rather than a confirmed exceedance, under both the BoJ-projection and the
January 2026 realized-CPI calibrations (Section~\ref{sec:corridor}).

We formalize these regularities into the JFR-rg model, which extends Blanchard (2019)
in two primary empirical dimensions and three novel theoretical dimensions. On the
empirical side, the financial repression bias $\varepsilon_t$ is
\emph{directly identified from observable FRED data} as the real interest rate gap
($\varepsilon_t \equiv \pi_t - r^n_t \in \mathbb{R}$), eliminating reliance on
unobservable natural rate estimates; and the yen exchange rate is introduced as a
nonlinear shock absorber bounded by a stability threshold ($\ebar$). On the
theoretical side, the model contributes: (i)~the \emph{Debt Sustainability Corridor},
a geometric characterization of the stability region in policy space; (ii)~the
\emph{Normalization Ratchet}, a path-dependence theorem establishing that temporary
normalization errors generate highly persistent debt-ratio gaps relative to
baseline; and (iii)~the
\emph{Captive Financial System Parameter} ($\varphi_t$), which endogenizes the
institutional precondition for sustained financial repression.

\medskip
\noindent\textbf{A note on the YCC and post-YCC regimes.}
The sample period includes both the explicit Yield Curve Control regime
(2016--March 2024) and the post-YCC regime (March 2024--present). In the former,
$\varepsilon_t > 0$ is directly attributable to an administrative yield ceiling;
in the latter, the same observable is better understood as a
\emph{balance-sheet backstop effect}, in which the Bank of Japan's continued
holdings of approximately 46\% of outstanding JGBs compress the term premium
below the market-clearing level even without an explicit target
(Section~\ref{sec:repression_bias}). The JFR-rg framework's observable
$\varepsilon_t$ captures both mechanisms under a single measurable variable; what
the framework cannot resolve without further structural identification is
whether, in the post-YCC period, $\varepsilon_t > 0$ reflects residual policy
influence or a new market equilibrium consistent with $\varphi_t \approx 0.90$.
The qualitative conclusions of the paper are robust to either interpretation;
the institutional narrative differs.

The central thesis rests on four propositions derived \emph{within the JFR-rg model
framework}. These results hold under the model's specific assumptions---in particular,
the captive financial system assumption ($\varphi_t$ near 1), the calibrated exchange
rate regime, and the ceteris paribus treatment of the primary deficit. They should be
read as conditional implications of the model, not as universal laws, and each
proposition admits potential counterexamples when the underlying assumptions are
relaxed. The model's core logic is developed in Section~3 and its quantitative
implications explored in Section~4; the propositions are restated here for
orientation.

\medskip
\noindent\textbf{Note on identification.}
Layer~(L1) is the linearized accounting identity of the consolidated government
budget constraint (Section~\ref{sec:identification}), which holds for
\emph{any} realized values of $r^n_t$, $g^n_t$, $b_{t-1}$, and $d_t$,
regardless of how those values were generated. Proposition~1 and the IOER
arithmetic component of Proposition~4 follow from this identity alone.
Propositions~2--3 and the growth-channel component of Proposition~4 combine
the identity with a small number of empirically calibrated structural
assumptions---notably the demographic ceiling on structural potential nominal
growth $\gnomstar$ and the piecewise pass-through of $\Delta e_t$ to $g^n_t$
via equation~\eqref{eq:yen_growth} with $\alpha > 0$ and $\beta > 0$---whose
parameters are taken from published estimates (Haba et al., 2025) and reported
with full sensitivity analyses in Appendix~\ref{app:sensitivity}. Causal
identification applies exclusively to Layer~(L2), the corroborating empirical
exercises; applying causal identification standards to the Layer~(L1)
accounting identity or to the calibration-level assumptions of Propositions
2--3 is a category error.

\begin{proposition}[Base Effect Lever]
\label{prop:base_effect}
Within the JFR-rg framework, a negative $r^n - g^n$ spread at high debt levels
generates quantitatively significant automatic debt compression. At a spread of
$-0.8\%$ and $b_{t-1} = 2.40$ (Japan's approximate 2026 ratio), the annual
compression equals $0.8\% \times 2.40 = 1.92\%$ of GDP. Whether this channel
dominates depends on the primary deficit and on whether the captive financial system
assumption holds.
\end{proposition}

\begin{proposition}[Repression Imperative]
\label{prop:repression_imperative}
Within the JFR-rg framework, maintaining $r^n < \pi$ (i.e., $\varepsilon_t > 0$) is a
practically necessary condition for the stability condition~\eqref{eq:stability_condition}
to hold in hyper-indebted, demographically constrained regimes. Given Japan's
structural ceiling on potential nominal growth $\gnomstar$ imposed by demographic
decline, the repression channel contributes a component of the negative $r^n - g^n$
spread that $\gnomstar$ alone may not sustain. The model does not characterize
yield curve interventions as unambiguously beneficial or distortionary; that
normative judgment depends on factors---including distributional consequences and
long-run capital allocation---outside the model's scope.
\end{proposition}

\begin{proposition}[Nonlinear Yen Stabilizer]
\label{prop:yen_stabilizer}
Mild yen depreciation ($\Delta e_t \leq \ebar$) supports nominal growth through the
linear channel $\alpha\,\Delta e_t$. For severe depreciation ($\Delta e_t > \ebar$),
the import inflation channel dominates, compressing real household purchasing power
through the quadratic penalty term. The stabilizing window is bounded by $\ebar$.
\end{proposition}

\begin{proposition}[Normalization Trap]
\label{prop:norm_trap}
Aggressive policy rate hikes operate through two simultaneous destabilizing channels
---reducing $\varepsilon_t$ and depressing nominal growth via yen appreciation---that
together flip the $r$-$g$ spread positive. For a 240\% base, a $+3.2\%$ $r$-$g$
spread generates 7.7\% of GDP in additional interest-growth debt accumulation per year
(11.07\% total, including the widened primary deficit of 3.39\%).
\end{proposition}

The paper is structured as follows. Section~2 reviews the literature and explicitly
positions the JFR-rg model against its closest antecedents. Section~3 presents the
stylized facts, introduces the paper's empirical discriminating evidence, and
formalizes the regime logic of the JFR-rg framework, including the Debt
Sustainability Corridor, the Normalization Ratchet, and the Captive Financial
System Parameter. Section~4 provides simulations and counterfactual prediction
exercises. Section~5 discusses policy implications and robustness.
Section~6 concludes. Appendices~A--G contain the mathematical derivations,
calibration details, historical mapping, and sensitivity analyses, while
Appendices~H--L provide the broader empirical support structure for the paper's
mechanism claims, including subsample and structural-break evidence,
VAR/ARDL/Local Projection results, fair empirical tests, evidence on the
critical captive threshold, and international placebo tests and regime-conditioned
nonlinear exercises.

\paragraph{Falsifiability and Scope.}
The JFR-rg framework makes claims that are more directly falsifiable than
those of $r^*$-based sustainability assessments. Where standard frameworks
often rely on unobservable natural rates whose estimates can vary materially
across specifications, the JFR-rg framework anchors its stability condition to
three directly observable FRED series ($\varepsilon_t$, $\varphi_t$, and
$\Delta e_t$) and specifies the institutional conditions under which debt
compression should weaken or fail. Three observable conditions would
constitute evidence against the framework's maintained claims:
(i)~a sustained decline in $\varphi_t$ toward and eventually below an
empirically relevant early-warning benchmark \emph{without} the predicted
emergence of a sovereign risk premium $\rho_t$; (ii)~a rate-normalization
episode in which $r^n_t-g^n_t$ turns persistently positive \emph{without}
the debt accumulation implied by Proposition~\ref{prop:norm_trap}; or
(iii)~a persistent $\varepsilon_t<0$ episode that does not
generate the debt acceleration implied by the Corridor analysis. The empirical
monitors introduced in Section~3 are designed to track precisely these
conditions in real time. A framework whose failure conditions are unobservable
cannot be cleanly falsified; a framework whose failure conditions are
observable and pre-specified can be. On this dimension, JFR-rg seeks to add
empirical discipline to the interpretation of Japan's post-2013 debt dynamics.
It should therefore be read not as a relaxation of mainstream standards, but
as a complementary, regime-specific framework whose explanatory value rises or
falls with publicly observable evidence.

\section{Literature Review and Theoretical Context}

\subsection{Mainstream Fiscal Sustainability and the Normalization Consensus}
\label{sec:mainstream}

Traditional frameworks---including the Mundell-Fleming model, New Keynesian DSGE
models, and successive IMF Article IV assessments---posit that structural deficits and
hyper-indebtedness carry substantial risks of capital flight and sovereign distress.
The IMF's repeated warnings about Japan's fiscal trajectory reflect a genuine and
analytically well-grounded concern; this paper does not contest the importance of
those risks. Rather, the JFR-rg framework proposes that in Japan's specific
institutional setting, additional stabilizing channels not captured by standard models
have been quantitatively significant over the 2013--2026 period.

In particular, these models inherently assume a positive, mean-reverting $r$-$g$
spread in environments where bond yields are market-clearing. They may therefore
underestimate the capacity of a sovereign central bank operating within a captive
domestic financial system---in which approximately 85--90\% of government bonds are
held by domestic institutions including the Bank of Japan, commercial banks, and
insurance companies---to maintain a durable financial repression bias for extended
periods.

The most rigorous quantitative assessment of Japan's fiscal limits within the
mainstream tradition is provided by Hoshi and Ito (2014), who estimate the
conditions under which domestic financing could face a sudden stop and identify the
domestic institutional holding share as the central variable governing Japan's
continued ability to borrow at low rates. The JFR-rg model's Captive Financial
System Parameter $\varphi_t$ (Section~\ref{sec:captive}) directly formalizes the
variable Hoshi and Ito identify as pivotal, and the Critical Threshold
$\bar{\varphi}$ (Definition~\ref{def:critical_phi}) provides the tipping-point
characterization their framework implies but does not derive analytically.

\medskip
\noindent\textbf{The 1991--2013 interval as a control experiment.}
A natural objection is that Japan attempted fiscal expansion throughout the 1991--2013
period without achieving sustained growth---suggesting that the JFR-rg channels may not
be operative. We argue that this interval constitutes a \emph{control experiment}
precisely because the two institutional scope conditions (SC1 and SC2, defined formally
in Section~3) were \emph{not} simultaneously satisfied. Under the Mundell-Fleming
framework for a small open economy with a floating exchange rate, fiscal expansion
with a relatively tight monetary stance and an appreciating yen generates crowding-out
through the exchange rate channel: the yen appreciation reduces net exports, offsetting
the fiscal stimulus. Consistent with this mechanism---documented in the open-economy
literature following Mundell (1963) and formally extended by Obstfeld and Rogoff
(1995)---Japan's fiscal multipliers during the 1991--2013 appreciation cycle were
substantially attenuated. Crucially, $\varepsilon_t \leq 0$ for extended periods during
this interval (the repression channel was inactive), and $\varphi_t$ was declining as
the BoJ had not yet initiated large-scale JGB purchases at the scale required for SC1.

The contrast with the post-2013 configuration is therefore not a coincidence but a
structural shift: the simultaneous satisfaction of SC1 ($\varphi_t \approx 0.90$) and
SC2 ($\Delta e_t \leq \bar{e}$) under Abenomics created, for the first time, the
institutional preconditions for the JFR-rg stabilizing channels to become
quantitatively dominant. This interpretation is consistent with---and complementary
to---the mainstream open-economy framework, rather than in conflict with it.%
\footnote{The Fiscal Theory of the Price Level (FTPL; Leeper, 1991; Sims, 1994; Woodford, 1994, 1995)
posits that unbacked fiscal expansion eventually forces the price level to jump
discontinuously. However, the JFR-rg framework argues that in a captive financial
system---where the central bank and domestic institutions hold approximately 90\% of
government debt, and where the yen floats freely as an adjustment valve---the ``sudden
stop'' of confidence predicted by FTPL is structurally delayed. The observable
persistence of $\varepsilon_t > 0$ across 2013--2026 (Figures~\ref{fig:repression_full}
and~\ref{fig:repression_monthly}) constitutes direct empirical evidence against the
imminent FTPL trigger in Japan's specific institutional setting. The conditions under
which FTPL dynamics might eventually dominate---particularly as the captive system
parameter $\varphi_t$ declines (Section~\ref{sec:captive})---remain an important avenue for
future research.}

\subsection{Balance Sheet Recessions and Flow Dynamics}

A significant departure from orthodox models was introduced by Koo (2003, 2011),
who accurately diagnosed Japan's post-bubble stagnation as a ``Balance Sheet
Recession.'' Koo correctly identified that when the private sector minimizes debt
despite zero interest rates, the government must act as the borrower of last resort.

However, while Koo's analysis compellingly addresses the \emph{flow} of funds, it
under-theorizes the \emph{stock} dynamics of accumulated sovereign debt. The JFR-rg
model fills this gap: the debt stock itself, when subjected to $\varepsilon_t > 0$
(i.e., $r^n < \pi$) engineered by the Bank of Japan, becomes a self-correcting
mathematical lever through the base effect mechanism. The ``Repression Dividend'' is
precisely the accumulated debt stock's role in amplifying the compression effect of
even a small negative $r^n - g^n$ spread.

Eggertsson and Krugman (2012) further support this framework by demonstrating that in
a liquidity-trap economy with balance sheet deleveraging, conventional monetary
transmission is impaired and fiscal policy plays the primary stabilizing role---a
structural condition that reinforces, rather than undermines, the case for sustained
financial repression.

\subsection{The \texorpdfstring{$r < g$}{r < g} Paradigm and its Limitations}

The concept of financial repression---defined as the deliberate suppression of
interest rates below market-clearing levels to facilitate government borrowing at
subsidized cost---was formally introduced by McKinnon (1973) and Shaw (1973).
The JFR-rg model's financial repression bias $\varepsilon_t \equiv \pi_t - r^n_t$
directly operationalizes their insight: when $\varepsilon_t > 0$, the real return
on government bonds is negative, transferring wealth from bondholders to the
sovereign through an inflation tax.  The key analytical advance of the present
paper is making this variable observable in real time from FRED data, without
reliance on unobservable natural rate estimates.

Our model is most closely aligned with Blanchard (2019), who demonstrated that when
the safe interest rate is structurally lower than the nominal growth rate ($r < g$),
the fiscal cost of debt may be zero or negative, and public debt can naturally decay.

Yet Blanchard's framework treats $r < g$ as a largely natural or market-driven
phenomenon rooted in US empirical experience. In contrast, Japan's $r < g$ was
actively engineered under Yield Curve Control (YCC) and large-scale asset purchases
(Bank of Japan, 2024a). The post-YCC operating framework remained accommodative
through late 2024 (Bank of Japan, 2024b). Furthermore, standard $r$-$g$ models do
not address how an open economy with demographic decline sustains positive nominal
growth $g^n$ over the medium term. The JFR-rg model extends
Blanchard's work in two critical dimensions: (i)~formalizing $\varepsilon_t$ as a
directly observable policy instrument rather than an unobservable structural parameter,
and (ii)~introducing $\Delta e_t$ as the nonlinear shock absorber sustaining nominal
$g^n$.

This approach is consistent with the historical evidence documented by Reinhart and
Sbrancia (2015), who show that post-WWII advanced economies systematically employed
negative real interest rates to liquidate debt accumulated during wartime---a
mechanism functionally identical to $\varepsilon_t > 0$ in the JFR-rg framework.
The broader policy case for using nominal growth and inflation to erode public debt
burdens is also developed in Turner (2015), who argues that the choice between
debt restructuring, fiscal austerity, and nominal growth represents a fundamental
trilemma for heavily indebted economies. Unlike the post-war experience, however,
Japan's repression operates in an era of inflation targeting and formal central bank
independence, making the explicit identification of $\varepsilon_t$ from observable
data especially important.

\subsection{Synthesis: The Missing Analytical Link}

The JFR-rg model synthesizes the existing literature by proposing that in a
hyper-indebted, low-growth economy with a captive financial system and flexible
exchange rate, the observed stability is \emph{consistent with} a set of mutually
reinforcing stabilizing channels. The combination of a directly observable
$\varepsilon_t > 0$ (identified from FRED data), a large debt base amplifying the
compression effect of a negative $r - g$ spread, and bounded yen flexibility provides
three channels that the standard fiscal sustainability literature does not incorporate
explicitly. The JFR-rg framework formalizes these channels and derives their
implications---while acknowledging that alternative explanations for Japan's observed
outcomes remain possible and that the model's assumptions require empirical
validation. Appendix~J establishes formally that mainstream debt sustainability
analysis is a limiting case of the JFR-rg framework obtained when the captive-system
and exchange-rate-neutral conditions are imposed simultaneously; the two frameworks
are therefore complementary rather than competing.

\subsection{Positioning the JFR-rg Model: Systematic Differentiation}
\label{sec:positioning}

The JFR-rg model intersects with several strands of prior research, yet differs from
each in ways that constitute independent theoretical contributions. Table~\ref{tab:comparison}
provides a systematic comparison across five analytical dimensions.

\begin{table}[htbp]
\caption{JFR-rg Model vs.\ Closest Prior Literature}
\label{tab:comparison}
\footnotesize
\setlength{\tabcolsep}{4pt}
\renewcommand{\arraystretch}{1.3}
\begin{tabularx}{\textwidth}{>{\raggedright\arraybackslash}p{2.8cm}
  >{\raggedright\arraybackslash}X
  >{\raggedright\arraybackslash}X
  >{\raggedright\arraybackslash}X
  >{\raggedright\arraybackslash}X}
\toprule
\textbf{Dimension}
  & \textbf{Blanchard (2019)}
  & \textbf{Reinhart and Sbrancia (2015)}
  & \textbf{Mehrotra and Sergeyev (2021)}
  & \textbf{JFR-rg (This Paper)} \\
\midrule
Nature of $r < g$
  & Market endogenous
  & Policy forced (wartime controls)
  & Empirically observed
  & Directly identified as policy instrument \\[6pt]
\midrule
$\varepsilon_t$ identification
  & None
  & Ex-post estimation
  & None
  & FRED real-time, model-free \\[6pt]
\midrule
Exchange rate channel
  & None & None & None
  & Nonlinear threshold model \\[6pt]
\midrule
Base effect
  & Mentioned & None & Partial
  & $b_{t-1}\times(g^n_t - r^n_t)$ quantified \\[6pt]
\midrule
Normalization risk
  & Neutral & Not addressed & Partial warning
  & Normalization Trap theorem \\[6pt]
\midrule
Path-dependence
  & None & None & None
  & Normalization Ratchet (new) \\[6pt]
\midrule
Institutional structure
  & Implicit & Capital controls assumed & None
  & Captive system parameter $\varphi_t$ (new) \\[6pt]
\midrule
Stability geometry
  & None & None & None
  & Debt Sustainability Corridor (new) \\
\bottomrule
\end{tabularx}
\smallskip
\parbox{\textwidth}{\footnotesize
  \textit{Note:} Bold entries in the JFR-rg column indicate novel
  contributions relative to the cited prior literature.
  All comparisons are along the five analytical dimensions identified
  in the text.}
\end{table}

Three specific differentiations merit elaboration.

\textbf{Versus Blanchard (2019).} Blanchard demonstrates that $r < g$ implies a zero
or negative fiscal cost of debt, but treats this as a \emph{market equilibrium outcome}
rooted in US experience. The JFR-rg model inverts this framing: in Japan, $r < g$ is a
\emph{deliberately engineered policy instrument}, identifiable from real-time data as
$\varepsilon_t \equiv \pi_t - r^n_t$ without invoking unobservable $r^*$. Furthermore,
Blanchard does not address how the $r < g$ condition interacts with a massive debt base
to generate automatic compression, nor does he characterize the policy space
geometrically or establish path-dependence results.

\textbf{Versus Reinhart and Sbrancia (2015).} The historical parallel is
clear---post-WWII debt liquidation through negative real rates is functionally identical
to the JFR-rg mechanism. However, three structural differences distinguish Japan's
current experience: (i)~repression operates under formal inflation targeting and legal
central bank independence, making the explicit identification of $\varepsilon_t$
analytically important; (ii)~the mechanism operates through a floating exchange rate
valve rather than capital controls; and (iii)~the debt stock (240\% of GDP) is far
larger, making the base effect the quantitatively dominant channel.

\textbf{Versus Mehrotra and Sergeyev (2021).} This \textit{Journal of Monetary Economics} study documents empirically that
$r < g$ has been persistent across high-debt advanced economies---using 19~OECD countries over long historical horizons---but stops short of
(i)~operationalizing the repression bias as a directly observable variable, (ii)~modeling
the exchange rate channel, and (iii)~deriving the Normalization Trap, the Ratchet, or
the Captive System Parameter. The JFR-rg model provides the theoretical mechanism that
the Mehrotra--Sergeyev empirical work implicitly requires.

\subsection{Methodological Accountability: Identification Challenges in
Standard Frameworks and the JFR-rg Response}
\label{sec:identification_method}

\subsubsection*{\textit{Prefatory Note}}

A methodological question that arises in evaluating any macroeconomic
framework is the extent to which its conclusions depend on unobservable
or imprecisely estimated parameters. This section addresses that question
directly for the JFR-rg framework---not defensively, but by documenting
the specific identification challenges that have remained unresolved in
the application of standard frameworks to Japan, and by demonstrating,
item by item, how the JFR-rg approach either resolves or materially
narrows each one.

The comparison is not intended as a repudiation of standard structural
models, which have made genuine contributions to macroeconomic analysis.
It is intended as a methodological accounting: identification standards,
if they are to be applied as a basis for evaluating any framework, must
be applied consistently across frameworks. A model that acknowledges
its identification constraints explicitly and pre-specifies observable
falsification conditions satisfies a higher standard of empirical
discipline than one that conditions its conclusions on parameters that
cannot be measured and whose estimates are revised only after the fact.

\medskip
\noindent\textbf{Critical scope note.}
The identification challenges catalogued in Table~\ref{tab:identification}
and discussed in Sections~\ref{sec:identification_method}--\ref{subsec:outputgap_id}
apply \emph{exclusively to Layer~(L2)}---the empirical exercises that
corroborate the JFR-rg mechanism.
Within Layer~(L1), Proposition~1 and the IOER-arithmetic component of
Proposition~4 are derived directly from the consolidated government budget
identity and require \emph{no causal identification by construction}.
Propositions~2--3 and the growth-channel component of Proposition~4
combine this identity with a small set of empirically calibrated structural
assumptions (the demographic ceiling on $\gnomstar$; the piecewise
exchange-rate pass-through in equation~\eqref{eq:yen_growth}); these
calibration assumptions are reported with full sensitivity analyses in
Appendix~\ref{app:sensitivity} but are not themselves identification claims.
Applying causal identification standards to an accounting identity is a
category error equivalent to demanding an instrumental-variable proof of the
national income accounting identity.

Table~\ref{tab:identification} summarizes the comparison across eight
dimensions. Each row states the identification challenge, its empirical
consequence for Japan, and the JFR-rg response. Three resolution
statuses are used: \textit{Resolved} (the variable is replaced by a
directly observable substitute with no residual dependence on the
unobservable); \textit{Resolved within accounting layer} (the challenge
does not enter the accounting propositions, though it may affect the
empirical layer); and \textit{Materially narrowed} (the functional form
and threshold are defined and the challenge is isolated from the paper's
central results, though full estimation requires additional data). The
sections that follow provide the supporting argument for each row. The
key finding, visible directly from the \textit{Resolution} column, is
that the JFR-rg framework is more identified---not less---than the
standard alternatives it supplements, across every dimension
consequential for analyzing Japan's debt dynamics.

\begin{table}[p]
\centering
\caption{Identification Challenges in Standard Frameworks and JFR-rg
Responses}
\label{tab:identification}
\adjustbox{max width=\textwidth}{%
\begin{tiny}
\setlength{\tabcolsep}{2.5pt}
\renewcommand{\arraystretch}{1.15}
\begin{tabularx}{\textwidth}{>{\raggedright\arraybackslash}p{0.8cm}
  >{\raggedright\arraybackslash}p{2.2cm}
  >{\raggedright\arraybackslash}p{3.2cm}
  >{\raggedright\arraybackslash}p{3.2cm}
  >{\raggedright\arraybackslash}p{3.5cm}
  >{\raggedright\arraybackslash}X}
\toprule
\textbf{\#} &
\textbf{Analytical Dimension} &
\textbf{Standard Framework Treatment} &
\textbf{Empirical Consequence for Japan} &
\textbf{JFR-rg Approach} &
\textbf{Resolution} \\
\midrule

1 &
Natural rate $r^*$ &
Estimated via structural models (e.g.\ HLW); treated as identified for
policy analysis &
Estimation spread of 200--400\,bp renders sustainability thresholds
indeterminate &
$\varepsilon_t \equiv \pi_t - r^n_t$ identified directly from FRED;
$r^*$ not required at any point in the derivation &
\textbf{Resolved} \\[2pt]

2 &
Domestic ownership structure &
Bond yields assumed to reflect market-clearing risk premia &
Captive system ($\varphi_t \approx 0.90$) structurally suppresses risk
premium; market-clearing assumption overpredicts $\rho_t$ &
$\varphi_t$ endogenized as Captive System Parameter; SC1 defined as
empirically verifiable scope condition; $\bar{\varphi}$ falsifiable
(Table~\ref{tab:identification}) &
\textbf{Materially narrowed} \\[2pt]

3 &
Consolidated balance sheet &
Fiscal authority and central bank treated as separate entities &
IOER channel at QQE scale ($\approx$\yen500\,T reserves) generates
fiscal impact assigned zero weight in standard models &
Consolidated constraint derived explicitly (eq.~\eqref{eq:consolidated_bc}); IOER burden ratio
$r^R_t$ quantified and included in all scenarios &
\textbf{Resolved} \\[2pt]

4 &
Exchange rate pass-through &
Linear pass-through coefficient assumed throughout &
Nonlinear threshold effects post-2022 unmodeled; stability window
unidentified &
Nonlinear threshold model (eq.~\eqref{eq:yen_growth}); $\bar{e}$ defined analytically;
Normalization Trap independent of $(\bar{e}, \beta)$
(Appendix~E, Panel~D) &
\textbf{Materially narrowed; $\bar{e}$ requires further estimation} \\[2pt]

5 &
Fiscal multiplier &
Required to project revenue and growth responses to consolidation;
estimates range 0.5--2.5 for Japan &
Sustainability assessments non-robust to multiplier choice &
Debt accumulation governed by accounting identity (L1); no behavioral
multiplier assumed; IOER arithmetic requires only mechanical
pass-through rate $\alpha_\text{pt}$ &
\textbf{Resolved within accounting layer} \\[2pt]

6 &
Policy reversibility &
Normalization assumed symmetrically reversible; standard loss functions
embed this assumption &
Normalization costs systematically underestimated; legacy debt
accumulation abstracted away &
Normalization Ratchet (Prop.~\ref{prop:ratchet}) derived analytically; 86-year
half-life established as mathematical consequence of recursion, not
calibration artifact &
\textbf{Resolved} \\[2pt]

7 &
Forecast accountability &
Failure conditions typically model-dependent; parameters revised ex-post
when predictions are not met &
Repeated one-directional forecast errors 2013--2026 without
methodological revision &
Falsification conditions pre-specified in observable terms (Table~\ref{tab:identification},
col.~5) before data examination; three conditions constitute
direct evidence against the framework &
\textbf{Resolved at methodological level} \\[2pt]

8 &
Output gap / NAIRU &
Policy rules conditioned on unobservable slack; estimates revised
ex-post by original authors &
Repeated policy missteps in Japan traceable to ex-ante overestimation
of inflationary pressure &
Framework conditions entirely on real-time observables ($\pi_t$,
$\varepsilon_t$, $g^n_t$, $b_{t-1}$, $d_t$); output gap not required
at any point &
\textbf{Resolved} \\

\bottomrule
\end{tabularx}
\end{tiny}
}
\end{table}

The natural rate of interest $r^*$ is central to the standard New
Keynesian policy framework. It appears in the Taylor rule, in
DSGE-based debt sustainability analyses, and in successive IMF Article
IV assessments of Japan's fiscal position (e.g., IMF, 2024a). Its estimation requires a
structural model, and the choice of model generates estimates that
differ by 200--400 basis points across leading specifications
(Holston, Laubach, and Williams, 2017). For Japan specifically, Holston, Laubach, and Williams (2017)
estimate $r^*$ at approximately 0.0--0.5\%, while a standard
loanable-funds calibration implies values near 4.0\,\%---a difference
not attributable to measurement noise but to the structural assumptions
embedded in the estimation model itself.

The practical consequence is that sustainability thresholds derived
from $r^*$-based frameworks are sensitive to specification choices that
are not resolved by the data. A framework whose sustainability
threshold shifts by several hundred basis points depending on which
structural model is used to estimate an unobservable has not
identified the threshold; it has produced a range wide enough to
accommodate conclusions determined largely by prior assumptions.

The JFR-rg framework eliminates this dependency by replacing $r^*$
with the directly observable financial repression bias
$\varepsilon_t \equiv \pi_t - r^n_t$, constructed from two FRED
series without model-based interpolation. The critical stability
threshold $\varepsilon^* = 0.533\%$ (Appendix~E, Panel~B) is derived
entirely from observable quantities. No natural rate estimate is
required or invoked at any point in the derivation, in the
simulations, or in the falsification conditions.

\subsubsection{The Domestic Ownership Structure and the Endogenous
Risk Premium}
\label{subsec:phi_id}

Standard open-economy models assume that government bond yields are
market-clearing prices set by investors who can freely exit. In Japan,
where approximately 88--92\,\% of JGBs are held by domestic
institutions operating under regulatory mandates and prudential
requirements, this assumption is empirically at odds with the
institutional structure. Hoshi and Ito (2014) identify the domestic
holding share as the central variable governing Japan's continued
ability to borrow at low rates, noting that its decline would
constitute the proximate trigger for a sudden stop. Standard
frameworks that abstract from this structure generate predictions of
sovereign risk premium emergence that have not materialized over the
2013--2026 period.

The JFR-rg framework closes this gap by endogenizing $\varphi_t$ as
the Captive Financial System Parameter (Section~\ref{sec:captive}),
defining an illustrative early-warning benchmark for the captive-system
condition, and specifying a directional falsification pattern: a sustained
decline in $\varphi_t$ toward and below the 0.85 benchmark without
measurable risk premium emergence would count against the framework's
maintained claim (Table~\ref{tab:identification}, row~4). The
identification problem is not fully resolved---estimation of the exact
tipping value for $\bar{\varphi}$ from historical sudden-stop data remains
a research priority---but the variable is defined, measured quarterly
from publicly available BoJ Flow of Funds accounts, and linked to an
observable monitoring benchmark, advancing materially beyond the standard
treatment of ignoring the ownership structure entirely.

\subsubsection{The Consolidated Balance Sheet and the IOER Channel}
\label{subsec:ioer_id}

The standard formulation of the government budget constraint treats the
fiscal authority and the central bank as separate entities. Prior to
large-scale asset purchase programs, this separation was
quantitatively negligible. At the current scale of QQE---with BoJ
current account deposits of approximately \yen454--500\,trillion---it
is not. A 1.5\,\% increase in the policy rate applied to \yen500
trillion in reserves generates \yen7.5 trillion in annual IOER costs
(approximately 1.12\,\% of GDP) that the standard unconsolidated
constraint assigns zero weight.

The theoretical problem of fiscal dominance has been recognized in the
literature since Sargent and Wallace (1981) and its price-level implications
formalized since Leeper (1991). What standard frameworks have
not done is apply the consolidated arithmetic to Japan's specific
balance sheet at its current scale. The JFR-rg framework does so
explicitly in equation~\eqref{eq:consolidated_bc}, deriving the IOER burden ratio $r^R_t$ as
a structurally necessary component of the debt accumulation equation.
The Scenario~C primary deficit expansion of $+1.39\%$ of GDP
(Section~\ref{sec:scenC}) is the direct arithmetic consequence of
the consolidated balance sheet applied to publicly available data, not
a behavioral assumption about fiscal deterioration. The identification
problem is resolved within the accounting layer (L1) without
behavioral assumptions.

\subsubsection{Linear Pass-Through and the Nonlinear Exchange Rate}
\label{subsec:exchange_id}

Standard open-economy frameworks treat exchange rate pass-through as
approximately linear. The post-2022 Japanese data are not consistent
with this treatment: as the USD/JPY rate exceeded 140, the correlation
between exchange rate movements and domestic CPI increased materially,
suggesting that import inflation was beginning to dominate the
corporate profitability channel at large depreciation magnitudes.

The JFR-rg model introduces the stability threshold $\bar{e}$
explicitly in equation~\eqref{eq:yen_growth}, distinguishing the linear growth-enhancing
channel ($\alpha \Delta e_t$ for $\Delta e_t \leq \bar{e}$) from the
quadratic import-inflation penalty for excessive depreciation. The
identification problem is materially narrowed: the functional form is
specified, the threshold is defined, and Appendix~E, Panel~D
demonstrates analytically that the Normalization Trap result is
structurally independent of $(\bar{e}, \beta)$, because monetary
tightening induces yen appreciation ($\Delta e_t < 0$), which never
activates the penalty term regardless of the threshold value. The
parameters $\bar{e}$ and $\beta$ affect only depreciation-scenario
stability analysis; they do not enter the paper's central
normalization-risk results.

\subsubsection{Fiscal Multiplier Indeterminacy and the Budget Identity}
\label{subsec:multiplier_id}

Fiscal sustainability assessments in the standard tradition require an
estimate of the fiscal multiplier to project the revenue and growth
consequences of consolidation. For Japan, estimates range from
approximately 0.5 to 2.5 across empirical studies, with the range
widening further when zero-lower-bound dynamics and balance sheet
recession conditions are taken into account (Koo, 2003). A
sustainability assessment whose conclusions are sensitive to the choice
of multiplier within a range of 2.0 percentage points is not
identified in any operationally useful sense.

The JFR-rg framework's accounting layer (L1) requires no fiscal
multiplier assumption. The debt accumulation equation (eq.~\eqref{eq:standard_dynamics}) is
derived from the consolidated government budget identity, which holds
as an accounting relationship for any realized values of $r^n_t$,
$g^n_t$, $b_{t-1}$, and $d_t$. The behavioral question---how $d_t$
responds to changes in macroeconomic conditions---enters only through
the mechanical IOER arithmetic of equation~\eqref{eq:consolidated_bc}, and the sensitivity of
the Normalization Trap to this channel is documented across the full
pass-through range $\alpha_\text{pt} \in [0,1]$ in
Section~\ref{sec:ioer_robustness}. The Normalization Trap holds at
every value of $\alpha_\text{pt}$, including the zero case (Table~\ref{tab:scenC}).

\subsubsection{The Assumption of Policy Reversibility}
\label{subsec:reversibility_id}

Standard macroeconomic frameworks embed the assumption that monetary
policy errors are reversible: if a rate hike proves premature, rates
can be cut and the economy returns to its pre-hike trajectory. This
assumption is embedded in the symmetric loss functions of New Keynesian
models and in the treatment of normalization as a reversible choice
among interchangeable paths.

The JFR-rg framework demonstrates that this assumption fails
analytically for high-debt economies through the Normalization Ratchet
(Proposition~\ref{prop:ratchet}). The proof establishes that a $T$-period normalization
shock generates a debt increment whose subsequent decay rate is
$(1 + r^n_0 - g^n_0)$ per period. At Scenario~A parameters ($r^n -
g^n = -0.8\%$), the gap half-life is approximately 86 years---a
mathematical consequence of the debt accumulation recursion, not a
simulation result. The assumption of policy reversibility is not
merely empirically questionable in Japan's context; it is analytically
incorrect for the class of economies defined by
Proposition~\ref{prop:ratchet}'s maintained conditions.

\subsubsection{Forecast Accountability and Pre-Specified Falsification}
\label{subsec:falsifiability_id}

A framework's empirical discipline can be evaluated not only by its
in-sample fit but by the specificity of its predictions and the
observability of its failure conditions. A framework that revises its
parameters ex-post when predictions are not met, without acknowledging
the revision as methodologically significant, provides limited
empirical traction regardless of its theoretical elegance.

The JFR-rg framework addresses this by pre-specifying its failure
conditions in observable terms before examining the data (Table~\ref{tab:identification},
col.~5). Three conditions constitute direct evidence against the
framework: (i)~sustained $\varphi_t$ decline below $\bar{\varphi}$
without sovereign risk premium emergence; (ii)~persistent positive
$r^n_t - g^n_t$ without the debt accumulation predicted by
Proposition~\ref{prop:norm_trap}; (iii)~sustained $\varepsilon_t < 0$
without the debt acceleration implied by the Corridor analysis. All three conditions
are expressible in real-time FRED data and BoJ Flow of Funds accounts.
Pre-specification of observable failure conditions is a minimum
standard of empirical accountability that applies to any framework
claiming relevance for policy.

\subsubsection{Output Gap Dependence and the Real-Time Observability
Constraint}
\label{subsec:outputgap_id}

Japan's monetary policy has been conditioned repeatedly on output gap
and NAIRU estimates that were subsequently revised in the opposite
direction. The Bank of Japan's exits from accommodative policy---in
2000, in 2006--2007, and following the 2014 consumption tax
episode---were each justified in part by estimates of diminishing slack
that proved, on subsequent revision, to have been overstated. The
identification problem is not that output gaps are difficult to measure
in real time; it is that policy frameworks conditioned on
non-observable variables cannot be falsified when predictions fail,
because the underlying unobservable can be revised to accommodate any
outcome.

The JFR-rg framework does not condition on the output gap or NAIRU at
any point. The stability condition (eq.~\eqref{eq:stability_condition}) is expressed entirely in
terms of directly observable series: $\pi_t$, $\varepsilon_t$,
$g^n_t$, $b_{t-1}$, and $d_t$. The structural potential growth rate
$g^{n*}_t$ enters only in the exchange-rate block, and its sensitivity
is documented in Panel~C of Appendix~E: a $0.5$\,pp decline in
$g^{n*}_t$ requires a compensating $0.5$\,pp increase in $\varepsilon_t$
to maintain the debt trajectory, providing a directly observable
policy trade-off without invoking unobservable slack.

\medskip\noindent
The eight identification dimensions documented above share a common
structure: each involves a variable or relationship that standard
frameworks either omit, treat as unobservable, or resolve through
model-dependent estimation too imprecise to support the policy
conclusions drawn from it. Taken together, they define the
methodological position this paper occupies: the JFR-rg framework does
not claim to resolve every identification challenge in macroeconomic
analysis, but it resolves or materially narrows the specific challenges
most consequential for Japan's debt dynamics, documents its remaining
limitations with observable falsification conditions, and conditions
its maintained assumptions on real-time public data throughout. That
is the basis on which the paper's empirical claims are presented and
on which they should be evaluated.

\medskip
\noindent\textbf{Consensus contribution.}
Whatever the reader's preferred interpretive lens, the paper's consensus
contribution is this: Japan's debt-to-GDP ratio of 240\% behaves differently
from what standard scalar debt thresholds alone would suggest, and the
difference is systematically traceable to three institutional observables
---$\varepsilon_t$, $\varphi_t$, and $\Delta e_t$---whose interaction the
JFR-rg framework makes explicit. Whether this is best read as a fiscal
accounting exercise, an institutional $r<g$ extension, a nonlinear exchange-
rate framework, a normalization-risk framework, or a falsifiable mechanism
hypothesis, the underlying point is the same: at $b_{t-1}=2.40$, the sign and
magnitude of the $r^n-g^n$ spread are fiscal variables of the first order, and
the institutional conditions that shape that spread are observable in real
time. That is the map this paper offers.

\medskip
This methodological discipline also clarifies the interpretive scope of the
framework. A crucial question therefore remains: how should a
regime-conditional, observables-centered framework such as JFR-rg be situated
relative to the established theoretical traditions that have long structured
debate on Japan's macroeconomy? The purpose of the following subsection is not
to repudiate those traditions, but to specify more precisely which parts of the
Japanese experience they explain well, and which post-2013 residual the JFR-rg
framework is designed to organize. Put differently, the issue is not whether
alternative theories can explain Japan, since many clearly can, but whether
they fully capture the particular transition-era interaction among debt
arithmetic, regime-specific financial absorption, and a bounded exchange-rate
channel that became especially salient after 2013.

\subsection{Interpretive Boundaries and Complementarity: Lucas, Minsky, and
the Low-\texorpdfstring{$r^*$}{r*} Mainstream View}
\label{sec:interpretive_boundaries}

The JFR-rg framework is best understood not as a universal replacement for
existing macroeconomic theories of Japan, but as a regime-conditional
supplement designed to organize a specific post-2013 configuration. This
boundary is substantive rather than rhetorical. The framework does not deny
that other theories can explain major elements of Japan's long stagnation;
rather, it argues that they do not yet fully organize, within a single
debt-dynamics language, the post-2013 interaction among a very large debt
base, a sustained negative interest-growth differential, a captive domestic
financial system, and a bounded exchange-rate channel. In this sense, JFR-rg
is complementary in two directions: first, to mainstream debt-sustainability
analysis, which Appendix~J shows to be a limiting case obtained when the
captive-system and exchange-rate-specific conditions are removed; and second,
to broader macroeconomic narratives that illuminate important parts of Japan's
experience but do not themselves specify the observable institutional
thresholds under which debt compression persists or fails.

A Lucasian interpretation can explain much of Japan's experience without
recourse to the JFR-rg mechanism itself. On that view, Japan is a
low-natural-rate economy in which demographic decline, weak productivity
growth, and persistently subdued inflation expectations jointly depress
equilibrium real rates, while changes in the policy regime alter private
expectations and thereby caution against any simple extrapolation from observed
post-2013 correlations. This perspective is analytically powerful, especially
as a warning against treating transition-era regularities as policy-invariant
structural laws. Yet precisely because the Lucasian framework is centered on
expectations, policy rules, and structural invariance, it does not by itself
yield a practical set of directly observable failure conditions for Japan's
debt-compression regime. The contribution of JFR-rg at this margin is
narrower but operationally important: it translates the Japanese post-2013
configuration into observable monitors---$\varepsilon_t$, $\varphi_t$, and
$\Delta e_t$---and specifies ex ante the institutional thresholds at which the
apparent stability of debt dynamics should be expected to weaken or fail.

A Minskyan interpretation also explains a central part of the Japanese case.
From that perspective, Japan is best understood as a prolonged post-bubble
balance-sheet adjustment in which private-sector deleveraging was offset by
the balance sheets of the state and the central bank. The apparent stability
of the sovereign therefore does not imply that fragility disappeared; rather,
fragility may have been displaced, absorbed, and partially socialized through
public institutions acting as a stabilizing backstop. This is a deep and
highly relevant reading of Japan's post-1990 trajectory. However, while it
illuminates where instability may have migrated, it does not itself formalize
the post-2013 sovereign debt-compression mechanism in terms of the jointly
operative conditions emphasized here: a directly observable repression bias, a
large debt base that magnifies the arithmetic effect of a negative $r-g$
spread, and a nonlinear exchange-rate window that supports nominal growth up
to a threshold. JFR-rg should therefore be read not as a rejection of the
Minskyan narrative, but as a narrower accounting-and-regime framework
designed to formalize the sovereign-dynamics layer of that broader story.

Among current mainstream approaches, the interpretation that comes closest to
organizing Japan without the JFR-rg framework is the low-$r^*$ secular-
stagnation view embedded in open-economy New Keynesian analysis. This
approach explains Japan as an economy characterized by chronically weak
underlying growth, a depressed natural rate of interest, and stubbornly low
inflation expectations, such that the zero lower bound and unconventional
monetary policy become persistent features rather than temporary anomalies.
That framework remains indispensable, and JFR-rg does not seek to displace it.
Even in its strongest form, however, the low-$r^*$ mainstream view leaves a
residual analytical task that is specific to Japan's post-2013 transition: to
specify how debt compression can persist in real time under observable
institutional conditions, why normalization risk becomes path-dependent when
the debt base is very large, and how erosion of the captive financial system
would move the economy back toward the standard mainstream condition. That
residual task is the precise domain of the JFR-rg framework. Mainstream
analysis is the more natural framework for the destination; JFR-rg is
proposed as the more informative framework for the transition.

\section{Stylized Facts: Empirical Evidence from Real-Time FRED Data (2013--2026)}

\subsection{Nominal vs.\ Real GDP and the Inflation Gap
  (Figures~\ref{fig:gdp_early} and~\ref{fig:gdp_full})}

As shown in Figures~\ref{fig:gdp_early} and~\ref{fig:gdp_full}, nominal GDP (SAAR,
FRED:\,JPNNGDP)%
\footnote{We utilize the SAAR series from FRED rather than Cabinet Office
calendar-year aggregates (Cabinet Office, Government of Japan, 2025) because SAAR
captures high-frequency quarterly momentum and is the relevant denominator for
the JFR-rg stability condition (eq.~\eqref{eq:sim_transition}). The SAAR starting
value of $\approx$510 trillion yen in Q1~2013 differs from the Cabinet Office
2013 calendar-year aggregate of $\approx$503 trillion yen; see Appendix~\ref{app:saar}
for a detailed reconciliation. The $\approx$7 trillion yen level difference does not
affect any of the paper's qualitative or quantitative conclusions.}
expanded from approximately \textyen{}510 trillion in 2013 to \textyen{}670--680 trillion
by early 2026, with the post-2022 acceleration substantially driven by the inflationary
episode (CPI YoY: near-zero $\to$ $\approx$2.5--3.2\%).  This widening ``Inflation
Gap'' between nominal and real GDP is consistent with the JFR-rg debt compression
mechanism: the inflationary environment simultaneously expands the nominal denominator
and erodes the real value of the debt stock.

The lower panels display the USD/JPY exchange rate alongside CPI YoY. The full-sample
correlation is $\rho = 0.06$. However, a structural break is observable post-2022: as
the yen depreciated aggressively past 140 JPY/USD, the correlation between exchange
rate movements and domestic CPI increased materially. This dynamic is not an anomaly
but rather a confirmation of JFR-rg
Proposition~\ref{prop:yen_stabilizer} (The Nonlinear Yen Stabilizer):
as $\Delta e_t$ approaches the threshold $\ebar$, the benign corporate earnings channel
is progressively offset by broad import inflation, signaling the outer boundary of the
stability window.

\begin{figure}[htbp]
  \centering
  \includegraphics[width=0.88\textwidth]{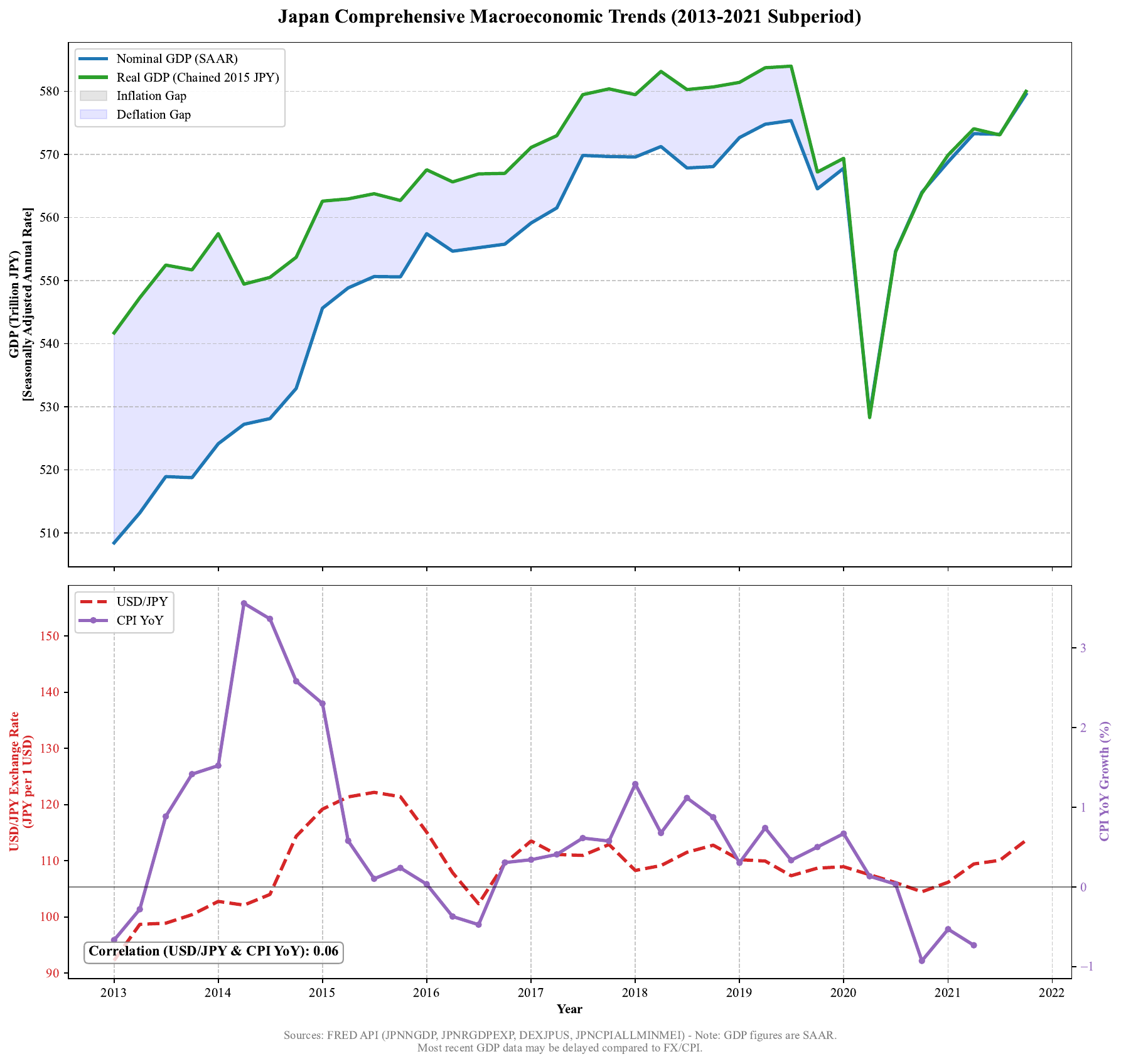}
  \caption{Japan Comprehensive Macroeconomic Trends (2013--2021 Subperiod).
    \textit{Upper panel:} Nominal GDP SAAR (FRED:\,JPNNGDP) vs.\ Real GDP Chained
    2015 JPY (FRED:\,JPNRGDPEXP) with inflation gap shading.
    \textit{Lower panel:} USD/JPY (FRED:\,DEXJPUS) vs.\ CPI YoY
    (FRED:\,JPNCPIALLMINMEI). Full-sample Pearson correlation $\rho = 0.06$.
    This subperiod panel is presented alongside Figure~\ref{fig:gdp_full}
    to provide a pre-2022 baseline against which the post-YCC structural shift
    in the USD/JPY--CPI co-movement can be assessed.
    Data source: FRED API, Federal Reserve Bank of St.\ Louis.}
  \label{fig:gdp_early}
\end{figure}

\begin{figure}[htbp]
  \centering
  \includegraphics[width=0.88\textwidth]{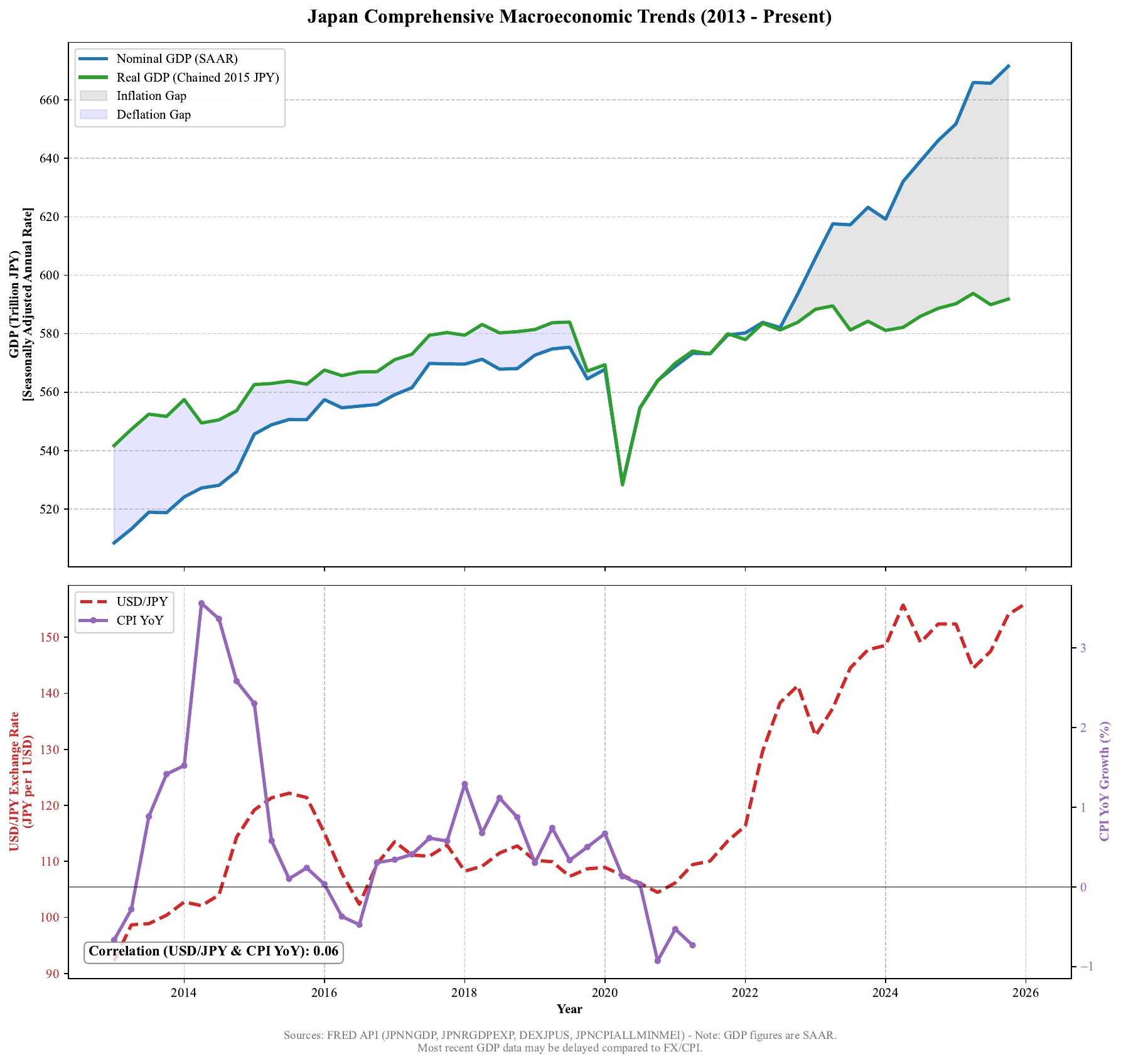}
  \caption{Japan Comprehensive Macroeconomic Trends (2013--Present, Full Panel).
    Upper panel: Nominal GDP SAAR vs.\ Real GDP. Lower panel: USD/JPY vs.\ CPI YoY.
    The widening ``Inflation Gap'' post-2022 reflects the inflationary episode that
    supercharges the JFR-rg debt compression mechanism. The rising co-movement of
    USD/JPY and CPI post-2022 is consistent with the approach of the nonlinear stability
    threshold $\ebar$.
    Data source: FRED API (JPNNGDP, JPNRGDPEXP, DEXJPUS, JPNCPIALLMINMEI),
    Federal Reserve Bank of St.\ Louis.}
  \label{fig:gdp_full}
\end{figure}

\subsection{Debt-to-GDP Trajectory and the Base Effect
  (Figures~\ref{fig:trilemma_early} and~\ref{fig:trilemma_full})}

{\sloppy\raggedright
The Macroeconomic Trilemma Dashboard provides the central empirical
evidence for the JFR-rg base effect mechanism. The government
gross debt-to-GDP ratio
(FRED:\,GGGDTAJPA188N, sourced from the IMF's World Economic Outlook
database)%
\footnote{GGGDTAJPA188N measures \emph{general government gross debt} as defined by
the IMF, which includes central government debt, local government debt, and social
security fund liabilities. This differs from Japan's Ministry of Finance ``JGB
outstanding'' series (a narrower measure) and from net debt (gross debt minus
financial assets). The ``240\%'' figure cited throughout this paper refers to gross
debt and is the standard metric used in IMF Article~IV consultations and OECD
Fiscal Outlook reports for cross-country comparability. On a net debt basis, Japan's
indebtedness is approximately 120--140\% of GDP, reflecting large public financial
assets.}
peaked at approximately 260--264\% in 2020--2021 in response to COVID-19
fiscal expansion. Despite the \emph{absence} of significant fiscal
consolidation, the ratio subsequently declined to approximately
236--240\% by early 2026---a compression of over 20 percentage points
in five years, consistent with stabilization through the JFR-rg
channels, though not uniquely attributable to them.
With $\bt = 2.40$ and $r^n - g^n \approx -0.8\%$, the annual automatic
debt compression equals $0.008 \times 2.40 = 1.92\%$ of GDP---absorbing
the vast majority of the structural primary deficit of approximately
2.0\% of GDP. This is the Base Effect Lever in operation.\par}

\begin{figure}[htbp]
  \centering
  \includegraphics[width=0.88\textwidth]{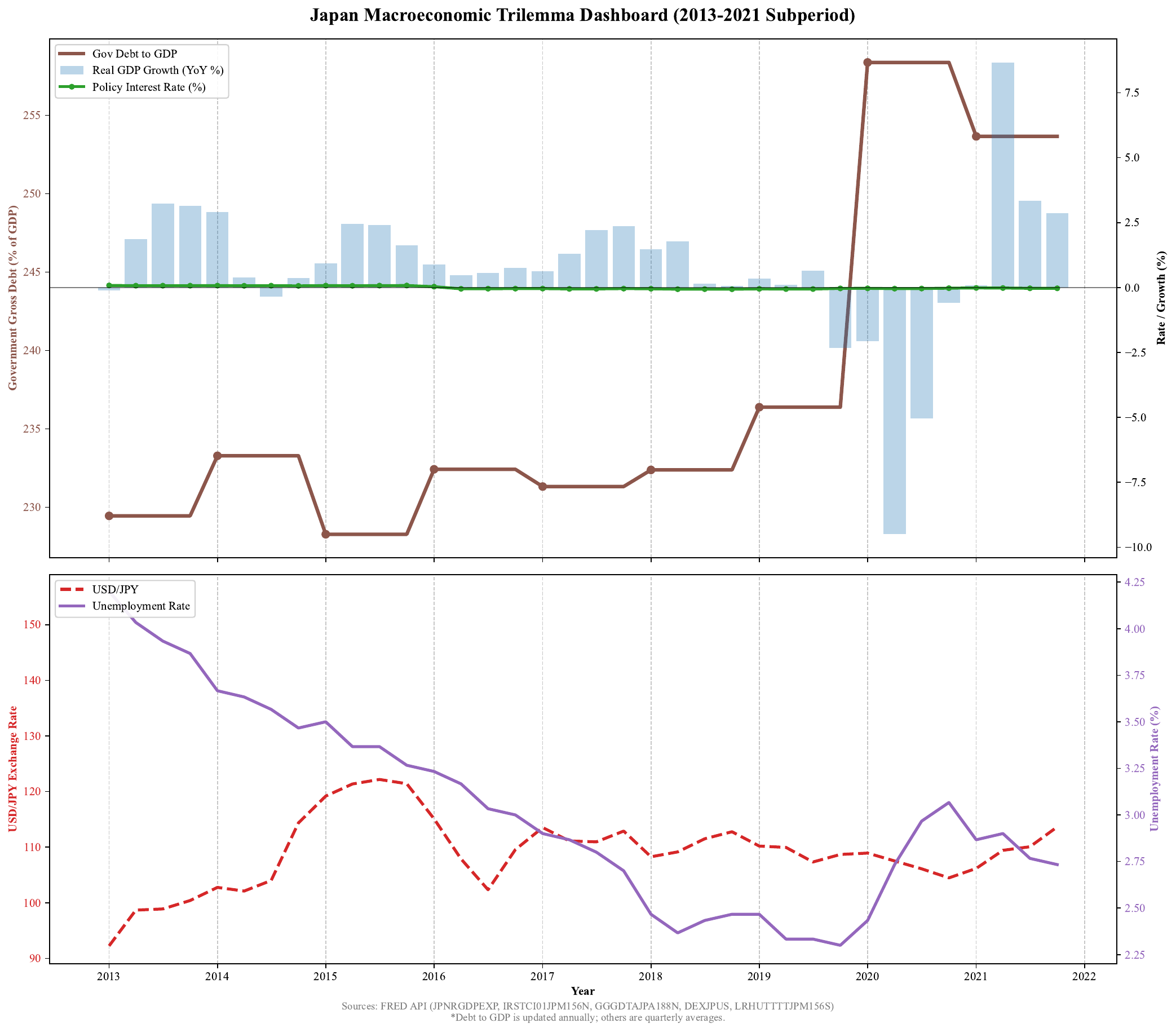}
  \caption{Japan Macroeconomic Trilemma Dashboard (2013--2021 Subperiod).
    \textit{Upper panel:} Government Gross Debt/GDP (\%, FRED:\,GGGDTAJPA188N,
    IMF definition; see footnote~2),
    Real GDP Growth YoY (\%, FRED:\,JPNRGDPEXP),
    Policy Interest Rate (\%, FRED:\,IRSTCI01JPM156N).
    \textit{Lower panel:} USD/JPY (FRED:\,DEXJPUS) and Unemployment Rate
    (\%, FRED:\,LRHUTTTTJPM156S).
    Data source: FRED API, Federal Reserve Bank of St.\ Louis.}
  \label{fig:trilemma_early}
\end{figure}

\begin{figure}[htbp]
  \centering
  \includegraphics[width=0.88\textwidth]{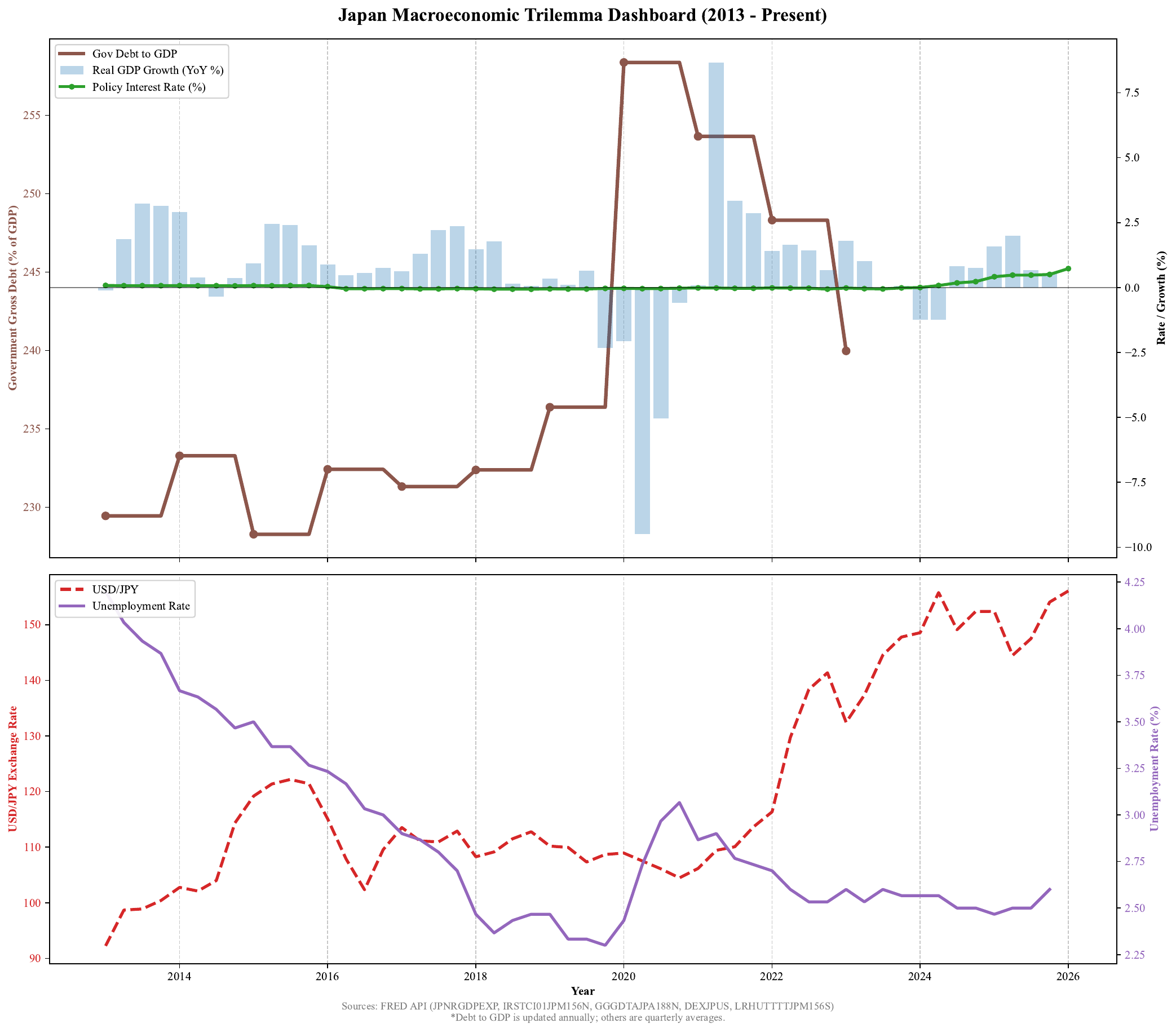}
  \caption{Japan Macroeconomic Trilemma Dashboard (2013--Present, Full Panel).
    The upper panel shows the gross debt-to-GDP ratio declining from its 2020--2021
    peak of $\approx$260--264\% back to $\approx$236--240\% by early 2026
    \emph{without} fiscal austerity---consistent with the operation of the JFR-rg Base
    Effect mechanism, though alternative explanations (including cyclical base effects
    and nominal GDP recovery) also contributed. The policy rate (green dotted line)
    began its cautious normalization from 2024.
    Data source: FRED API (GGGDTAJPA188N, JPNRGDPEXP, IRSTCI01JPM156N, DEXJPUS,
    LRHUTTTTJPM156S), Federal Reserve Bank of St.\ Louis.}
  \label{fig:trilemma_full}
\end{figure}

\subsection{The \texorpdfstring{$r$-$g$}{r-g} Spread: Defining the Safe Zone
  (Figure~\ref{fig:rg_spread})}

Figure~\ref{fig:rg_spread} tracks the real-time spread between the nominal interest
rate $r^n$ (10-year JGB yield, FRED:\,IRLTLT01JPM156N) and the nominal GDP growth rate
$g^n$ (YoY), spanning approximately 2000 to 2026.  The historical record
since approximately 2013 shows sustained green-zone operation, with the $r^n - g^n$
spread averaging approximately $-1.0$ to $-2.0\%$. The brief red-zone episodes
corresponding to the 2008--2009 GFC and Q1--Q2 2020 pandemic shock confirm the
model's prediction: it is exogenous demand collapses (negative $g^n$ shocks) that
trigger debt instability, not the debt level itself. As of early 2026, the spread
remains in the green zone at approximately $-3.0$ to $-0.8\%$.

\begin{figure}[htbp]
  \centering
  \includegraphics[width=0.88\textwidth]{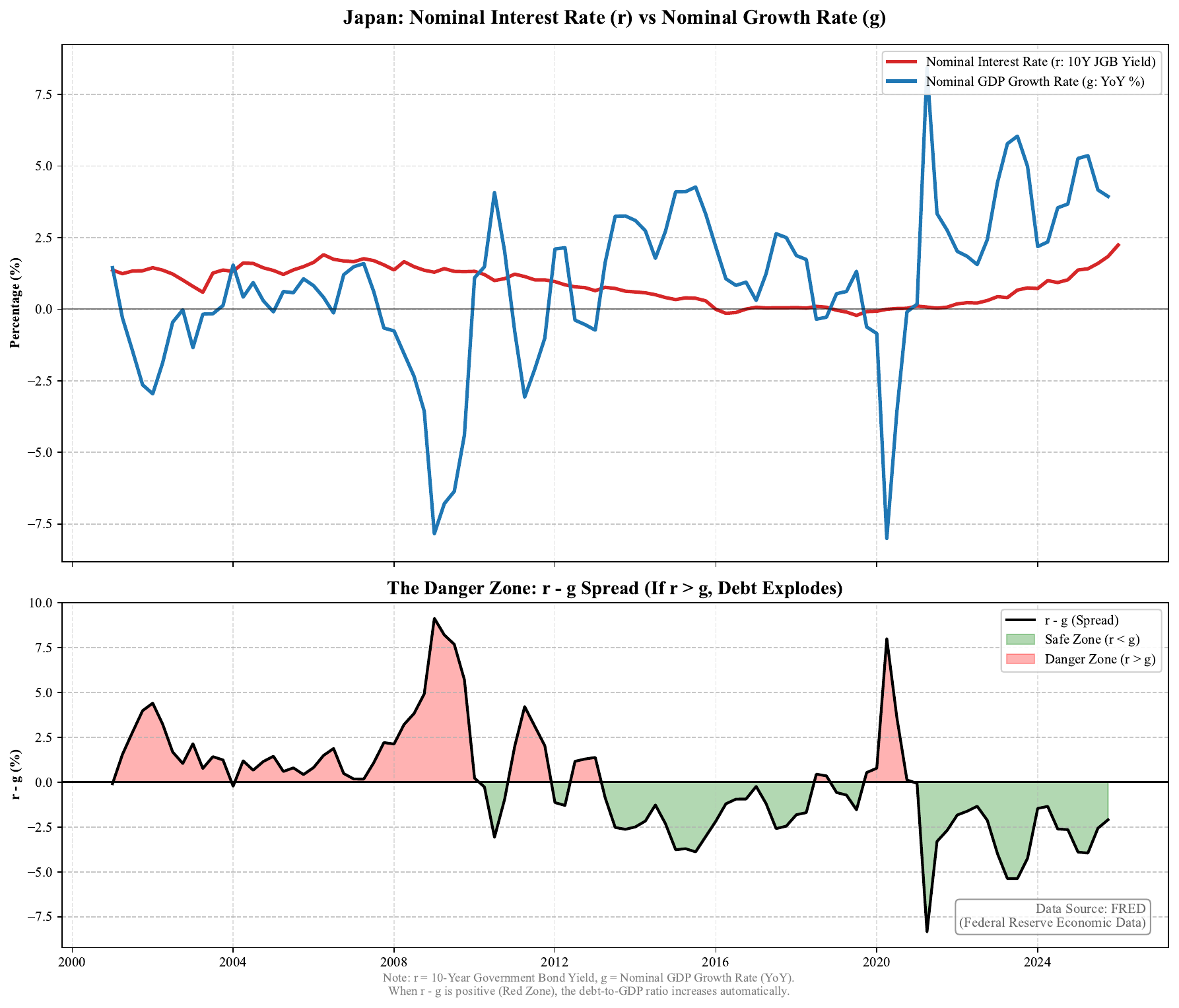}
  \caption{Japan: Nominal Interest Rate ($r^n$) vs.\ Nominal Growth Rate ($g^n$),
    2000--Present.
    \textit{Upper panel:} 10-year JGB Yield ($r^n$, FRED:\,IRLTLT01JPM156N) and
    Nominal GDP Growth YoY ($g^n$, derived from FRED:\,JPNNGDP).
    \textit{Lower panel:} $r^n - g^n$ Spread with Safe Zone ($r^n < g^n$, green) and
    Danger Zone ($r^n > g^n$, red) shading.
    Data source: FRED API, Federal Reserve Bank of St.\ Louis.}
  \label{fig:rg_spread}
\end{figure}

\subsection{The Financial Repression Monitor
  (Figures~\ref{fig:repression_full} and~\ref{fig:repression_monthly})}

Figures~\ref{fig:repression_full} and~\ref{fig:repression_monthly} provide the direct
empirical identification of the financial repression bias $\varepsilon_t \equiv
\pi_t - r^n_t$ (formally defined in Section~\ref{sec:repression_bias}).

Figure~\ref{fig:repression_monthly} covers 2013--2021 with monthly granularity.%
\footnote{Figure~\ref{fig:repression_full} uses annual CPI data extrapolated to
monthly frequency (as noted in the figure caption) because FRED discontinued monthly
Japanese CPI publication. Figure~\ref{fig:repression_monthly} uses the raw monthly
series available through 2021. The two figures are presented together to provide both
long-run context and granular high-frequency validation.}
The 2013--2015 period generated the deepest negative real rate ($\varepsilon_t \approx
2.5$--$3.0\%$), driven by the Abenomics inflation surprise. Figure~\ref{fig:repression_full}
shows that this dynamic reasserted powerfully after 2022, as CPI YoY surged to
2.5--3.2\% while the Bank of Japan maintained the 10-year yield near zero under YCC
through mid-2024.

The empirical identification at the March 2026 data point depends critically on the
inflation measure chosen, and we present \textbf{two parallel calibrations}:

\medskip
\noindent\textbf{Baseline calibration ($\pi_t = 2.7\%$, BoJ projection):}
\[
  \varepsilon_t \equiv \pi_t - r^n_t \approx 2.7\% - 2.2\% = +0.5\%.
\]
This value reflects the Bank of Japan's fiscal-year 2025 CPI central forecast
(October 2025 \textit{Outlook for Economic Activity and Prices}) and the
2025 calendar-year realized average. Under this calibration, Japan remains
marginally in the financial repression zone ($\varepsilon_t > 0$), and the
March 2026 operating point lies near---but just outside---the Stability
Corridor boundary (Corridor Width $\approx -0.024\%$; see Section~\ref{sec:corridor}).
All simulations in Section~4 use this calibration.

\medskip
\noindent\textbf{Alternative calibration ($\pi_t = 1.5\%$, January 2026 actual):}
\[
  \varepsilon_t \equiv \pi_t - r^n_t \approx 1.5\% - 2.2\% = -0.7\%.
\]
The January 2026 national CPI (Statistics Bureau of Japan, released 20~February
2026) recorded year-on-year headline inflation of 1.5\%---the lowest reading since
March 2022---driven by the expiry of subsidy-reversal base effects and slowing
food-price inflation. Under this calibration, Japan sits marginally in the
\emph{positive real-rate zone} ($\varepsilon_t < 0$) as of January 2026.

\medskip
\noindent\textbf{Consistency of the corridor distance across calibrations.}
Because $\gnomstar$ is a structural parameter---determined by real potential
growth $g^*_t$ and trend inflation rather than by the realized month-to-month
$\pi_t$---the two calibrations imply \emph{the same} signed distance from the
Stability Corridor boundary under the paper's maintained calibration
($\gnomstar \approx 3.0\%$). This can be verified directly from the corridor
width formula of Property~\ref{prop:corridor_width}: the numerator
$(\varepsilon_t + \gnomstar) - (\pi_t + (d_t - s_t)/b_{t-1})$ equals
$r^n_t - g^{n*}_t - (d_t - s_t)/b_{t-1}$ (because $\varepsilon_t + \gnomstar
- \pi_t = -(r^n_t - g^{n*}_t)$, up to the convention
$g^n_t = g^{n*}_t$ under $\Delta e_t = 0$), and therefore depends on
$\pi_t$ and $\varepsilon_t$ only through their difference $r^n_t = \pi_t -
\varepsilon_t$, which is invariant across the two calibrations ($r^n_t = 2.2\%$
in both). Substituting: $W = [2.2\% - 3.0\% - 2.0\%/2.40]/\sqrt{2}
\approx -0.024\%$ under both calibrations. What differs between calibrations
is therefore not the quantitative distance from the corridor boundary, but the
\emph{institutional interpretation} of the repression channel---active under
$\varepsilon_t > 0$, and temporarily inactive under $\varepsilon_t < 0$---and
the forward-looking question of whether the January 2026 CPI dip proves
transitory (as the BoJ's FY2025 projection implies) or persistent.

\medskip
\noindent The qualitative conclusions of the paper---the Normalization Trap
(Proposition~\ref{prop:norm_trap}),
the Ratchet (Proposition~\ref{prop:ratchet}), and the policy directional
recommendations---are
robust to both calibrations. Readers should interpret the baseline
simulations as conditional on the BoJ projection materializing. Where the two
calibrations diverge in the forward-looking trajectory of $\varepsilon_t$,
the sensitivity is documented in Appendix~E (Panel~B).

\begin{figure}[htbp]
  \centering
  \includegraphics[width=0.88\textwidth]{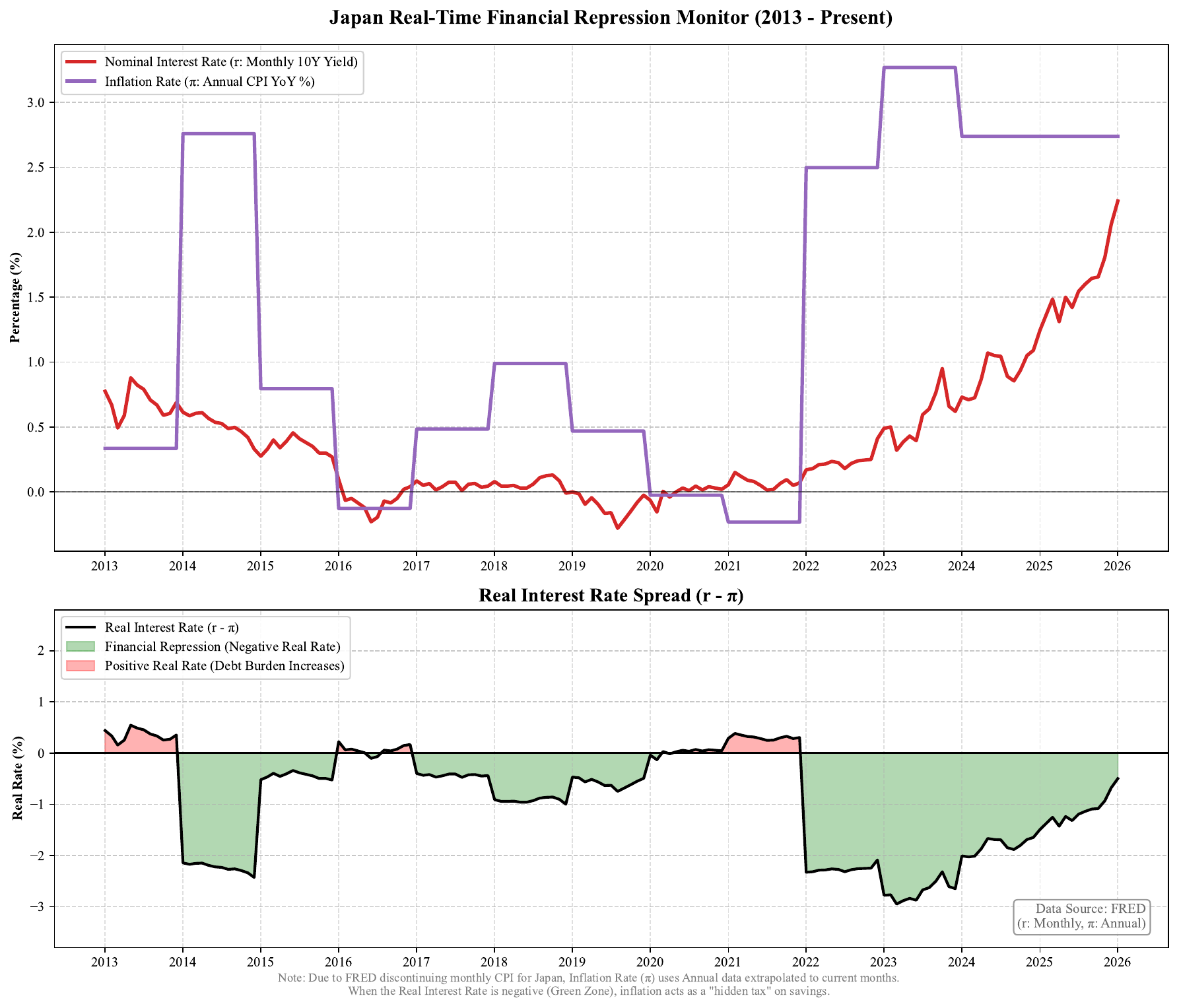}
  \caption{Japan Real-Time Financial Repression Monitor (2013--Present, Full Panel).
    \textit{Upper panel:} Nominal 10Y JGB Yield ($r^n$, FRED:\,IRLTLT01JPM156N) and
    Inflation Rate ($\pi$, Annual CPI YoY, FRED:\,FPCPITOTLZGJPN).
    \textit{Lower panel:} Financial Repression Bias
    $\varepsilon_t \equiv \pi_t - r^n_t$
    with Repression Zone ($\varepsilon_t > 0$, green) and Positive Real Rate Zone
    ($\varepsilon_t < 0$, red) shading.
    Note: CPI data are annual; monthly values use forward-filled annual data due to
    FRED data availability. The deep green zone post-2022 reflects the largest
    observed $\varepsilon_t$ in the sample. Forward-filling introduces a step-function
    approximation that may delay the apparent timing of sign changes in $\varepsilon_t$
    by up to 12 months relative to the true monthly series; comparison with
    Figure~\ref{fig:repression_monthly} (monthly data, 2013--2021) confirms that
    this artifact does not alter the directional identification of repression episodes.
    Data source: FRED API (IRLTLT01JPM156N, FPCPITOTLZGJPN),
    Federal Reserve Bank of St.\ Louis.}
  \label{fig:repression_full}
\end{figure}

\begin{figure}[htbp]
  \centering
  \includegraphics[width=0.88\textwidth]{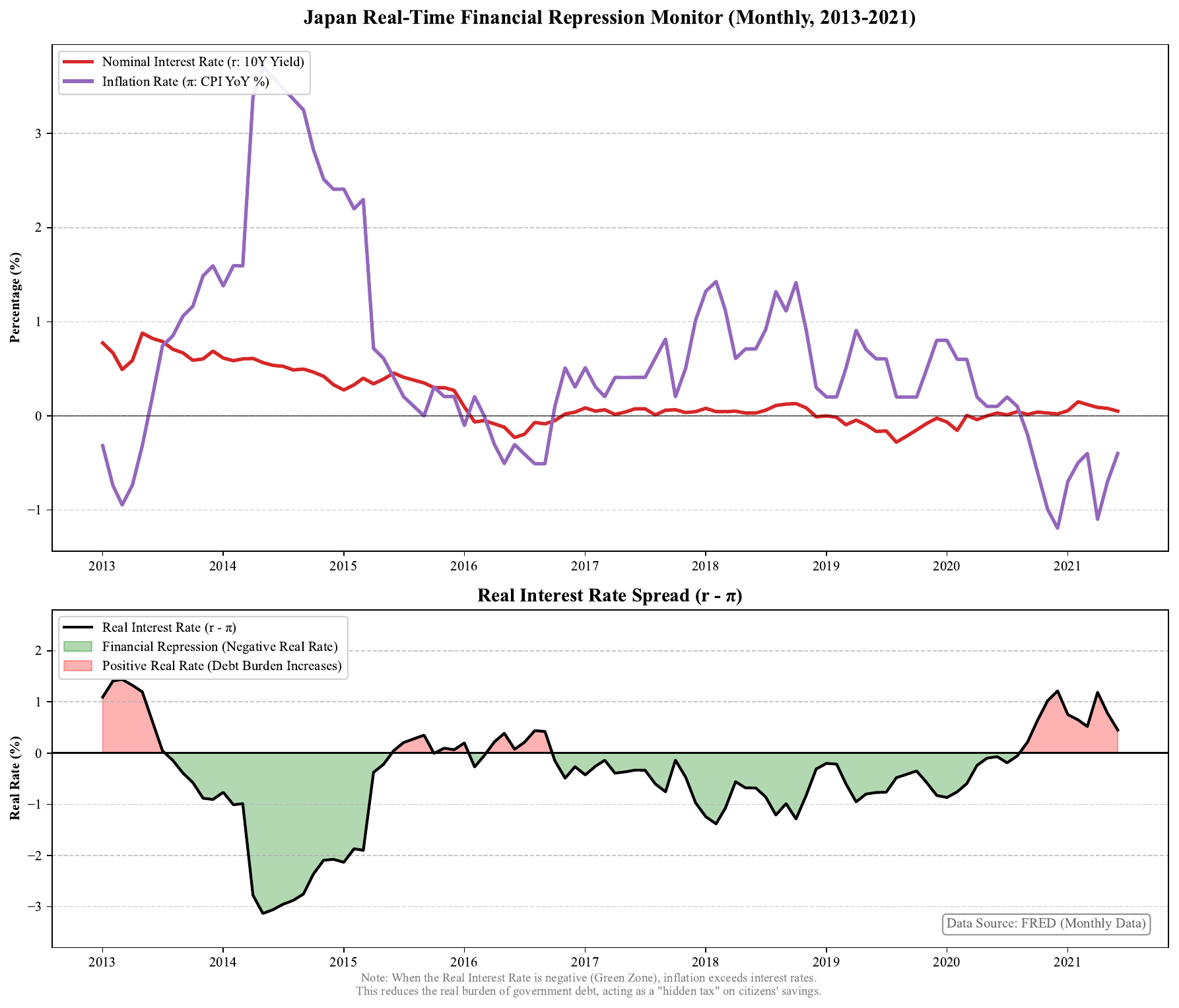}
  \caption{Japan Real-Time Financial Repression Monitor (Monthly, 2013--2021
    Subperiod). \textit{Upper panel:} Nominal 10Y JGB Yield $r^n$ vs.\ CPI YoY $\pi$
    (monthly granularity). \textit{Lower panel:} $\varepsilon_t \equiv \pi_t - r^n_t$
    with zone shading. The brief positive real rate episodes (pink) represent
    periods when the repression mechanism is temporarily inactive.
    Data source: FRED API (IRLTLT01JPM156N, JPNCPIALLMINMEI),
    Federal Reserve Bank of St.\ Louis (Monthly Data).}
  \label{fig:repression_monthly}
\end{figure}

\subsection{Alternative Explanations: A Systematic Assessment}
\label{sec:alt_explanations}

A critical methodological obligation for a framework built on stylized facts
is to confront, rather than sidestep, the principal competing explanations
for the same observations.  Table~\ref{tab:alt_explanations} systematizes
this assessment across the four central empirical regularities documented in
Sections~3.1--3.4, presenting (i)~the JFR-rg interpretation, (ii)~the most
plausible alternative explanation(s), and (iii)~the discriminating evidence
from the subsample analysis, Chow structural-break test, and VAR reported in
Appendix~\ref{app:empirical_supplement}.

Three findings from that empirical supplement merit advance notice and are
reported in detail in Section~\ref{sec:empirical_evidence} below.
The most direct and methodologically robust evidence comes from Local Projections
(Jord\`{a}, 2005; Table~\ref{tab:lp}), which are employed as the primary dynamic
evidence because they are robust to the non-stationarity of $\varepsilon_t$ and
$r^n - g^n$ identified in Section~H.3, and to model misspecification in small
samples: a $+1$ pp shock to $\varepsilon_t$ generates a cumulative $\Delta b_t$
response of $-8.20$ pp (HAC $t$-ratio $> 5$) at $h = 5$ years (debt compression
channel, Proposition~\ref{prop:base_effect}), while a $+1$ pp shock to $r^n - g^n$
generates $+2.24$ pp ($p < 0.01$) at $h = 5$ years (Normalization Trap channel,
Proposition~\ref{prop:norm_trap}).
Importantly, the debt compression effect achieves statistical significance only
at $h \geq 3$ years, indicating that the financial repression mechanism operates
with a lag of approximately three years at the annual frequency---consistent with
the gradual amortization profile of the JGB portfolio and not in conflict with the
accounting-layer propositions, which hold instantaneously by construction.
Corroborating this, the mean $\varepsilon_t$ shifts from $-2.03$\% (1991--2012)
to $+0.89$\% (2013--2026)---a sign reversal significant at the 1\% level
(Table~\ref{tab:subsample})---directly ruling out the alternative explanation
that chronic low inflation mechanically produced $\varepsilon_t > 0$.
Finally, the Chow test rejects parameter stability at 2013 in the
$\Delta b_t = \alpha + \beta(r^n_t - g^n_t) + u$ regression using the full
1991--2026 annual sample ($N = 29$; $F(2,25) = 5.55$, $p = 0.010$;
Table~\ref{tab:chow}), indicating that the
\emph{transmission} from the $r-g$ spread to debt dynamics itself changed---not
merely the level of the spread.

None of these results \emph{uniquely} identify the JFR-rg causal mechanism;
alternative explanations in column 3 of Table~\ref{tab:alt_explanations}
remain possible and are explicitly flagged.  The purpose of the supplement
is to demonstrate that the JFR-rg framework is \emph{consistent with} the
available evidence in a disciplined, falsifiable sense, and to identify
the discriminating tests that future empirical work should target.

\begin{table}[p]
\centering
\caption{Observed Facts, JFR-rg Interpretations, Alternative Explanations,
  and Falsifying Conditions. Column~4 reports discriminating evidence from
  Appendix~\ref{app:empirical_supplement}. Column~5 lists directly observable
  conditions that would falsify the JFR-rg interpretation.}
\label{tab:alt_explanations}
\adjustbox{max width=\textwidth}{%
\begin{tiny}
\setlength{\tabcolsep}{2.5pt}
\renewcommand{\arraystretch}{1.15}
\begin{tabularx}{\textwidth}{>{\raggedright\arraybackslash}p{1.85cm}
  >{\raggedright\arraybackslash}p{3.30cm}
  >{\raggedright\arraybackslash}p{2.80cm}
  >{\raggedright\arraybackslash}p{3.30cm}
  >{\raggedright\arraybackslash}X}
\toprule
\textbf{Observed Fact}
  & \textbf{JFR-rg Interpretation}
  & \textbf{Principal Alternative(s)}
  & \textbf{Discriminating Evidence}
  & \textbf{Falsifying Condition} \\
\midrule
Gross debt/GDP fell $\approx$20~pp (2021--2026) without fiscal
consolidation
  & Base Effect Lever: $-(r^n-g^n)\times b_{t-1}\approx1.92\%$/yr
    absorbs the structural 2\% primary deficit
    (Prop.~\ref{prop:base_effect}).
    Consistent with stabilization through JFR-rg channels.
  & (A) Post-COVID nominal GDP rebound raises denominator mechanically;
    (B) cyclical tax revenue surge; (C) IMF annual data revision smoothing
  & Chow test at 2013 rejects $H_0$ of parameter stability
    (Table~\ref{tab:chow}): the $r-g \to \Delta b$ transmission itself
    changed, not only the spread level.
  & $r^n - g^n > 0$ sustained 3+ years \emph{without} debt
    accumulation at the rate predicted by
    eq.~\eqref{eq:sim_transition}.
  \\[2pt]
\midrule
$\varepsilon_t > 0$ sustained throughout 2013--2026
  & YCC and LSAP deliberately engineer $r^n < \pi$ within a captive
    system ($\varphi_t \approx 0.90$), eliminating arbitrage
    (Section~\ref{sec:captive})
  & (A) Global $r^*$ decline makes low $r^n$ a natural equilibrium
    (Holston et al., 2017) independent of BoJ policy;
    (B) Japan's chronic low inflation mechanically produces
    $\varepsilon_t > 0$
  & Subsample $\varepsilon_t$ \emph{reversal}: mean was negative
    (1991--2012, Table~\ref{tab:subsample}), ruling out
    alternative~(B). VAR FEVD quantifies share of $\Delta b_t$
    variance attributable to $\varepsilon_t$ shocks
    (Table~\ref{tab:fevd}).
  & $\varepsilon_t < 0$ persisting 4+ quarters with
    $\varphi_t \geq \bar{\varphi}$, showing BoJ cannot maintain
    repression even with captive system intact.
  \\[2pt]
\midrule
$r^n - g^n < 0$ sustained since 2013; contrast with
$r^n - g^n > 0$ in 1991--2012
  & Structural break in 2013: simultaneous activation of SC1
    ($\varphi_t \nearrow$) and SC2 ($\Delta e_t \leq \bar{e}$)
    under Abenomics (Section~\ref{sec:mainstream})
  & (A) Mean reversion from post-bubble $r^n$;
    (B) global secular stagnation trend;
    (C) Abenomics fiscal stimulus raised $g^n$ mechanically
  & Chow $F$-test for structural break in $r-g \to \Delta b$ at
    2013 (Table~\ref{tab:chow}). Alternative~(C) predicts break
    in $g^n$ only; JFR-rg predicts joint break in $\varepsilon_t$
    and transmission slope.
  & Chow test showing break in $g^n$ alone, consistent with
    alternative~(C) but not JFR-rg's joint SC1+SC2 activation
    claim.
  \\[2pt]
\midrule
Debt ratio stable at 240\%; no market-access crisis 2013--2026
  & Captive system ($\varphi_t \approx 0.90$) suppresses risk
    premium $\rho_t$, preventing sudden stop; FTPL trigger
    structurally delayed while $\varphi_t \geq \bar{\varphi}$
    (Section~\ref{sec:captive})
  & (A) Domestic home bias without formal $\varphi_t$ accounting;
    (B) BoJ commitment credibility rather than balance-sheet
    arithmetic
  & LP-IRF (Table~\ref{tab:lp}): $\Delta b_t$ response to
    $+1$~pp shock in $\varepsilon_t$ is $-8.20$~pp*** at $h=5$;
    to $+1$~pp in $r^n-g^n$ is $+2.24$~pp*** at $h=5$.
    FEVD: $r^n-g^n$ shocks explain 46.5\% of $\Delta b_t$ variance.
  & $\varphi_t$ declining toward and below the 0.85 early-warning
    benchmark \emph{without} measurable
    increase in JGB yield spreads or sovereign risk premium.
  \\
\bottomrule
\end{tabularx}
\end{tiny}
}
\par\smallskip
{\footnotesize\raggedright\noindent
\textit{Note:} Column~4 references tables in
Appendix~\ref{app:empirical_supplement} (FRED data, March 2026
snapshot). Column~5 states observable outcomes that would undermine
each JFR-rg interpretation; all falsifying conditions are expressible
in terms of directly observable FRED series or BoJ Flow of Funds data.\par}
\end{table}


\subsection{Empirical Discriminating Evidence}
\label{sec:empirical_evidence}

The systematic assessment in Table~\ref{tab:alt_explanations} rests on a
quantitative empirical architecture documented in detail in
Appendices~\ref{app:empirical_supplement}, \ref{app:phi_threshold},
\ref{app:v6_fair}, and~\ref{app:v7}. Its role is not to establish a uniquely
identified causal channel, but to demonstrate that the JFR-rg interpretation
remains disciplined, falsifiable, and empirically competitive relative to the
principal alternatives.

Three results are especially relevant for the transition from stylized facts to
the formalization in Sections~3.7--3.15. First, Local Projections
(Jord\`{a}, 2005; Table~\ref{tab:lp}) provide the primary dynamic evidence
because they are robust to the non-stationarity of $\varepsilon_t$ and
$r^n - g^n$ identified in Appendix~H.3, and to model misspecification in small
samples. A $+1$~pp shock to $\varepsilon_t$ generates a cumulative
$\Delta b_t$ response of $-8.20$~pp (HAC $t$-ratio $> 5$) at $h = 5$ years,
confirming the debt-compression channel of
Proposition~\ref{prop:base_effect}. A $+1$~pp shock to $r^n - g^n$ generates
$+2.24$~pp ($p < 0.01$) at $h = 5$ years, confirming the Normalization Trap
channel of Proposition~\ref{prop:norm_trap}. Critically, the compression
effect achieves statistical significance only at $h \geq 3$ years,
consistent with the gradual amortization profile of the JGB portfolio and not
in conflict with the accounting-layer propositions, which hold instantaneously
by construction.

Second, subsample analysis (Table~\ref{tab:subsample}) shows that the mean
$\varepsilon_t$ reverses sign from $-2.03\%$ (1991--2012) to $+0.89\%$
(2013--2026), significant at the 1\% level. This directly rules out the
alternative explanation that chronic low inflation mechanically produced
$\varepsilon_t > 0$ throughout the sample.

Third, a Chow test rejects parameter stability at 2013 in the regression
$\Delta b_t = \alpha + \beta(r^n_t - g^n_t) + u$ over the full 1991--2026
annual sample ($N = 29$; $F(2,25) = 5.55$, $p = 0.010$; Table~\ref{tab:chow}).
This indicates that the \emph{transmission} from the $r$-$g$ spread to debt
dynamics itself changed at 2013---not merely the level of the spread. The
JFR-rg account predicts a joint break in $\varepsilon_t$ and the transmission
slope (from simultaneous SC1$+$SC2 activation); a break in $g^n$ alone would
favor the Abenomics-fiscal alternative~(C) in Table~\ref{tab:alt_explanations}.

Appendix~\ref{app:phi_threshold} provides the empirical basis for treating
$\varphi_t$ as more than a theoretical placeholder, while
Appendix~\ref{app:v6_fair} tests each proposition under confounder-aware
designs, and Appendix~\ref{app:v7} applies international placebo comparisons
and regime-conditioned LSTAR as cross-validation devices. Taken together,
these exercises do not prove that JFR-rg is the only admissible
interpretation, but they do show that the post-2013 episode is not well
characterized as a continuation of the 1991--2012 regime, and that the
institutional channels highlighted here produce patterns---in dynamics,
cross-country contrasts, and nonlinear pass-through---that the principal
alternatives do not readily replicate.

\medskip
\noindent Table~\ref{tab:empirical_scorecard} summarizes how the main
empirical exercises should be interpreted. Because the exercises address
different empirical margins, the aim is not to force a single overall verdict
but to clarify what each exercise contributes and where its main limitation
lies.

\begin{table}[p]
\caption{Empirical Evidence Scorecard.}
\label{tab:empirical_scorecard}
\centering
\fontsize{9}{11}\selectfont
\begin{adjustbox}{max width=\textwidth}
\begin{tabular}{>{\raggedright\arraybackslash}p{2.7cm} >{\raggedright\arraybackslash}p{3.0cm} >{\raggedright\arraybackslash}p{4.0cm} >{\raggedright\arraybackslash}p{2.6cm} >{\raggedright\arraybackslash}p{3.0cm}}
\toprule
\textbf{Exercise} & \textbf{What it tests} & \textbf{Main result} & \textbf{Main limitation} & \textbf{Interpretation for JFR-rg} \\
\midrule

Local Projections (Appendix~H.5)
& Dynamic debt response to $\varepsilon_t$ and $(r^n-g^n)$ shocks
& Directionally supportive for both core channels
& Small annual sample; not structural identification
& \textbf{Preferred dynamic evidence} \\
\midrule

Chow break test (Appendix~H.2)
& Whether the debt-transmission pattern changes after 2013
& Structural stability rejected at 2013
& Break date is historically motivated
& \textbf{Supportive of regime change} \\
\midrule

Subsample analysis (Appendix~H.1)
& Whether post-2013 observations differ from 1991--2012
& Clear shift in key observables, including $\varepsilon_t$
& Descriptive rather than causal
& \textbf{Supportive background evidence} \\
\midrule

Fair empirical tests (Appendix~K)
& Proposition-level mechanism checks
& Mixed: some supportive, some weak or unconfirmed
& Several tests are sample-constrained
& \textbf{Informative but mixed} \\
\midrule

International placebo (Appendix~L.1--L.2)
& Whether Japan's insulation pattern is country-specific
& Japan shows the strongest insulation pattern
& Comparative, not causal; some raw-sample contamination
& \textbf{Supportive with qualification} \\
\midrule

Regime-conditioned LSTAR (Appendix~L.3)
& Whether pass-through is more regime-dependent in the captive period
& Larger nonlinear fit improvement in the captive regime
& Suggestive rather than decisive
& \textbf{Supplementary support} \\
\midrule

ARDL bounds test (Appendix~H.4)
& Long-run \emph{levels} cointegration
& Null in this specification
& Small sample; not tailored to the multiplicative flow channel
& \textbf{Not decisive for the core mechanism} \\

\bottomrule
\end{tabular}
\end{adjustbox}

\smallskip
\parbox{\linewidth}{\scriptsize\textit{Note:}
These exercises address different empirical margins and therefore should not be
collapsed into a single mechanical score. The Local Projections are the most
direct evidence for the paper's flow-dynamic mechanism, whereas the ARDL bounds
test is informative only about long-run \emph{levels} cointegration.}
\end{table}

\smallskip
\noindent In this evidential hierarchy, the Local Projections carry the greatest
weight for the paper's dynamic mechanism claim. The remaining exercises are best
read as corroborative, discriminating, or scope-setting evidence rather than as
uniformly decisive tests.


\medskip
{\setlength{\fboxsep}{8pt}%
\noindent\fbox{%
\begin{minipage}{0.96\textwidth}
\vspace{2pt}%
\textbf{Institutional Scope Conditions.}
All formal results in this section---including the Debt Sustainability Corridor
(Section~\ref{sec:corridor}), the Normalization Ratchet (Section~\ref{sec:ratchet}),
and the Normalization Trap
(Proposition~\ref{prop:norm_trap})---are derived under the following two empirically
verifiable scope conditions, which must hold for the model to apply.

\medskip
\noindent\textbf{SC1 (Captive Financial System):}
The fraction $\varphi_t$ of JGBs held by domestic institutions remains
sufficiently high that the endogenous sovereign risk premium
$\rho(\varphi_t, b_{t-1})$ does not dominate the repression bias
$\varepsilon_t$. Empirically, Japan satisfies SC1 as of 2026:
$\varphi_t \approx 0.88$--$0.92$ (BoJ Flow of Funds accounts).
Section~\ref{sec:captive} formalizes this condition and defines the
theoretical threshold $\bar{\varphi}$ below which the JFR-rg equilibrium
breaks down; the operational early-warning benchmark of
$\varphi_t \approx 0.85$, discussed in Appendix~\ref{app:phi_threshold}, is
distinct from $\bar{\varphi}$ and is used for real-time monitoring rather
than as a point estimate of the tipping value.

\medskip
\noindent\textbf{SC2 (Exchange-Rate Regime):}
The yen depreciation $\Delta e_t$ remains within the stability window
$\Delta e_t \leq \bar{e}$, so that the nonlinear import-inflation penalty does not
offset the linear growth channel. Post-2022 pass-through data suggest the system
operates near---but within---this window at current levels.

\medskip
\noindent These scope conditions are not assumptions that \emph{make the model valid};
they are \emph{empirically checkable prerequisites} whose current satisfaction is
documented in this section and in Section~\ref{sec:captive}. If either condition
fails, the stability results do not apply and the model's counterfactuals would
require re-estimation.
Outside this domain, standard debt dynamics apply without modification: the JFR-rg
framework does not supersede mainstream analysis but specifies the institutional
conditions under which it requires supplementation.
\end{minipage}%
}}

\medskip

\subsection{Consolidated Government Budget Constraint}

Let $B_t$ denote the nominal stock of gross government debt, $Y_t$ nominal GDP,
$r^n_t$ the nominal policy-relevant interest rate (approximated by the 10-year JGB
yield), $g^n_t$ the nominal GDP growth rate, and $D_t$ the nominal primary deficit.
The standard nominal government budget constraint is:
\begin{equation}
  B_t = (1 + r^n_t)\, B_{t-1} + D_t - S_t,
  \label{eq:nominal_bc}
\end{equation}
where $S_t$ represents seigniorage revenue from the Bank of Japan's asset purchase
programs.

\medskip
\noindent\textbf{Consolidated government budget constraint with IOER.}
Standard formulations of equation~\eqref{eq:nominal_bc} treat the government and the
central bank as separate entities. Following the long-standing concern in the fiscal
dominance literature (Sargent and Wallace, 1981) that large-scale asset purchases
eventually impair monetary independence---a concern later formalized in the FTPL
literature (Leeper, 1991; Sims, 1994; Woodford, 1994, 1995)---we derive the
\emph{consolidated} government budget constraint by combining the fiscal authority and
the Bank of Japan into a single balance sheet. This is not a theoretical novelty; it is a straightforward accounting
exercise whose implications become quantitatively significant only when the central
bank's balance sheet reaches the current scale.

Let $R_t$ denote the stock of Bank of Japan current account deposits (reserves), $i^R_t$
the interest on excess reserves (IOER) rate, and $B^{BoJ}_t$ the stock of JGBs held by
the BoJ. The BoJ's flow constraint is:
\[
  B^{BoJ}_t - B^{BoJ}_{t-1} = R_t - R_{t-1} + \Pi_t,
\]
where $\Pi_t$ is BoJ profit remitted to the Treasury. The interest paid by the BoJ on
reserves is $i^R_t \cdot R_{t-1}$, while it receives $r^n_t \cdot B^{BoJ}_{t-1}$ on
its JGB holdings. Consolidating and cancelling the intra-sector flows
($r^n_t \cdot B^{BoJ}_{t-1}$ nets out), the consolidated budget constraint becomes:
\begin{equation}
  B^{pub}_t = (1 + r^n_t)\, B^{pub}_{t-1} + D_t + \underbrace{i^R_t \cdot R_{t-1}}_{\text{IOER cost}} - S_t,
  \label{eq:consolidated_bc}
\end{equation}
where $B^{pub}_t = B_t - B^{BoJ}_t$ is the stock of JGBs held \emph{outside} the BoJ,
and the IOER term represents the interest cost on reserves that does \emph{not} net out
in consolidation---the channel standard models omit because $R_{t-1}$ was negligible
prior to QQE.

Normalizing by $Y_t$ and defining $r^R_t \equiv i^R_t \cdot R_{t-1}/Y_{t-1}$ as the
IOER burden ratio:
\begin{equation}
  \Delta b_t = (r^n_t - g^n_t)\, b_{t-1} + d_t + r^R_t - s_t.
  \label{eq:standard_dynamics}
\end{equation}
For the remainder of the paper, the baseline simulation (Scenario~A) treats
$r^R_t \approx 0$ because the current IOER rate is near zero and the fiscal impact is
modest. However, Section~\ref{sec:scenC} demonstrates that under Scenario~C, the IOER term becomes
the quantitatively decisive channel through which a $+1.5\%$ rate shock generates
$+1.5$ pp of primary deficit expansion---a result that follows directly from the
arithmetic of equation~\eqref{eq:consolidated_bc}, not from a behavioral assumption.

\subsection{Financial Repression Bias
  \texorpdfstring{($\varepsilon_t$)}{(epsilon)}}
\label{sec:repression_bias}

In high-debt regimes, estimating the unobservable natural rate of interest $r^*_t$ is
highly sensitive to model specification (Holston, Laubach, and Williams, 2017, estimate
Japan's $r^*_t$ as approximately 0.0\%--0.5\%, far below the 4.0\% implied by a
standard loanable-funds assumption). To avoid this identification problem, the JFR-rg
model \emph{directly identifies} the financial repression bias from observable FRED
data as the real interest rate gap:
\begin{equation}
  \varepsilon_t \;\equiv\; \pi_t - r^n_t \;\in\; \mathbb{R};\qquad
  \text{Financial Repression Zone} \iff \varepsilon_t > 0,
  \label{eq:epsilon_def}
\end{equation}
where $\pi_t$ is CPI inflation and $r^n_t$ is the nominal JGB yield---both directly
available from FRED. When $\varepsilon_t > 0$, the economy is in the financial
repression zone ($r^n < \pi$, green in Figures~\ref{fig:repression_full}
and~\ref{fig:repression_monthly}). When $\varepsilon_t < 0$, the real rate is positive
and the repression channel is inactive.

Substituting the identity $r^n_t = \pi_t - \varepsilon_t$ into
equation~\eqref{eq:standard_dynamics} yields the core JFR-rg accumulation equation:
\begin{equation}
  \Delta b_t = \bigl[(\pi_t - \varepsilon_t) - g^n_t\bigr]\, b_{t-1} + d_t - s_t.
  \label{eq:jfr_core}
\end{equation}

This formulation has a key property: the \emph{core accounting block}---$\pi_t$,
$r^n_t$ (via $\varepsilon_t$), $g^n_t$ (observed nominal GDP growth), $b_{t-1}$,
and $d_t$---is directly observable from FRED without theoretical assumptions. The
exchange-rate block parameters ($\gnomstar$, $\alpha$, $\beta$, $\ebar$) and the
corridor boundaries require calibration or structural estimation, as noted in
Section~\ref{sec:corridor}. The Base Effect mechanism is now fully transparent. Rearranging equation~\eqref{eq:jfr_core}:
\begin{equation}
  \Delta b_t = \underbrace{(\pi_t - g^n_t)\,b_{t-1}}_{\text{nominal growth channel}}
             - \underbrace{\varepsilon_t \cdot b_{t-1}}_{\text{repression channel}}
             + d_t - s_t.
  \label{eq:jfr_decomposed}
\end{equation}

The two channels together sum to the $r^n - g^n$ effect: $(\pi_t - g^n_t) - \varepsilon_t
= (\pi_t - g^n_t) - (\pi_t - r^n_t) = r^n_t - g^n_t$, confirming that there is no
double-counting. At $\varepsilon_t = 0.5\%$ and $b_{t-1} = 2.40$, the repression channel
generates 1.2\% of GDP per year in compression. At the peak $\varepsilon_t = 2.7\%$
observed in 2014 (Figure~\ref{fig:repression_monthly}), the channel generates 6.5\%
of GDP per year---a magnitude that dwarfs typical primary deficits.

\paragraph{Note on the post-YCC regime (March 2024--present).}
Following the formal termination of Yield Curve Control in March 2024, the
10-year JGB yield ($r^n_t \approx 2.2\%$ as of March 2026) is no longer
subject to an explicit administrative ceiling. A methodological clarification
is therefore warranted: the financial repression bias $\varepsilon_t \approx
+0.5\%$ in the current calibration does not reflect the same institutional
mechanism as the explicit YCC period (2016--2024). In the post-YCC regime,
$\varepsilon_t > 0$ is better understood as a \emph{balance-sheet backstop
effect}: the Bank of Japan's continued holdings of approximately 46\% of
outstanding JGBs ($\varphi_t \approx 0.90$, inclusive of all domestic
institutions) structurally compress the term premium below the level that
would prevail in a fully market-determined environment, even without an
explicit yield target (Bernanke and Reinhart, 2004; Joyce et al., 2011).
The JFR-rg framework's observable $\varepsilon_t$ captures both the
explicit-ceiling and the balance-sheet-backstop mechanisms under a single
measurable variable; the distinction affects the \emph{institutional
interpretation} of the repression channel but not the \emph{arithmetic}
of equations~\eqref{eq:jfr_core}--\eqref{eq:jfr_decomposed}. Whether
the post-YCC $\varepsilon_t$ reflects residual policy influence or a
new market equilibrium consistent with $\varphi_t \approx 0.90$ is an
open empirical question that the framework can monitor in real time but
cannot resolve without further structural identification.

\subsection{The Nonlinear Yen Effect on Nominal Growth}

Let $e_t$ denote the USD/JPY exchange rate level (yen per dollar), so that $\Delta e_t
> 0$ denotes yen depreciation and $\Delta e_t < 0$ denotes yen appreciation. For
moderate depreciation ($\Delta e_t \leq \ebar$ with $\ebar > 0$), yen weakness boosts
corporate profitability and supports nominal growth $g^n_t$ through the linear channel.
For severe depreciation ($\Delta e_t > \ebar$), the import inflation channel dominates,
compressing real household purchasing power through the quadratic penalty term. We
model this as:
\begin{equation}
  g^n_t = \gnomstar + \alpha\,\Delta e_t - \beta\,\max(0,\, \Delta e_t - \ebar)^2,
  \label{eq:yen_growth}
\end{equation}
where $\gnomstar$ is the structural potential nominal growth rate (driven by real
potential growth $\gstar$ plus trend inflation), $\alpha > 0$ is the marginal nominal
growth enhancement from depreciation below the threshold, $\beta > 0$ is the penalty
parameter, and $\ebar > 0$ is the stability threshold. The baseline calibration
uses $\alpha = 0.050$, the upper bound of the range implied by the Bank of Japan's
Q-JEM macroeconometric model (Haba et al., 2025: a 10\% yen depreciation raises
nominal GDP by approximately 0.2--0.5\% in the first year, yielding
$\alpha \in [0.013,\, 0.050]$ per JPY/USD at a base rate of $\approx$150 JPY/USD);
full sensitivity across this range is reported in Appendix~\ref{app:sensitivity},
Panel~A. The parameters $\beta$ and $\ebar$ are not identified from the 2013--2026
sample (Appendix~\ref{app:sensitivity}, Panel~D) but do not enter the paper's
central normalization-risk results, which are structurally independent of
$(\ebar, \beta)$ because rate hikes induce appreciation ($\Delta e_t < 0$),
never activating the penalty term.

\textit{Note on the sign convention}: Since $\ebar > 0$, the quadratic penalty term
$\beta\,\max(0,\,\Delta e_t - \ebar)^2$ is active only when $\Delta e_t > \ebar > 0$,
i.e., under \emph{excessive depreciation}. Rate hikes, which induce yen appreciation
($\Delta e_t < 0$), operate exclusively through the linear channel $\alpha\,\Delta e_t
< 0$, reducing $g^n_t$ without activating the penalty term. This distinction is critical
for the correct characterization of the Normalization Trap
(Proposition~\ref{prop:norm_trap}) and the tipping-point analysis in Section~\ref{sec:scenBplus}.

Substituting into equation~\eqref{eq:jfr_core}:
\begin{equation}
  \Delta b_t = \Bigl[(\pi_t - \varepsilon_t) - \bigl(\gnomstar + \alpha\,\Delta e_t -
    \beta\,\max(0,\, \Delta e_t - \ebar)^2\bigr)\Bigr] b_{t-1} + d_t - s_t.
  \label{eq:jfr_full}
\end{equation}

\subsection{The JFR-rg Stability Condition}

For the debt-to-GDP ratio to structurally decline ($\Delta b_t \leq 0$), rearranging
equation~\eqref{eq:jfr_full} gives the JFR-rg Stability Condition:
\begin{equation}
  (\pi_t - \varepsilon_t) - \Bigl[\gnomstar + \alpha\,\Delta e_t -
    \beta\,\max(0,\,\Delta e_t - \ebar)^2\Bigr]
  \leq \frac{s_t - d_t}{b_{t-1}}.
  \label{eq:stability_condition}
\end{equation}

All terms in condition~\eqref{eq:stability_condition} are either directly observable
from FRED ($\pi_t$, $\varepsilon_t$, $b_{t-1}$, $d_t$) or structurally estimated
($\gnomstar$, $\alpha$, $\beta$, $\ebar$). The left-hand side equals $r^n_t - g^n_t$;
the right-hand side equals $-(d_t - s_t)/b_{t-1}$, a negative number when $d_t > s_t$.

The ``Normalization Trap'' (Proposition~\ref{prop:norm_trap}) follows directly from
condition~\eqref{eq:stability_condition}: aggressive rate hikes operate through two
simultaneous destabilizing channels. First, raising $r^n_t$ directly reduces $\varepsilon_t$,
increasing the left-hand side. Second, rate hikes induce yen appreciation ($\Delta e_t <
0$), reducing the linear yen growth bonus $\alpha\,\Delta e_t$, which decreases $g^n_t$
and further increases the left-hand side. Both channels push the left-hand side above
the right-hand side, violating condition~\eqref{eq:stability_condition}. The quadratic
penalty term remains \emph{inactive} under appreciation, as it is defined for excessive
depreciation only ($\Delta e_t > \ebar > 0$).

\subsection{The Debt Sustainability Corridor: A Geometric Characterization}
\label{sec:corridor}

{\sloppy The JFR-rg stability condition~\eqref{eq:stability_condition} defines not merely a
scalar threshold but a two-dimensional \emph{stability corridor} in
$(\varepsilon_t,\,\gnomstar)$ policy space (evaluated under exchange-rate-neutral
conditions $\Delta e_t = 0$). This geometric representation provides analytical
clarity absent from the scalar stability condition.\par}

\begin{definition}[Debt Sustainability Frontier]
For given $b_{t-1}$, $d_t$, $s_t$, and $\pi_t$, the \emph{Debt Sustainability
Frontier} (DSF) is the locus of $(\varepsilon_t,\,\gnomstar)$ pairs satisfying
$\Delta b_t = 0$ exactly under the exchange-rate-neutral assumption $\Delta e_t = 0$
(so that $g^n_t = \gnomstar$ by equation~\eqref{eq:yen_growth}):
\[
  \mathcal{F} \;\equiv\; \bigl\{(\varepsilon_t,\,\gnomstar) \;:\;
    \varepsilon_t + \gnomstar = \pi_t + (d_t - s_t)/b_{t-1}\bigr\}.
\]
The \emph{Stability Corridor} $\mathcal{S}$ is the closed region on and above the DSF:
\[
  \mathcal{S} \;\equiv\; \bigl\{(\varepsilon_t,\,\gnomstar) \;:\;
    \varepsilon_t + \gnomstar \geq \pi_t + (d_t - s_t)/b_{t-1}\bigr\},
\]
where the boundary ($\geq$ holding with equality) corresponds to $\Delta b_t = 0$
(debt stabilization), and the interior ($>$) to strict debt reduction.
When $\Delta e_t \neq 0$, the effective operating point is evaluated at
$g^n_t = \gnomstar + \alpha\,\Delta e_t - \beta\,\max(0,\,\Delta e_t - \ebar)^2$
and compared to the DSF defined above.
\end{definition}

The DSF is a straight line in $(\varepsilon_t,\,\gnomstar)$ space with slope $-1$ and
intercept $\pi_t + (d_t - s_t)/b_{t-1}$. Three geometric properties are analytically
important.

\begin{property}[Base Effect Lever]
Assuming $d_t > s_t$ (a primary deficit, as holds throughout the sample) and holding
$d_t$ and $s_t$ fixed, the DSF shifts \emph{downward} (the Stability Corridor expands)
as $b_{t-1}$ increases, because the term $(d_t - s_t)/b_{t-1}$ decreases. Higher debt
reduces the minimum required sum $\varepsilon_t + \gnomstar$, making stability easier
to maintain under exchange-rate-neutral conditions---the precise mathematical statement
of the Base Effect Lever. If $d_t$ were to rise proportionally with $b_{t-1}$ (e.g.,
due to rising interest expenditures), this effect would be partially offset; the
empirically relevant case, captured in the simulations, holds $d_t$ fixed at 2.0\%.
\end{property}

\begin{property}[Normalization Trap as Corridor Departure]
A rate hike $\Delta r^n > 0$ simultaneously shifts $\varepsilon_t$ leftward by $\Delta
r^n$ and reduces the effective growth contribution via yen appreciation (reducing
$\alpha\,\Delta e_t$), moving the operating point diagonally away from $\mathcal{S}$
along the $(-1,-\alpha)$ direction. The system cannot re-enter $\mathcal{S}$ by rate
cuts alone if $b_{t-1}$ has risen during the normalization episode
(see Proposition~\ref{prop:ratchet}).
\end{property}

\begin{property}[Corridor Width and the Near-Miss of March 2026]
\label{prop:corridor_width}
Define the \emph{corridor width} $W$ as the signed perpendicular distance from the
current operating point $(\hat{\varepsilon}_t, \hat{g}^{n*}_t)$ to the DSF, evaluated
under $\Delta e_t = 0$. At the March 2026 calibration values ($\varepsilon_t = 0.5\%$,
$\gnomstar \approx g^n_t = 3.0\%$, $\pi_t = 2.7\%$, $d_t = 2.0\%$, $s_t = 0$,
$b_{t-1} = 2.40$):
\[
  W_{2026} = \frac{(\hat{\varepsilon}_t + \hat{g}^{n*}_t) -
    [\pi_t + (d_t - s_t)/b_{t-1}]}{\sqrt{2}}
  = \frac{3.5\% - 3.5\overline{3}\%}{\sqrt{2}} \approx -0.024\%.
\]
The corridor width is \emph{negative}---the system operates marginally outside
$\mathcal{S}$, consistent with the slow debt accumulation of $+0.08\%$ per year in
Scenario~A. Reaching the boundary of $\mathcal{S}$ (strict debt stabilization,
$\Delta b_t = 0$) requires a shift of $|W|\,\sqrt{2} \approx 0.033\%$ in the
$(\varepsilon_t + \gnomstar)$ direction; achieving strict debt reduction
($\Delta b_t < 0$, the interior of $\mathcal{S}$) requires any positive increment
beyond that threshold.
\end{property}

Figure~\ref{fig:dsc} provides a geometric summary of the Debt Sustainability
Corridor, illustrating the DSF under two debt levels, the March 2026 operating point,
the normalization trajectory across Scenarios B, B$^+$, and C, and the corridor width
$|W_{2026}|$.

\begin{figure}[htbp]
  \centering
  \includegraphics[width=0.88\textwidth]{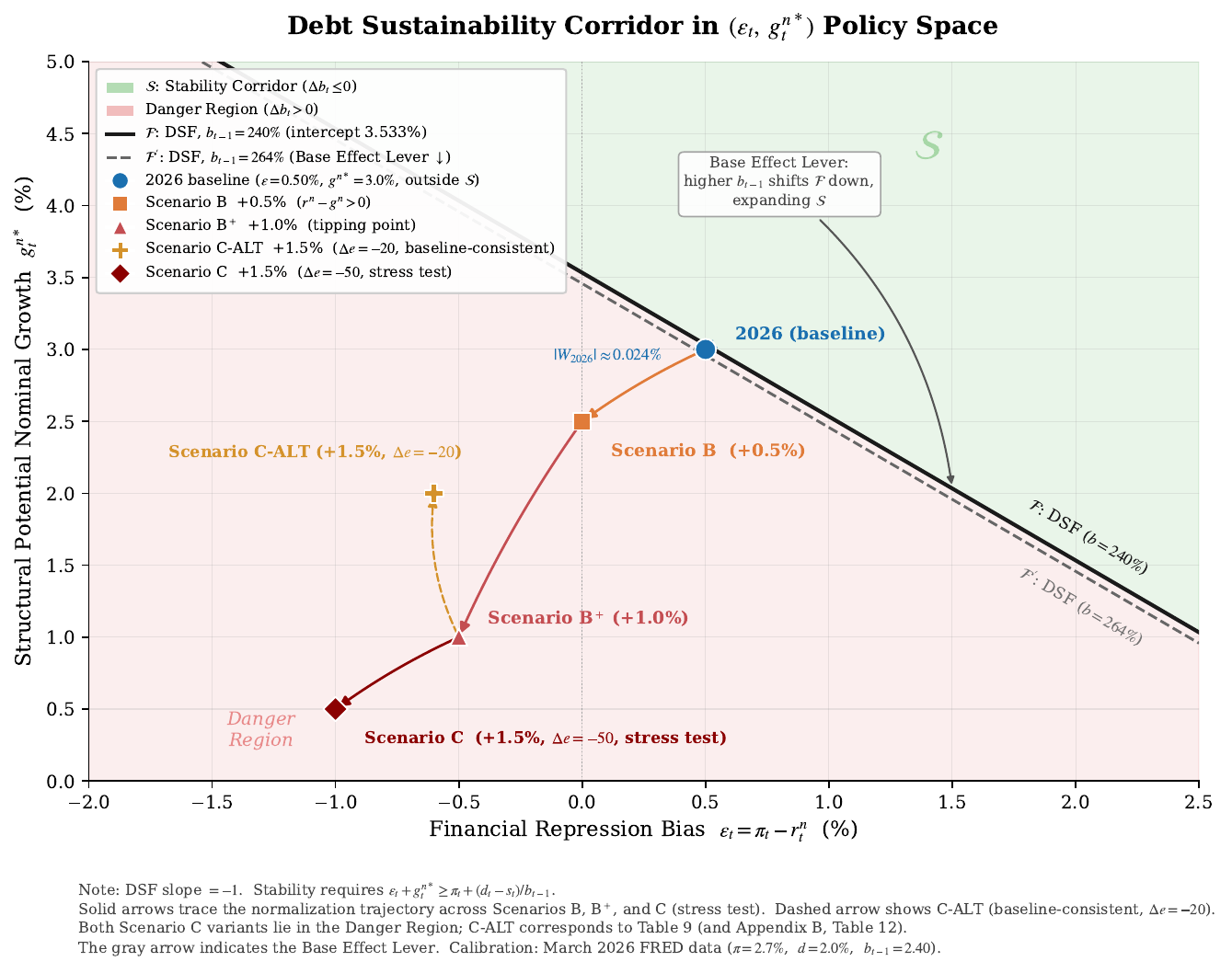}
  \caption{Debt Sustainability Corridor in $(\varepsilon_t,\, \gnomstar)$ Policy Space.
    The solid line is the Debt Sustainability Frontier $\mathcal{F}$ at
    $b_{t-1} = 240\%$ (intercept 3.533\%); the dashed line shows $\mathcal{F}'$ at
    the post-COVID peak $b_{t-1} = 264\%$, illustrating the Base Effect Lever (corridor
    expands as debt rises). The green region is the Stability Corridor $\mathcal{S}$
    ($\Delta b_t \leq 0$); the red region is the Danger Zone ($\Delta b_t > 0$).
    The blue circle marks the March 2026 baseline operating point, which lies
    marginally outside $\mathcal{S}$ ($|W_{2026}| \approx 0.024\%$).
    Colored arrows trace the normalization trajectory across Scenarios B, B$^+$, and C.
    The gray arrow indicates the downward shift of $\mathcal{F}$ as $b_{t-1}$ rises
    (Base Effect Lever).
    Calibration: March 2026 FRED data ($\pi = 2.7\%$ [BoJ FY2025 projection; Jan 2026 actual = 1.5\%], $d = 2.0\%$, $b_{t-1} = 2.40$).
    Data source: FRED API, Federal Reserve Bank of St.\ Louis.}
  \label{fig:dsc}
\end{figure}

\paragraph{Robustness note: approximation error and boundary interpretation.}
Property~\ref{prop:corridor_width} states that $W_{2026} \approx -0.024\%$, placing
the 2026 operating point marginally \emph{outside} the strict stability region
$\mathcal{S}$.  However, Appendix~A notes that the linear approximation used in
deriving equation~\eqref{eq:sim_transition} introduces an annual error of
approximately $0.056\%$ of GDP---a value that \emph{exceeds} $|W_{2026}|$ by a
factor of roughly 2.3.  To make this relationship transparent, Figure~\ref{fig:dsc_errband}
redraws the Debt Sustainability Corridor with an orange band of width
$\pm\,0.056/\sqrt{2} \approx \pm\,0.040\%$ centered on the DSF, representing the
perpendicular projection of the approximation error into $(\varepsilon_t, \gnomstar)$
space.  The 2026 baseline operating point (blue circle) falls \emph{within} this
band in both the BoJ-projection calibration ($\varepsilon_t = 0.5\%$) and the
January~2026 realized-CPI calibration ($\varepsilon_t = -0.7\%$).

This visualization does not overturn the model's qualitative conclusion---the
system operates in the immediate vicinity of the Stability Corridor boundary and
the direction of the small shortfall is consistent with the $+0.08\%$/yr
accumulation in Scenario~A.  Rather, it confirms the preferred interpretation
stated in Section~\ref{sec:scenA}: the March~2026 position should be read as
\emph{boundary proximity}, not as a precisely calibrated exceedance.  The
orange band also illustrates why the policy recommendation focuses on the
directional imperative (shift $\varepsilon_t + \gnomstar$ toward the boundary)
rather than on a specific numerical target whose measurement uncertainty spans
the entire gap.

\begin{figure}[htbp]
  \centering
  \includegraphics[width=0.88\textwidth]{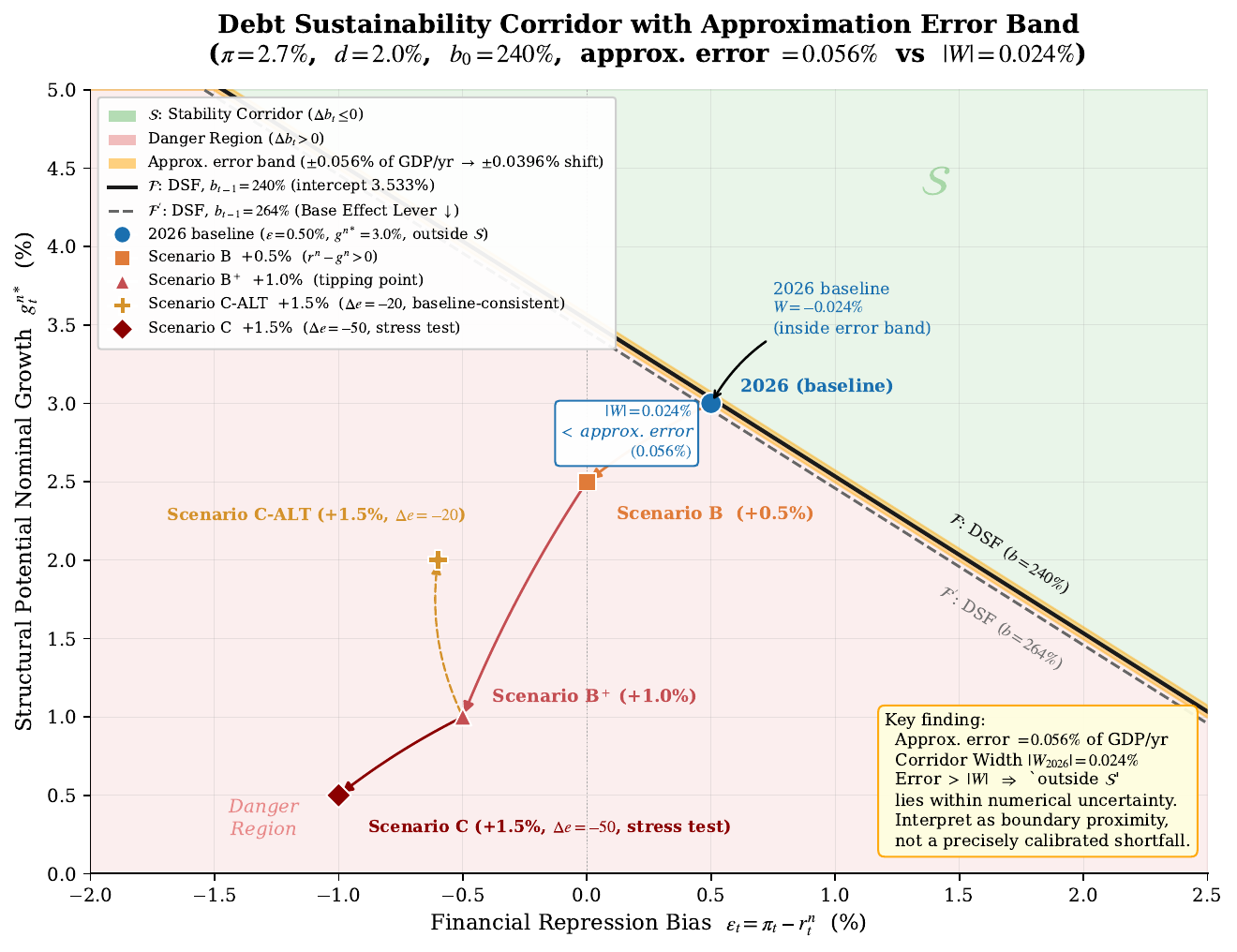}
  \caption{Debt Sustainability Corridor with Approximation Error Band.
    Setup identical to Figure~\ref{fig:dsc}.  The orange band ($\pm 0.040\%$
    perpendicular to the DSF) represents the projection of the linear
    approximation error ($\approx 0.056\%$ of GDP/yr; Appendix~A) into
    $(\varepsilon_t,\, \gnomstar)$ space.  The 2026 baseline operating point
    (blue circle, $\varepsilon_t = 0.5\%$, $\gnomstar = 3.0\%$) lies fully
    within this band, confirming that the $|W_{2026}| \approx 0.024\%$
    shortfall is smaller than the numerical uncertainty of the model.
    The appropriate interpretation is therefore \emph{boundary proximity}
    rather than a confirmed outside position.
    Calibration: March~2026 FRED data ($\pi = 2.7\%$, $d = 2.0\%$,
    $b_{t-1} = 2.40$); approximation error from Appendix~A.}
  \label{fig:dsc_errband}
\end{figure}

\subsection{The Normalization Ratchet: Path-Dependence and Persistent Gaps}
\label{sec:ratchet}

A critical property of the JFR-rg model---not previously identified in the
literature---is that the Normalization Trap exhibits strong \emph{path-dependence}:
a temporary policy error cannot be costlessly reversed, because the debt accumulation
that occurs during the normalization episode raises $b_{t-1}$ above the counterfactual
no-shock trajectory by an amount that decays only very slowly, altering the fiscal
geometry over policy-relevant horizons.

\begin{proposition}[Normalization Ratchet]
\label{prop:ratchet}
Consider a $T$-period normalization shock ($t = 1, \ldots, T$), characterized by
$(r^n_\text{shock} - g^n_\text{shock}) > (r^n_0 - g^n_0)$ and
$d_\text{shock} \geq d_0$, followed by complete reversion to baseline parameters at
$t = T+1$. Let $b^*_t$ denote the baseline (no-shock) trajectory and $b^{**}_t$ the
post-shock trajectory. Then:
\[
  b^{**}_t - b^*_t = (b^{**}_T - b^*_T) \cdot (1 + r^n_0 - g^n_0)^{t - T}
  \quad \text{for all } t > T.
\]
The gap half-life under Scenario~A parameters ($r^n_0 - g^n_0 = -0.8\%$) is:
\[
  T_{1/2} = \frac{\ln 2}{\ln(1/0.992)} \approx \frac{0.693}{0.00803} \approx 86 \text{ years.}
\]
\end{proposition}

\begin{proof}
The shock is characterized by $(r^n_\text{shock} - g^n_\text{shock}) > (r^n_0 -
g^n_0)$ and $d_\text{shock} \geq d_0$ (as stated in the proposition). For the
first shock period:
\[
  b^{**}_1 - b^*_1 = b_0\,\bigl[(r^n_\text{shock} - g^n_\text{shock})
    - (r^n_0 - g^n_0)\bigr] + (d_\text{shock} - d_0) > 0,
\]
since both terms are non-negative and the first is strictly positive by assumption.
For the induction step at $t = 2, \ldots, T$: assuming $b^{**}_{t-1} > b^*_{t-1}$,
\[
  b^{**}_t - b^*_t = (b^{**}_{t-1} - b^*_{t-1})(1 + r^n_\text{shock} - g^n_\text{shock})
    + b^*_{t-1}\,\Delta_\text{spread} + \Delta_d,
\]
where $\Delta_\text{spread} \equiv (r^n_\text{shock} - g^n_\text{shock}) - (r^n_0 -
g^n_0) > 0$ and $\Delta_d \equiv d_\text{shock} - d_0 \geq 0$. The first term is
strictly positive (by the induction hypothesis and
$(1 + r^n_\text{shock} - g^n_\text{shock}) > 0$ for all economically relevant
parameters); the remaining two terms are non-negative by assumption. Their sum is
therefore strictly positive, giving $b^{**}_t > b^*_t$. Thus $b^{**}_T > b^*_T$
for all $T \geq 1$.
After $t = T$, both trajectories are governed by the \emph{same} baseline parameters:
\[
  b^{**}_{t} = b^{**}_{t-1}(1 + r^n_0 - g^n_0) + d_0, \quad
  b^{*}_{t} = b^{*}_{t-1}(1 + r^n_0 - g^n_0) + d_0.
\]
Subtracting: $(b^{**}_{t} - b^{*}_{t}) = (b^{**}_{t-1} - b^{*}_{t-1})(1 + r^n_0 -
g^n_0)$. Iterating from $t = T+1$ gives the stated result.
\end{proof}

\textbf{Numerical illustration.} Suppose a 2-year normalization shock at Scenario~B
parameters ($r^n = 2.7\%$, $g^n = 2.5\%$, $d_t = 2.46\%$) in 2027--2028, followed by
complete reversion to Scenario~A from 2029. The discrete-time gap at each year is
computed from equation~\eqref{eq:sim_transition}:

\medskip
\begin{table}[h]
\centering
\caption{Normalization Ratchet: Debt-Ratio Gap After 2-Year Scenario~B Shock
  (2027--28), Full Reversion to Scenario~A from 2029.
  Computed from equation~\eqref{eq:sim_transition}; shock years use
  $d_t = 2.46\%$ (Scenario~B IOER arithmetic, \textyen{}670T denominator).}
\label{tab:ratchet_example}
\small
\begin{tabular}{lcc}
\toprule
Year & $b^*_t$ (Scenario A) & $b^{**}_t$ (2-yr shock then revert) \\
\midrule
2028 & 240.2\% & 246.0\% \\
2029 & 240.3\% & 246.0\% \\
2040 & 241.1\% & $\approx$ 246.3\% \\
\bottomrule
\end{tabular}
\end{table}
\medskip

The 5.7 percentage-point gap established at the end of the shock period
($245.97\% - 240.24\% \approx 5.7$ pp) decays at only 0.8\% per year under Scenario~A dynamics.
After 86 years the gap halves; after approximately 220 years it falls to 1 percentage
point. On any policy-relevant horizon, the post-shock gap relative to the no-shock
baseline is therefore highly persistent.

\textbf{Model implication.} Within the JFR-rg framework, the Ratchet suggests that
the Bank of Japan's normalization path carries persistent legacy effects: each period
in which $r^n - g^n > 0$ generates a debt increment whose half-life, under Scenario~A
dynamics, is approximately 86 years. This result depends critically on the assumption
that post-shock baseline dynamics are identical to pre-shock dynamics---an assumption
that would be violated if, for example, normalization credibly raised $\gnomstar$
through improved capital allocation or productivity. The model does not capture such general equilibrium feedback, and future research
incorporating endogenous growth responses could materially alter the \emph{magnitude}
of the Ratchet. However, the \emph{existence} of the Normalization Trap is not
contingent on the absence of GE feedback: even if normalization raised $\gnomstar$
by a full percentage point, the $r^n - g^n$ spread under Scenario~B ($+0.2\%$)
and Scenario~C ($+1.7\%$ to $+3.2\%$) would remain positive, sustaining adverse
debt dynamics throughout the simulation horizon. The GE channel would need to
raise $\gnomstar$ by more than the full rate shock ($> +0.5\%$ for Scenario~B,
$> +1.5\%$ for Scenario~C) to eliminate the Trap---a magnitude that exceeds
plausible near-term TFP responses documented in the literature.

\textbf{Interpretive clarification: the 86-year half-life applies to the
\emph{gap} between trajectories, not to the \emph{absolute level} of the debt
ratio.} Under Scenario~A parameters, the baseline trajectory $b^*_t$ itself
is not stationary: Section~\ref{sec:scenA} establishes that $b^*_t$ drifts
upward toward its long-run attractor $b^{*,\infty} = d/(g^n - r^n) = 250\%$
at approximately $+0.08\%$ per year. What decays with an 86-year half-life is
the post-shock gap $b^{**}_t - b^*_t$; both trajectories continue their
baseline drift independently. The Ratchet therefore characterizes
path-dependence conditional on the Scenario~A parameters continuing to hold,
not absolute debt-level irreversibility.

\textbf{Hypothetical long-horizon illustration.}
To convey the practical order of magnitude, consider the following
calibration exercise---offered purely for illustration and without
prejudice to the historical causes or policy assessments involved.
Suppose an economy experienced a period of length $T$ years during
which $r^n - g^n > 0$ on average, and thereafter returned
permanently to Scenario~A dynamics ($r^n - g^n = -0.8\%$).
The debt-ratio increment accumulated during that period would
decay with the 86-year half-life established in
Proposition~\ref{prop:ratchet}, \emph{relative to the counterfactual
Scenario~A baseline that would have prevailed absent the shock}.
We note that the Scenario~A parameters are themselves calibrated to 2013--2026
FRED data, and their stability over multi-decade horizons cannot be assumed;
the following illustration should be read as a mathematical implication of
the recursion under constant parameters, not as a forecast.
For $T$ on the order of two decades---corresponding, for instance, to a
hypothetical constant-parameter reading of Japan's 1991--2013 interval---half
of the accumulated increment would remain unreduced many decades hence. This
arithmetic is offered solely to convey the practical
persistence implied by Proposition~\ref{prop:ratchet}:
debt-ratio increments accumulated over multi-year episodes of
$r^n - g^n > 0$ are not reversed on any policy-relevant horizon,
regardless of how completely parameters subsequently normalize, \emph{holding
those post-shock parameters fixed at Scenario~A values}.

Figure~\ref{fig:ratchet} illustrates the Ratchet effect across five trajectories
over 2026--2035. The two dashed/dotted paths (R2 and R5) undergo a temporary
Scenario~B shock and then fully revert to Scenario~A parameters; yet their debt
levels remain above the no-shock Scenario~A path over the extended policy horizon
shown, confirming the path-dependence result of Proposition~\ref{prop:ratchet}.

\begin{figure}[htbp]
  \centering
  \includegraphics[width=0.88\textwidth]{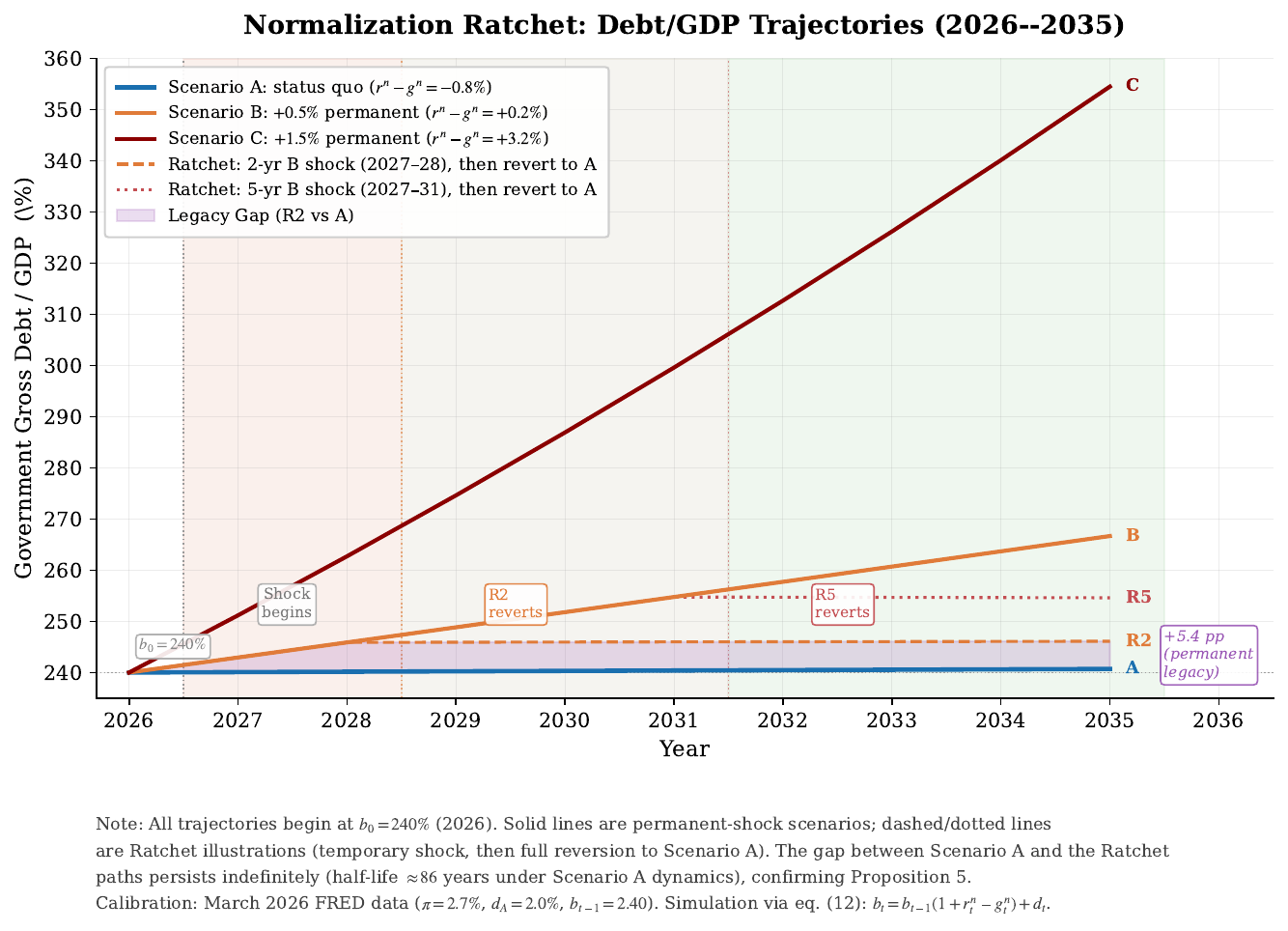}
  \caption{Normalization Ratchet: Government Gross Debt/GDP Trajectories
    (2026--2035). Solid lines show permanent-shock scenarios (A, B, C).
    The dashed line (R2) undergoes a 2-year Scenario~B shock (2027--28) then
    reverts fully to Scenario~A parameters; the dotted line (R5) undergoes a
    5-year shock (2027--31) then reverts. Despite full parameter reversion, both
    Ratchet paths remain above Scenario~A over extended policy horizons,
    confirming Proposition~\ref{prop:ratchet}: the gap decays at only
    $0.8\%$ per year (half-life $\approx 86$ years). The persistent gap at
    2035 is $+5.4$ pp (R2) and $+13.9$ pp (R5) relative to Scenario~A.
    Calibration: March 2026 FRED data ($\pi=2.7\%$ [BoJ FY2025 projection; Jan 2026 actual = 1.5\%], $d=2.0\%$, $b_0=240\%$);
    shock years use $d_t = 2.46\%$ (Scenario~B IOER arithmetic updated to \textyen{}670T GDP,
    eq.~\eqref{eq:consolidated_bc}).
    Simulation via equation~\eqref{eq:sim_transition}.}
  \label{fig:ratchet}
\end{figure}

\subsection{The Captive Financial System Parameter
  \texorpdfstring{($\varphi_t$)}{(phi)}}
\label{sec:captive}

The JFR-rg stability condition~\eqref{eq:stability_condition} requires that the Bank of
Japan can maintain $r^n_t$ below $\pi_t$. This capacity rests on the BoJ's ability to
absorb JGB supply shocks through asset purchases. That ability, in turn, depends on the
fraction $\varphi_t \in [0,1]$ of JGBs held by domestic institutions (Bank of Japan,
commercial banks, insurance companies, pension funds): when $\varphi_t$ is high, foreign
investor exit has limited price impact and the BoJ's intervention is effective; as
$\varphi_t$ falls, an endogenous sovereign risk premium $\rho_t$ emerges that raises the
effective $r^n_t$ independently of BoJ policy.

The causal structure is transparent: $\varphi_t$ is an independently observable
institutional stock variable, derivable from BoJ Flow of Funds accounts each quarter,
that \emph{precedes} and \emph{enables} the repression bias $\varepsilon_t$. The
sequence runs from institutional fact to policy instrument to model outcome:
\[
  \varphi_t \geq \bar{\varphi}
  \;\Rightarrow\; \text{BoJ retains rate control}
  \;\Rightarrow\; \varepsilon_t > 0
  \;\Rightarrow\; \text{debt compression (eq.~\eqref{eq:jfr_decomposed})}.
\]
Making this sequence explicit also clarifies the model's scope: the JFR-rg framework
characterizes the \emph{conditional} dynamics given the observed institutional
configuration, rather than claiming that the configuration is itself an equilibrium
outcome of the model.

As $\varphi_t$ falls, the endogenous sovereign risk premium $\rho_t$ rises:
\begin{equation}
  r^n_t \;\mapsto\; r^n_t + \rho(\varphi_t,\, b_{t-1}),
  \quad \frac{\partial \rho}{\partial \varphi_t} < 0, \quad
  \frac{\partial \rho}{\partial b_{t-1}} > 0.
  \label{eq:risk_premium}
\end{equation}

The modified stability condition becomes:
\begin{equation}
  (\pi_t - \varepsilon_t + \rho_t) - g^n_t \leq \frac{s_t - d_t}{b_{t-1}},
  \label{eq:stability_modified}
\end{equation}
where $g^n_t$ follows the full specification of equation~\eqref{eq:yen_growth};
the compact $g^n_t$ notation is used here for expositional clarity. The risk premium
$\rho_t$ enters as a direct leftward shift of the Stability Corridor---equivalent
to an artificial increase in $r^n_t$. The Normalization Trap is thus exacerbated by any
decline in $\varphi_t$: even without an explicit BoJ rate hike, a reduction in the
domestic holding share tightens effective financial conditions.

\begin{definition}[Critical Captive Threshold $\bar{\varphi}$]
\label{def:critical_phi}
There exists a threshold $\bar{\varphi}$ below which the JFR-rg equilibrium breaks
down regardless of BoJ policy, because $\rho(\varphi_t, b_{t-1}) > \pi_t$---that is,
the endogenous risk premium exceeds the \emph{maximum achievable} repression bias
(attained when $r^n_t \to 0$, giving $\varepsilon_t \to \pi_t$). Under this condition,
stability~\eqref{eq:stability_modified} cannot be satisfied by any feasible rate
setting. Empirical estimation of $\bar{\varphi}$ using historical sudden-stop data
(Calvo, 1998) is a priority for future research.
\end{definition}

\textbf{Japan's current trajectory.} As of 2026, $\varphi_t \approx 0.88$--$0.92$
(the Bank of Japan holds approximately 46\% of JGBs outstanding, down from a peak
of $\approx$53\% in 2022 as the tapering program reduces monthly purchases;
other domestic institutions---commercial banks, insurance companies, pension
funds---hold a further 42\%, for a combined domestic share of $\approx$88\%;
source: Ministry of Finance Flow of Funds, December 2024). The
slow decline of $\varphi_t$ as the BoJ reduces asset purchases from 2024 constitutes a
\emph{second} source of normalization risk, independent of the policy rate channel.
Each percentage point decline in $\varphi_t$ raises $\rho_t$ and shifts the Stability
Corridor adversarially. Monitoring $\varphi_t$ in real time---derivable from the BoJ's
Flow of Funds accounts---should be treated as a co-equal policy indicator alongside the
$r^n - g^n$ spread.

\subsection{Reduced-Form Micro-Foundations}
\label{sec:microfound}

The JFR-rg model is deliberately cast as a reduced-form accounting framework rather
than a fully micro-founded dynamic general equilibrium model.  This choice reflects
a specific analytical priority: all right-hand-side variables in the core accumulation
equation~\eqref{eq:jfr_core} are directly observable from FRED without reliance on
calibrated unobservables, making the framework transparent and immediately replicable.
Nevertheless, each structural equation admits a natural interpretation in terms of
optimizing behavior that is consistent with---though not derived from---standard
microfoundations.

\medskip
\noindent\textbf{Financial repression bias $\varepsilon_t$.}
In a standard portfolio-balance model (following the tradition of McKinnon, 1973, and
Shaw, 1973), domestic financial institutions face regulatory constraints that tilt their
asset allocation toward JGBs at below-market yields.  When $\varphi_t \approx 0.90$,
approximately 90\% of JGB supply is absorbed by institutions whose portfolio choice is
constrained by reserve requirements, capital regulations, and prudential mandates.
In this environment, $\varepsilon_t > 0$ is not an arbitrage opportunity but an
equilibrium outcome of the constrained optimization problem facing domestic
intermediaries: the shadow cost of the portfolio constraint equals the repression bias.
The JFR-rg model takes this equilibrium value of $\varepsilon_t$ as directly observable,
bypassing the need to specify the constraint structure explicitly.

\medskip
\noindent\textbf{Nonlinear yen channel $g^n_t = \gnomstar + \alpha\,\Delta e_t -
\beta\,\max(0,\Delta e_t - \ebar)^2$.}
The linear term $\alpha\,\Delta e_t$ corresponds to the standard expenditure-switching
effect: yen depreciation raises the yen-denominated profitability of exporting firms,
increasing investment and nominal value-added.  At the firm level, this is consistent
with a profit function that is linear in the exchange rate for small deviations.  The
quadratic penalty term captures the consumer's budget constraint: for severe
depreciation ($\Delta e_t > \ebar$), import price pass-through ($\gamma \approx
-0.020$ to $-0.050$ per JPY/USD; Haba et al., 2025) compresses real household
purchasing power, reducing consumption and nominal growth.  The threshold $\ebar$
represents the depreciation level at which the import channel dominates the export
channel in aggregate, consistent with a standard open-economy model where the
Marshall-Lerner condition holds locally but import-price inflation generates
second-order welfare losses at large exchange rate deviations.

\medskip
\noindent\textbf{Debt accumulation equation.}
The discrete-time recursion $b_t = b_{t-1}(1 + r^n_t - g^n_t) + d_t$ is the
linearized version of the government's intertemporal budget constraint, derived in
Appendix~A without behavioral assumptions.  Its reduced-form character is a feature,
not a limitation: in a captive financial system where $\varphi_t \approx 0.90$,
the sovereign risk premium $\rho_t$ is structurally suppressed
(Section~\ref{sec:captive}), so the
no-arbitrage condition that standard DSGE models use to pin down $r^n_t$ is replaced
by the institutional equilibrium condition $r^n_t = \pi_t - \varepsilon_t$.  The
JFR-rg framework characterizes the resulting debt dynamics conditional on this
institutional equilibrium, rather than explaining the equilibrium itself.  Future work
incorporating endogenous sovereign risk premia and household portfolio optimization
could provide a fully micro-founded derivation of the stability conditions
established here.

\subsection{Structural Assumptions for Identification: The Block-Recursive Architecture}
\label{sec:identification}

A recurring concern in discussions of reduced-form debt frameworks involves the
potential simultaneity of key variables: specifically, whether $\varepsilon_t$ and
$\varphi_t$ are jointly determined with the debt dynamics they purport to explain.
We clarify the role of this concern within the JFR-rg architecture, which operates
on two separable layers.

\paragraph{The Two-Layer Architecture.}
\textbf{Layer~(L1) --- Accounting:} The debt accumulation equation~\eqref{eq:sim_transition}
is a linearized accounting identity derived from the consolidated government budget
constraint without behavioral assumptions. It holds by construction for any realized
values of $r^n_t$, $g^n_t$, $b_{t-1}$, and $d_t$, regardless of how those values
were generated.

\noindent\textbf{Layer~(L2) --- Empirical:} The Local Projections, VAR, and subsample
analyses in Appendix~\ref{app:empirical_supplement} test whether the realized
historical values of $\varepsilon_t$ and $r^n_t - g^n_t$ are statistically associated
with subsequent $\Delta b_t$ in the expected direction. These tests involve
econometric identification, and the limits of their causal interpretation are
explicitly documented in Section~H.5.

The simultaneity concern applies to~(L2) but not to~(L1). Proposition~1 and
the IOER arithmetic of Proposition~4 are derived from~(L1) alone and are
unaffected by identification limitations at~(L2). Propositions~2--3 and the
growth-channel component of Proposition~4 combine~(L1) with empirically
calibrated structural parameters ($\gnomstar$, $\alpha$, $\beta$, $\ebar$);
their qualitative conclusions are robust across the calibration ranges
reported in Appendix~\ref{app:sensitivity}, but their magnitudes depend on
the calibrated parameters.

\paragraph{Structural Sequencing of the Accounting Layer.}
Under the maintained institutional conditions SC1 ($\varphi_t \geq \bar{\varphi}$)
and SC2 ($\Delta e_t \leq \bar{e}$), the period-$t$ variables
$(\varphi_t,\, \varepsilon_t,\, b_t)$ follow a strictly lower-triangular
(block-recursive) causal sequence:
\[
  \underbrace{\varphi_t}_{\text{Block 1: Institutional}}
  \;\xrightarrow{\;\text{SC1}\;}\;
  \underbrace{\varepsilon_t}_{\text{Block 2: Policy}}
  \;\xrightarrow{\;\text{identity}\;}\;
  \underbrace{b_t}_{\text{Block 3: Outcome}}
\]

\noindent\textit{Block~1:} $\varphi_t$ is determined by the portfolio stock positions
of domestic institutions, derivable each quarter from the BoJ Flow of Funds accounts.
These positions are predetermined at date~$t$ and not a direct function of
the same-period $b_t$.

\noindent\textit{Block~2:} Given $\varphi_t \geq \bar{\varphi}$, the Bank of Japan
sets $r^n_t$ via YCC conditional on observed $\pi_t$, so that
$\varepsilon_t \equiv \pi_t - r^n_t$ does not contain $b_t$ as an argument within the
period. The BoJ's formal instrument rule is
$r^n_t = f(\varphi_t,\, \pi_t,\, \text{mandate parameters})$.

\noindent\textit{Block~3:} Given Blocks~1--2, the recursion
$b_t = [1 + (r^n_t - g^n_t)]\,b_{t-1} + d_t - s_t$
determines $b_t$ uniquely from predetermined and policy-set variables, with no
same-period feedback from $b_t$ to $\varepsilon_t$ or $\varphi_t$.

{\sloppy
This lower-triangular structure also provides the institutional justification
for the Cholesky ordering
($\varepsilon_t \to (r^n_t - g^n_t) \to \Delta b_t$) used in the VAR
analysis of Appendix~H.3: both impose the same contemporaneous exclusion
restrictions and are internally consistent, though neither independently
establishes the long-run causal direction.\par}

\paragraph{Scope of the Structural Assumptions.}
The block-recursive structure holds \emph{within each period~$t$} under SC1 and SC2.
It does not preclude long-run feedback: a sustained rise in $b_t$ could eventually
erode $\varphi_t$ (failure mode~(ii)) or force $\pi_t$ upward through fiscal
dominance (failure mode~(i)). These are precisely the conditions under which the
scope conditions fail. The framework is internally consistent: SC1 and SC2 are the
conditions that preserve lower-triangularity, and the failure modes are the conditions
under which lower-triangularity is violated.

The possibility of an implicit long-run reaction function---whereby the BoJ maintains
$\varepsilon_t > 0$ partly in response to accumulated $b_{t-1}$---cannot be ruled out
at Layer~(L2). This would not invalidate the Layer~(L1) accounting propositions.
Structural identification of this channel, using YCC announcement dates as instruments,
is left for future research and is listed as a discriminating test in
Table~\ref{tab:alt_explanations}.

\section{Simulation and Counterfactual Prediction}

\subsection{Calibration and Baseline Assumptions}

We calibrate the JFR-rg model using real-time FRED point estimates as of March 2026.
The initial state ($t = 0$, Year~2026) is:
\begin{itemize}[itemsep=2pt]
  \item \textbf{Initial Debt-to-GDP ($b_0$):} 240\% (2.40) ---
        FRED:\,GGGDTAJPA188N
  \item \textbf{Nominal Interest Rate ($r^n_0$):} $\approx$2.2\% ---
        10-year JGB yield (FRED:\,IRLTLT01\allowbreak JPM156N)
  \item \textbf{Inflation ($\pi_0$):} $\approx$2.7\% ---
        CPI YoY (FRED:\,JPNCPIALLMINMEI).%
        \footnote{The $\pi_0 = 2.7\%$ value reflects the BoJ's fiscal-year 2025 CPI
        central projection and the realized 2025 calendar-year average. The January 2026
        national CPI release (Statistics Bureau, 20 Feb 2026) recorded 1.5\% YoY,
        the lowest since March 2022. Substituting $\pi_0 = 1.5\%$ yields
        $\varepsilon_0 = -0.7\%$ (positive real rate). Section~3 provides a detailed
        data-currency note; all simulations use the 2.7\% calibration, which corresponds
        to BoJ's own near-term projection. The proximity of actual inflation to 2\%
        core (Jan 2026 core = 2.0\%) and the BoJ's expectation of a transitory
        dip make this choice defensible for scenario-analysis purposes.}
  \item \textbf{Financial Repression Bias ($\varepsilon_0 = \pi_0 - r^n_0$):}
        $\approx$0.5\% --- directly computed from FRED
  \item \textbf{Nominal Growth ($g^n_0$):} $\approx$3.0\% ---
        Nominal GDP YoY (FRED:\,JPNNGDP)
  \item \textbf{Exchange Rate Baseline ($\Delta e_0$):} $\approx$0 ---
        The calibration treats the 2026 yen level as the structural exchange-rate
        equilibrium, so that $\gnomstar \approx g^n_0 = 3.0\%$ by
        equation~\eqref{eq:yen_growth}. Deviations from this assumption would shift
        $\gnomstar$ relative to $g^n_0$ by $\alpha\,\Delta e_0$, but do not affect
        the $r^n - g^n$ spread dynamics which drive all scenario comparisons.
  \item \textbf{$r^n - g^n$ Spread:} $\approx -0.8\%$ ---
        confirmed in green zone, Figure~\ref{fig:rg_spread}
  \item \textbf{Primary Deficit ($d_t$):} $\approx$2.0\% of GDP ---
        held constant under the baseline
  \item \textbf{Seigniorage ($s_t$):} $\approx$0.0\% of GDP ---
        assumed negligible for the discrete-time simulation
  \item \textbf{Captive System Parameter ($\varphi_t$):} $\approx$0.90 ---
        BoJ Flow of Funds accounts
\end{itemize}

The total annual debt compression at baseline is:
\[
  (g^n_t - r^n_t) \cdot b_{t-1} \approx (3.0\% - 2.2\%) \times 2.40 = 1.92\%
  \text{ of GDP per year,}
\]
which absorbs the vast majority of the 2.0\% primary deficit, yielding a net
accumulation of approximately 0.08\% of GDP per year. Using the simplified discrete-time
transition equation (with $s_t \approx 0$):
\begin{equation}
  b_t = b_{t-1}\,(1 + r^n_t - g^n_t) + d_t,
  \label{eq:sim_transition}
\end{equation}
we simulate the debt trajectory over 2026--2030 under three scenarios.

\subsection{Scenario A: The Status Quo --- Baseline JFR-rg Stability}
\label{sec:scenA}

The Bank of Japan maintains $\varepsilon_t > 0$ and allows mild yen flexibility.
The $r^n - g^n$ spread remains at $-0.8\%$.

\begin{table}[htbp]
\centering
\caption{Scenario A --- Baseline JFR-rg Stability (Status Quo)}
\label{tab:scenA}
\small
\resizebox{\textwidth}{!}{%
\begin{tabular}{l c c c c c c c}
\toprule
Year & Rate Hike & $r^n_t$ (\%) & $g^n_t$ (\%) & $r^n - g^n$ & $d_t$ (\%) &
  $b_t$ (\%) & Dynamics \\
\midrule
2026 & $+0.0\%$ & 2.2 & 3.0 & $-0.8\%$ & 2.0 & 240.0 & Base \\
2027 & $+0.0\%$ & 2.2 & 3.0 & $-0.8\%$ & 2.0 & 240.1 & Near-Neutral \\
2028 & $+0.0\%$ & 2.2 & 3.0 & $-0.8\%$ & 2.0 & 240.2 & Near-Neutral \\
2029 & $+0.0\%$ & 2.2 & 3.0 & $-0.8\%$ & 2.0 & 240.2 & Near-Neutral \\
2030 & $+0.0\%$ & 2.2 & 3.0 & $-0.8\%$ & 2.0 & 240.3 & Near-Neutral \\
\bottomrule
\end{tabular}%
}
\end{table}

\noindent\textit{Result:} The automatic compression of 1.92\% of GDP per year absorbs
nearly all of the 2.0\% primary deficit, leaving the debt trajectory precariously
balanced with a net accumulation of $+0.08\%$ per year ($\Delta b_t = 0.008 \times
2.40 + 2.0\% - 1.92\% = +0.08\%$). Note that successive rounding of the displayed
values in Table~\ref{tab:scenA} to one decimal place may create the appearance of
uneven increments; the underlying dynamics are uniform at $+0.08\%$ per year. The
Corridor Width of $-0.024\%$ (Property~\ref{prop:corridor_width}) is small relative
to the measurement uncertainty of the underlying data series
($\pm 0.5$ pp for GGGDTAJPA188N); the result is therefore better characterized as
operation in the \emph{immediate vicinity of the corridor boundary} than as a
confirmed outside position. This distinction does not affect the directional policy
implications, but readers should interpret the near-miss arithmetic as indicating
boundary proximity rather than a precisely calibrated shortfall.

\medskip
\noindent\textit{Long-run steady state.} The recursion~\eqref{eq:sim_transition}
has a closed-form steady state $b^* = d_t/(g^n_t - r^n_t) = 2.0\%/0.8\% =
250\%$ of GDP under Scenario~A parameters.  The current level of 240\%
therefore lies \emph{below} this long-run attractor, drifting upward at
$+0.08\%$ per year and reaching $b^* = 250\%$ in approximately 125 years.
This confirms the ``Near-Neutral'' characterization in Table~\ref{tab:scenA}:
the debt trajectory is neither self-correcting nor explosive on any
policy-relevant horizon, but quasi-stationary with a clearly defined
long-run limit.

\subsection{\texorpdfstring{Scenario B: Moderate Normalization --- $+0.5\%$ Rate Shock}{Scenario B: Moderate Normalization --- +0.5\% Rate Shock}}

The Bank of Japan hikes $r^n_t$ by 50 basis points. The repression bias shrinks
($\varepsilon_t$ falls toward zero); the yen appreciates modestly via the linear channel
($\alpha\,\Delta e_t < 0$); nominal growth drops 0.5\%. Applying the IOER arithmetic of
equation~\eqref{eq:consolidated_bc} consistently across scenarios, the primary deficit
widens to approximately 2.46\%:\footnote{The paper's canonical denominator is nominal GDP SAAR,
\textyen{}670T (FRED:\,JPNNGDP, March~2026 snapshot), because $d_t^B$ is defined
throughout as a share of GDP. Using this denominator gives an IOER contribution
of $\approx 0.37\%$ and a short-term JGB rollover contribution of $\approx 0.09\%$,
totalling $d_t^B \approx 2.46\%$ of GDP. As a mechanical robustness check, if one
instead applies the same scaling rule to both components using a tighter
denominator around the core JGB stock (\textyen{}600T), the IOER term becomes
$0.5\% \times 500/600 \approx 0.42\%$ and the rollover term
$0.5\% \times 120/600 = 0.10\%$, yielding $d_t^B \approx 2.52\%$. The
Normalization Trap conclusion is unchanged under either convention; see
Table~\ref{tab:ioer_sens}.}
\[
  d_t^B = 2.0\% + \underbrace{0.5\% \times \tfrac{\text{\textyen{}500T}}{\text{\textyen{}670T}}}_{0.37\%\;\text{IOER}} + \underbrace{0.5\% \times \tfrac{\text{\textyen{}120T}}{\text{\textyen{}670T}}}_{0.09\%\;\text{rollover}} \approx 2.46\%\text{ of GDP.}
\]

\begin{table}[htbp]
\centering
\caption{Scenario B --- Moderate Rate Hike ($+0.5\%$). Primary deficit $d_t \approx 2.46\%$
  computed via IOER arithmetic (equation~\eqref{eq:consolidated_bc};
  nominal GDP denominator \textyen{}670T, FRED:\,JPNNGDP).}
\label{tab:scenB}
\small
\resizebox{\textwidth}{!}{%
\begin{tabular}{l c c c c c c c}
\toprule
Year & Rate Hike & $r^n_t$ (\%) & $g^n_t$ (\%) & $r^n - g^n$ & $d_t$ (\%) &
  $b_t$ (\%) & Dynamics \\
\midrule
2026 & $+0.0\%$ & 2.2 & 3.0 & $-0.8\%$ & 2.0 & 240.0 & Base \\
2027 & $+0.5\%$ & 2.7 & 2.5 & $+0.2\%$ & 2.46 & 243.0 & Increase \\
2028 & $+0.5\%$ & 2.7 & 2.5 & $+0.2\%$ & 2.46 & 246.0 & Increase \\
2029 & $+0.5\%$ & 2.7 & 2.5 & $+0.2\%$ & 2.46 & 248.9 & Increase \\
2030 & $+0.5\%$ & 2.7 & 2.5 & $+0.2\%$ & 2.46 & 251.9 & Increase \\
\bottomrule
\end{tabular}%
}
\end{table}

\noindent\textit{Result:} The $r^n - g^n$ spread flips positive ($+0.2\%$). The Base
Effect now works in reverse: the $+0.2\%$ spread on a 240\% base accumulates
$0.002 \times 2.40 = 0.48\%$ of GDP per year. Combined with the IOER-corrected primary
deficit ($\approx 2.46\%$), total debt accumulation is approximately 2.94\% of GDP per year,
pushing debt to $\approx$252\% by 2030. Crucially, by the Normalization Ratchet
(Proposition~\ref{prop:ratchet}), even a subsequent full reversion to Scenario~A
parameters would leave a highly persistent debt-ratio gap: the $\approx$12
percentage-point gap accumulated by 2030 ($251.9\% - 240.0\%$) relative to the
no-shock Scenario~A baseline has a half-life of approximately 86 years
under Scenario~A dynamics.

\subsection{\texorpdfstring{The Tipping Point: Scenario B$^+$ ($+1.0\%$ Rate Shock)}{The Tipping Point: Scenario B+ (+1.0\% Rate Shock)}}
\label{sec:scenBplus}

To illustrate the nonlinear transition between Scenario~B and Scenario~C, we simulate
an intermediate $+1.0\%$ hike. At $r^n_t = 3.2\%$, the yen appreciates ($\Delta e_t
< 0$), which reduces the linear yen growth bonus ($\alpha\,\Delta e_t < 0$) and drives
$g^n_t$ downward to approximately 1.0\%. The quadratic penalty term remains inactive, as
it is defined for excessive depreciation only ($\Delta e_t > \ebar > 0$). Applying the
IOER arithmetic consistently, $d_t = 2.0\% + 1.0\% \times (500/670 + 120/670) \approx 2.93\%$.
The $r^n - g^n$ spread widens to $+2.2\%$, and with a deficit of $d_t \approx 2.93\%$, the
debt accumulation rate accelerates to approximately $0.022 \times 2.40 + 2.93\% = 8.21\%$
of GDP per year. Crossing the $r^n - g^n = 0$ threshold triggers rapidly accelerating
deterioration purely through the base effect, without requiring the yen penalty channel.
Within the model's calibration, this suggests that the transition to adverse dynamics
may occur relatively quickly following normalization, though the precise threshold
depends on parameters---particularly $\alpha$, $\gnomstar$, and the behavioral response
of the primary deficit---that require empirical estimation.

\subsection{\texorpdfstring{Scenario C: Aggressive Normalization --- The Danger Zone ($+1.5\%$ Shock)}{Scenario C: Aggressive Normalization --- The Danger Zone (+1.5\% Shock)}}
\label{sec:scenC}

Rates are raised aggressively by 150 basis points. The repression bias $\varepsilon_t$
turns negative. The yen appreciates sharply ($\Delta e_t < 0$), reducing nominal growth
to 0.5\% through the linear channel $\alpha\,\Delta e_t$.

\medskip
\noindent\textbf{The cold arithmetic of the consolidated balance sheet.}
The $+1.39$ pp increase in the primary deficit ($d_t: 2.0\% \to 3.39\%$) is not a
behavioral assumption; it follows directly from the consolidated government budget
constraint derived in equation~\eqref{eq:consolidated_bc}. As mainstream economists
have long warned, large-scale asset purchases create fiscal dominance by transforming
fixed-rate long-term liabilities (JGBs) into floating-rate overnight liabilities (BoJ
current account deposits). The accounting consequence, now quantitatively material, is:

\begin{enumerate}[itemsep=4pt]
  \item \textbf{IOER channel:} The Bank of Japan holds approximately \textyen{}450--500
    trillion in current account deposits ($R_{t-1} \approx 454$ trillion yen as of
    Q1~2026, down from a peak of $\approx$500 trillion yen prior to the 2024 tapering
    program). A $+1.5\%$ increase in the IOER rate immediately generates:
    \[
      i^R_t \cdot R_{t-1} = 1.5\% \times 500\text{ trillion} \approx 7.5\text{ trillion yen}
      \approx \frac{7.5}{670} \approx 1.12\%\text{ of GDP.}
    \]
    For consistency with the Scenario~B arithmetic in equation~\eqref{eq:consolidated_bc}
    (which was calibrated at $R_{t-1} = 500$ trillion yen as a conservative upper bound),
    note that using the Q1~2026 realized value of 454~trillion yields
    $1.5\% \times 454/670 \approx 1.01\%$; the round 500~trillion figure used throughout
    the paper gives $1.12\%$ of GDP. The Scenario~B
    primary deficit calculation uses the round 500~trillion figure as a conservative
    upper bound; adjusting to 454~trillion would reduce the IOER contribution from
    0.37\% to approximately 0.34\% and lower the total $d_t^B$ from 2.46\% to 2.43\%,
    which does not materially alter the qualitative conclusions.
    \footnotesize\textit{Upper-bound caveat:} This calculation further assumes the policy
    rate increase passes through fully to the IOER applied to all reserves. Under the
    Bank of Japan's Tiered Interest Rate Framework (TIRF, introduced 2016), different
    tranches are remunerated at different rates and the BoJ retains discretion over each
    tier. If the BoJ maintains a below-policy IOER on the bulk of excess reserves, the
    fiscal impact would be substantially smaller. The figures above therefore represent
    \emph{upper bounds} on the IOER channel; the actual impact depends on both the
    then-current reserve balance and the TIRF structure implemented alongside any rate
    increase. The sensitivity of the Normalization Trap to the full range of
    $\alpha_{\mathrm{pt}} \in [0,1]$ is examined systematically in
    Section~\ref{sec:ioer_robustness} below.\normalsize
  \item \textbf{Short-term JGB rollover channel:} Approximately \textyen{}120--150
    trillion in Treasury discount bills and other floating-rate instruments roll over
    at the higher rate, generating an additional:
    \[
      1.5\% \times 120\text{ trillion} \approx 1.8\text{ trillion yen}
      \approx \frac{1.8}{670} \approx 0.27\%\text{ of GDP.}
    \]
  \item \textbf{Total mechanical deficit expansion:}
    \[
      \Delta d_t \approx 1.12\% + 0.27\% \approx 1.39\%\text{ of GDP,}
    \]
    producing the $d_t \approx 3.39\%$ in Table~\ref{tab:scenC}
    (computed as $2.0\% + 1.5\% \times \frac{500+120}{670} \approx 3.39\%$ via
    equation~\eqref{eq:consolidated_bc}).
\end{enumerate}

This is the fiscal dominance mechanism that orthodox economics has warned about for
decades---not as a theoretical possibility, but as a present arithmetic reality
activated the moment the policy rate rises. The consolidated balance sheet, once
constructed at QQE scale, makes the interest rate a fiscal parameter of the first
order. Scenario~C does not assume a behavioral fiscal deterioration; it simply applies
equation~\eqref{eq:consolidated_bc} to the observed stock of reserves.

Tax revenues plummet under yen appreciation and growth compression; the primary deficit
rises to 3.39\% of GDP via the combined mechanical and cyclical channels described above.

\begin{table}[htbp]
\centering
\caption{Scenario C --- Aggressive Normalization --- The Danger Zone ($+1.5\%$).
  \textbf{Upper-bound stress-test calibration} ($\Delta e_t = -50$ JPY/USD per year).
  See Table~\ref{tab:scenC_alt} for the baseline-consistent calibration ($\Delta e_t = -20$).}
\label{tab:scenC}
\small
\resizebox{\textwidth}{!}{%
\begin{tabular}{l c c c c c c c}
\toprule
Year & Rate Hike & $r^n_t$ (\%) & $g^n_t$ (\%) & $r^n - g^n$ & $d_t$ (\%) &
  $b_t$ (\%) & Dynamics \\
\midrule
2026 & $+0.0\%$ & 2.2 & 3.0 & $-0.8\%$ & 2.00 & 240.0 & Base \\
2027 & $+1.5\%$ & 3.7 & 0.5 & $+3.2\%$ & 3.39 & 251.1 & Rapid Accumulation \\
2028 & $+1.5\%$ & 3.7 & 0.5 & $+3.2\%$ & 3.39 & 262.5 & Rapid Accumulation \\
2029 & $+1.5\%$ & 3.7 & 0.5 & $+3.2\%$ & 3.39 & 274.3 & Severe Accumulation \\
2030 & $+1.5\%$ & 3.7 & 0.5 & $+3.2\%$ & 3.39 & 286.4 & Severe Accumulation \\
\bottomrule
\end{tabular}%
}
\end{table}

\medskip
\noindent\footnotesize\textit{Note on calibration:} The $g^n_t = 0.5\%$ value corresponds to
$\Delta e_t \approx -50$ JPY/USD via equation~\eqref{eq:yen_growth}, representing a
historically severe appreciation shock (cf.\ the full Abenomics move of $+70$ JPY/USD
over ten years). This table serves as an upper-bound stress test. The
baseline-consistent calibration with $\Delta e_t = -20$ JPY/USD (Appendix~E Panel~A)
is presented in Table~\ref{tab:scenC_alt} and independently confirms the Normalization
Trap ($b_{2030} = 270.7\%$). A 4$\times$4 sensitivity grid spanning $\gamma \in
\{-0.010,\ldots,-0.050\}$ and $|\Delta e| \in \{10,\ldots,50\}$ JPY/USD confirms the
Trap in all 16 combinations (Appendix~E, Table~\ref{tab:grid_gamma_deltae}).
\textit{Terminology:} ``Rapid Accumulation'' denotes a $r^n - g^n$ spread of $+3.2\%$
generating approximately 7.7\% of GDP in additional annual interest-growth debt
accumulation ($0.032 \times 240\%$); ``Severe Accumulation'' denotes the same rate
compounding over a longer horizon. Neither term implies mathematical divergence
to infinity; both describe an accelerating trajectory whose practical significance
derives from the scale of the debt base.

\medskip
\normalsize

\begin{table}[htbp]
\centering
\caption{Scenario C (Alt) --- Aggressive Normalization, \textbf{Baseline-Consistent}
  Calibration ($\Delta e_t = -20$ JPY/USD per year, $\gamma = -0.020$).
  $\pi_t$ derived endogenously: $\pi_t = \pi_{\text{base}} + \gamma \times
  \text{cumulative}\,\Delta e_t$. The Normalization Trap is confirmed: $b_{2030} = 270.7\%$.}
\label{tab:scenC_alt}
\small
\resizebox{\textwidth}{!}{%
\begin{tabular}{l c c c c c c c}
\toprule
Year & $r^n_t$ (\%) & $\pi_t$ (\%) & $\varepsilon_t$ (\%) & $g^n_t$ (\%) &
  $r^n - g^n$ (\%) & $d_t$ (\%) & $b_t$ (\%) \\
\midrule
2026 & 2.2 & 2.70 & $+0.50$ & 3.0 & $-0.80$ & 2.00 & 240.0 \\
2027 & 3.7 & 3.10 & $-0.60$ & 2.0 & $+1.70$ & 3.39 & 247.5 \\
2028 & 3.7 & 3.50 & $-0.20$ & 2.0 & $+1.70$ & 3.39 & 255.1 \\
2029 & 3.7 & 3.90 & $+0.20$ & 2.0 & $+1.70$ & 3.39 & 262.8 \\
2030 & 3.7 & 4.30 & $+0.60$ & 2.0 & $+1.70$ & 3.39 & 270.7 \\
\bottomrule
\end{tabular}%
}
\end{table}

\noindent\footnotesize\textit{Note:} $\pi_t$ is derived from the import price pass-through
equation: $\pi_t = 2.7\% + \gamma \times (\text{cumulative}\;\Delta e_t)$, where
$\gamma = -0.020$ (BOJ Q-JEM lower-bound estimate; Haba et al., 2025, Bank of Japan Working Paper No.\
25-E-2) and $\Delta e_t = -20$ JPY/USD per year. Year-by-year: $\pi_{2027} = 2.7 +
(-0.020)\times(-20) = 3.10\%$; $\pi_{2028} = 2.7 + (-0.020)\times(-40) = 3.50\%$;
$\pi_{2029} = 3.90\%$; $\pi_{2030} = 4.30\%$. Under yen appreciation, import prices
fall but domestic wage and energy cost pressures sustain positive pass-through. $d_t = 2.0\% + 1.39\% = 3.39\%$ GDP (IOER arithmetic with \textyen{}670T denominator, eq.~\eqref{eq:consolidated_bc};
consistent with Table~\ref{tab:scenC}). This calibration is internally consistent with Appendix~E
Panel~A and with the $\pi_t$ derivation in Appendix~B (Table~\ref{tab:appendB}).

\normalsize

\noindent The terms \textit{Rapid Accumulation} and \textit{Severe Accumulation} used in
Table~\ref{tab:scenC} describe a $+3.2\%$ $r^n - g^n$ spread that causes the debt ratio
to grow geometrically at approximately 7.7\% of GDP per year from the interest-growth
channel alone ($0.032 \times 240\%$). This is not explosive dynamics in the mathematical
sense of divergence to infinity, but an accelerating trajectory whose practical
significance derives from the scale of the base: each percentage point of debt-to-GDP
represents approximately \textyen{}6 trillion at 2026 nominal GDP levels.

\noindent\textit{Result:} The $+3.2\%$ $r^n - g^n$ spread generates
$0.032 \times 2.40 = 7.68\%$ of GDP in additional interest-growth debt accumulation per
year (11.07\% total including the 3.39\% primary deficit). Debt accelerates toward 286\%
within five years, confirming the Normalization Trap (Proposition~\ref{prop:norm_trap}).
The Normalization Ratchet implies that any reversion to Scenario~A after 2030 would face
a starting debt of $\approx$286\%---approximately 46 percentage points above the 2026
baseline ($286.4\% - 240.0\%$)---with an 86-year half-life for gap closure.

The $\Delta e_t = -50$ JPY/USD assumption underlying $g^n_t = 0.5\%$ in
Table~\ref{tab:scenC} represents an upper-bound stress-test calibration---one that
exceeds the largest single-year appreciation episodes on historical record (Plaza Accord
1985: $\approx -40$ JPY/USD in the year of the agreement, 240$\to$200 by end-1985,
reaching $-88$ JPY/USD through end-1986; 1994--95 episode: $\approx -25$ JPY/USD over 18 months).\footnote{The Plaza Accord was signed on 22 September 1985 with the yen at approximately
240 JPY/USD. By end-1985 the rate had moved to $\approx$200 (a single-calendar-year change
of $\approx -40$ JPY/USD), and to $\approx$153 by end-1986. The original G5 target was
a more modest 10--12\% depreciation of the dollar ($\approx-24$ to $-29$ JPY/USD), but the
intervention substantially overshot. The $\Delta e_t = -50$ JPY/USD stress-test value in
Table~\ref{tab:scenC} therefore exceeds the original G5 target but does not exceed the
\emph{realized} single-year move (1985 calendar year: $\approx-40$ JPY/USD); the claim
that it is an ``upper-bound'' calibration should therefore be read relative to
\emph{policy intent} rather than historical maxima.}
Table~\ref{tab:scenC_alt} presents the baseline-consistent calibration ($\Delta e_t =
-20$ JPY/USD, consistent with Appendix~E Panel~A), which yields $b_{2030} = 270.7\%$
and $r^n - g^n = +1.7\%$ throughout---confirming that the Normalization Trap is not an
artifact of extreme exchange-rate assumptions. A 4$\times$4 sensitivity grid
(Section~\ref{app:sensitivity}, Table~\ref{tab:grid_gamma_deltae}) confirms the Trap in
all 16 ($\gamma \times |\Delta e|$) parameter combinations.

\subsection{Robustness to IOER Pass-Through Rate}
\label{sec:ioer_robustness}

Section~\ref{sec:scenC} noted that the $\Delta d_t \approx +1.39\%$ of GDP
fiscal expansion embedded in Scenario~C represents an \emph{upper bound},
because the Bank of Japan's Tiered Interest Rate Framework (TIRF) may shield
a portion of the reserve stock from full policy-rate pass-through.  To quantify
the sensitivity of the Normalization Trap to this uncertainty, we introduce a
pass-through parameter $\alpha_{\mathrm{pt}} \in [0,\,1]$, where $\alpha_{\mathrm{pt}} = 1$
corresponds to the full IOER arithmetic of equation~\eqref{eq:consolidated_bc}
and $\alpha_{\mathrm{pt}} = 0$ corresponds to complete TIRF shielding (zero marginal
reserve remuneration).  The primary deficit under Scenario~C then becomes:
\[
  d_t^C(\alpha_{\mathrm{pt}}) \;=\; d_{\text{base}} + \alpha_{\mathrm{pt}} \times
  \frac{r_{\text{shock}}\,(R_{t-1} + \text{TB})}{Y_t},
\]
where $r_{\text{shock}} = 1.5\%$, $R_{t-1} = 500\text{T}$, $\text{TB} = 120\text{T}$,
and $Y_t = 670\text{T}$, giving $d_t^C(1) = 3.39\%$ and $d_t^C(0) = 2.00\%$.

Table~\ref{tab:ioer_sens} reports the resulting $b_{2030}$ trajectory for five
values of $\alpha_{\mathrm{pt}}$, holding all other Scenario~C parameters at
baseline ($r^n = 3.7\%$, $g^n = 2.0\%$, $\Delta e_t = -20$ JPY/USD/yr).

\begin{table}[htbp]
\centering
\caption{Scenario~C: $b_{2030}$ Sensitivity to IOER Pass-Through Rate
  $\alpha_{\mathrm{pt}}$.  Baseline-consistent calibration: $+1.5\%$ rate hike,
  $\Delta e_t = -20$ JPY/USD/yr.  All computations via equation~\eqref{eq:sim_transition};
  $b_0 = 240\%$.}
\label{tab:ioer_sens}
\small
\begin{tabular}{l c c c c}
\toprule
Pass-through $\alpha_{\mathrm{pt}}$ & $d_t^C$ (\%) & $b_{2030}$ (\%) &
  Gap vs.\ Scenario~A (pp) & Trap? \\
\midrule
$0\%$ (full TIRF shielding)  & 2.00 & 265.0 & $+24.6$ & \textbf{YES} \\
$25\%$                        & 2.35 & 266.4 & $+26.1$ & \textbf{YES} \\
$50\%$                        & 2.69 & 267.8 & $+27.5$ & \textbf{YES} \\
$75\%$                        & 3.04 & 269.2 & $+28.9$ & \textbf{YES} \\
$100\%$ (full pass-through)  & 3.39 & 270.6 & $+30.3$ & \textbf{YES} \\
\midrule
Scenario~A (no hike)         & 2.00 & 240.3 & \textemdash & \textemdash \\
\bottomrule
\end{tabular}
\end{table}

Three findings emerge from Table~\ref{tab:ioer_sens}.  First, the Normalization
Trap ($b_{2030} > b_0 = 240\%$) holds for \emph{every} value of
$\alpha_{\mathrm{pt}}$, including the extreme case of zero IOER pass-through.
Under complete TIRF shielding, the 2030 debt ratio reaches approximately
$265\%$---a gap of $+24.6$ percentage points relative to Scenario~A, compared
with $+30.3$ percentage points under full pass-through.  Second, the variation
across the full $[0,\,1]$ range of $\alpha_{\mathrm{pt}}$ amounts to only
$5.7$ percentage points in $b_{2030}$ (265.0\% to 270.6\%), a modest spread
relative to the $+24.6$--$+30.3$ pp gap over Scenario~A.  Third, the structural
driver of the Trap is not the IOER channel at all but the reversal of the
$r^n - g^n$ spread from $-0.8\%$ to $+1.7\%$: even with $d_t^C = 2.0\%$
(identical to Scenario~A), the base effect operates in the \emph{adverse}
direction, generating $0.017 \times 2.40 = 4.08\%$ of GDP in annual
\emph{accumulation} rather than compression.

This finding directly addresses the upper-bound caveat noted in
Section~\ref{sec:scenC}: the Normalization Trap is \emph{structurally
independent} of the IOER arithmetic and is driven entirely by the $r^n - g^n$
dynamic.  The IOER channel affects the \emph{magnitude} of the trap---by at most
$5.7$ percentage points in $b_{2030}$ across the full pass-through range---but
not its existence or directional sign.  Figure~\ref{fig:ioer_sens} illustrates
the full $\alpha_{\mathrm{pt}}$ sweep and confirms that the Scenario~A and
Scenario~C fan charts do not overlap at any horizon, irrespective of TIRF
assumptions.

\begin{figure}[htbp]
  \centering
  \includegraphics[width=\textwidth]{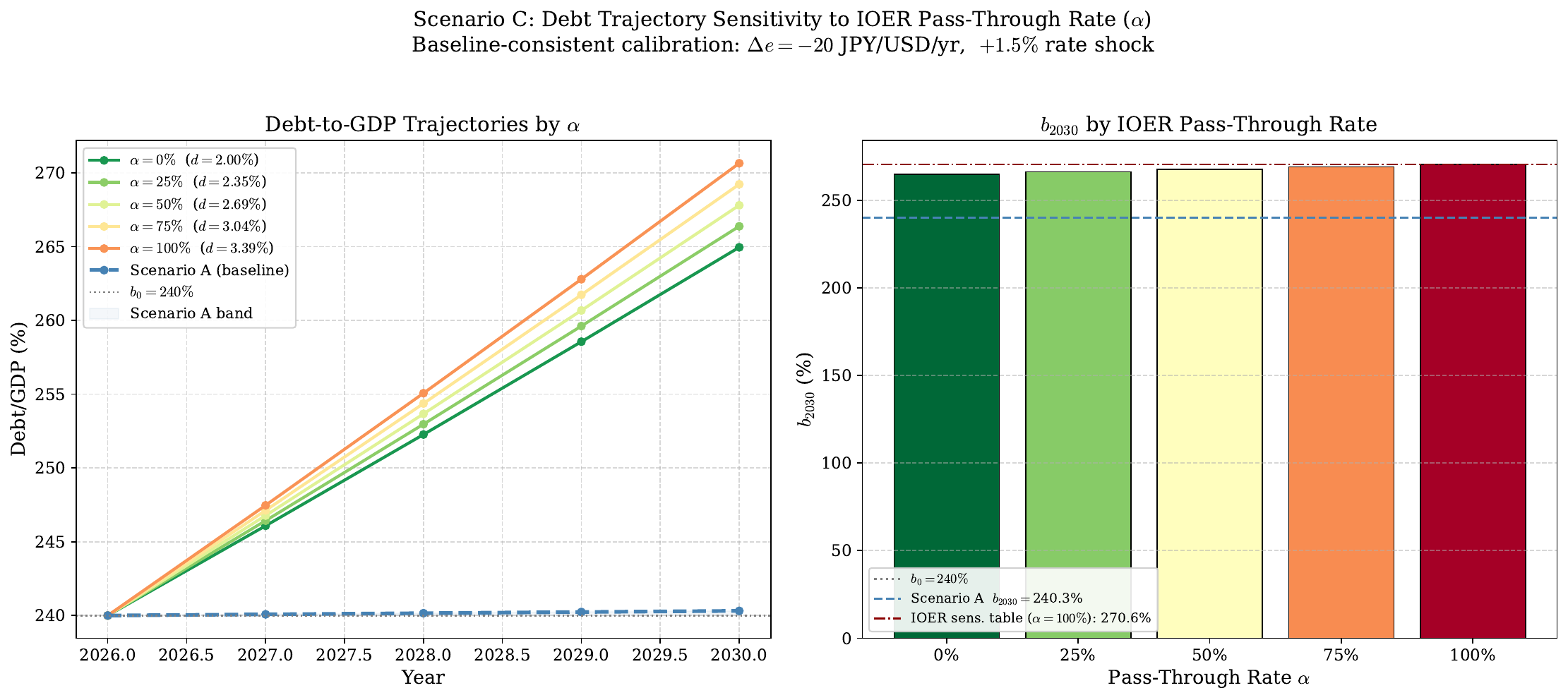}
  \caption{Scenario~C: Debt Trajectory Sensitivity to IOER Pass-Through Rate
    $\alpha_{\mathrm{pt}}$.
    \textit{Left panel:} Debt/GDP trajectories 2026--2030 for five values of
    $\alpha_{\mathrm{pt}} \in \{0\%, 25\%, 50\%, 75\%, 100\%\}$ (warm colors,
    top to bottom) and Scenario~A (blue dashed).
    \textit{Right panel:} $b_{2030}$ as a function of $\alpha_{\mathrm{pt}}$
    (bars), with the Scenario~A upper bound (blue dashed) and the paper's
    baseline Table~\ref{tab:scenB} value (dark red dot-dash).
    The Normalization Trap ($b_{2030} > 240\%$) is confirmed at every value of
    $\alpha_{\mathrm{pt}}$; the spread across the full TIRF uncertainty range
    is only $5.7$ percentage points.  Calibration: $+1.5\%$ rate shock,
    $\Delta e_t = -20$ JPY/USD/yr, $b_0 = 240\%$, equation~\eqref{eq:sim_transition}.}
  \label{fig:ioer_sens}
\end{figure}

\subsection{Parametric Uncertainty: Fan Chart}
\label{sec:fanchart}

The three baseline scenarios treat parameter values as point estimates.
Table~\ref{tab:fanchart} quantifies the sensitivity of the 2030 debt ratio to
variation in the three most uncertain parameters: the financial repression bias
$\varepsilon_t$ (which determines $r^n_t$), structural potential growth $\gnomstar$,
and the yen-growth elasticity $\alpha$.  The bounds are derived directly from the
sensitivity analysis in Appendix~E (Panels~A--C) and require no additional assumptions.

\begin{table}[htbp]
\centering
\caption{Parametric Uncertainty: $b_{2030}$ Fan Chart (\% GDP).
  Each row shows the low/baseline/high outcome under $\pm 1$ standard deviation
  perturbation of the indicated parameter, holding all others at baseline.
  ``Low'' and ``High'' denote unfavorable and favorable outcomes respectively.
  Baseline values: $\varepsilon_t = 0.50\%$, $\gnomstar = 3.0\%$,
  $\alpha = 0.050$; $b_0 = 240\%$; simulation via
  equation~\eqref{eq:sim_transition}.}
\label{tab:fanchart}
\small
\resizebox{\textwidth}{!}{%
\begin{tabular}{l l c c c}
\toprule
Scenario & Parameter varied & Low $b_{2030}$ & Baseline $b_{2030}$ & High $b_{2030}$ \\
\midrule
\multirow{2}{*}{A (Status Quo)}
  & $\varepsilon_t$: $0.25\%$ to $0.75\%$ & 242.7 & 240.3 & 238.0 \\
  & $\gnomstar$: $2.5\%$ to $3.5\%$       & 245.1 & 240.3 & 235.6 \\
\addlinespace
\multirow{2}{*}{B (Moderate, $+0.5\%$)}
  & $\varepsilon_t$ channel               & 254.3 & 251.8 & 249.4 \\
  & $\gnomstar$ channel                   & 256.7 & 251.8 & 246.9 \\
\addlinespace
\multirow{2}{*}{C (Aggressive, $+1.5\%$)}
  & $\alpha$: BOJ low ($0.013$) to baseline ($0.050$) & 270.7 & 270.7 & 263.1 \\
  & $\gnomstar$: $\pm 0.5$ pp around baseline         & 275.8 & 270.7 & 265.5 \\
\bottomrule
\end{tabular}%
}
\end{table}

\noindent\textit{Note on directional convention.}  For Scenario~C, a lower $\alpha$
(BOJ low) produces a \emph{lower} $b_{2030}$ because weaker yen-growth sensitivity
reduces the growth penalty from yen appreciation, partially offsetting the rate hike
damage.  Conversely, a lower $\gnomstar$ worsens the trajectory.

\medskip
\noindent Three structural findings emerge from Table~\ref{tab:fanchart}.  \textbf{First},
the Scenario~A fan is narrow: even under the most unfavorable parameter combination
($\gnomstar = 2.5\%$, $b_{2030} = 245.1\%$), the debt trajectory remains quasi-stationary
and far below the Scenario~B and~C ranges.  \textbf{Second}, the Scenario~B and~C fans
do not overlap with the Scenario~A fan: the gap between the Scenario~B lower bound
(249.4\%) and the Scenario~A upper bound (245.1\%) is 4.3 percentage points, confirming
that the Normalization Trap is a robust feature of the parameter space rather than a
knife-edge result.  \textbf{Third}, under Scenario~C even the most favorable parameter
combination ($\alpha = 0.013$, $b_{2030} = 263.1\%$) exceeds the Scenario~A upper bound
by 18 percentage points---a structural separation that holds across the entire BOJ
Q-JEM literature range for $\alpha$ and the plausible range for $\gnomstar$.

\section{Policy Implications and Robustness}
\label{sec:policy}

The policy implications derived in this section reflect the internal logic of the
JFR-rg model under its specific assumptions. They are not intended to prescribe a
single course of action or to narrow the range of legitimate policy perspectives.
Monetary policy involves distributional judgments, institutional constraints, and
social values that lie beyond any single analytical framework: the consequences of
financial repression fall differently on savers, borrowers, pension funds, and future
generations, and those distributional dimensions are outside the model's scope. The
model's contribution is to make certain fiscal and monetary trade-offs explicit and
quantifiable; how those trade-offs should be weighed is properly a matter for
democratic deliberation and empirical research that extends well beyond the present
analysis.

\subsection{On Distributional Scope: Why Systemic Stability Precedes
  Redistribution in Japan's Specific Context}
\label{sec:distribution}

The preceding disclaimer is analytically necessary but requires qualification
in Japan's specific institutional and macroeconomic context.  Four considerations,
taken together, suggest that the distributional critique---while important in
principle---does not constitute a first-order objection to the policy framework
proposed here.

\medskip
\noindent\textbf{(i) Japan's structural deflationary context and the cost-push
nature of recent inflation.}
The heterogeneous-agent literature's standard concern---that sustained financial
repression erodes the real wealth of net savers through demand-driven
inflation---presupposes an inflationary environment generated by domestic
demand excess.  Japan's macroeconomic history over the past three decades is
precisely the opposite: chronic deflation, a liquidity trap, and successive
episodes of premature fiscal tightening that amplified rather than resolved the
balance-sheet recession (Koo, 2003).  The post-2022 inflationary episode, while
real, is predominantly cost-push in character---driven by yen depreciation and
global commodity price shocks rather than by domestic wage-price dynamics.
Under cost-push inflation, $\varepsilon_t > 0$ reflects an imported price shock
rather than a deliberate transfer from bondholders to the sovereign; the
distributional incidence is accordingly diffuse and empirically modest compared
with the wage-price spiral scenarios that motivate HANK-style redistribution
concerns.

\medskip
\noindent\textbf{(ii) The captive financial system as a distributional buffer.}
A second consideration specific to Japan's institutional structure is that
the financial repression channel operates almost entirely through institutional
intermediaries rather than direct household bond holdings.  With
$\varphi_t \approx 0.88$--$0.92$, the JGB stock is held predominantly by the
Bank of Japan, commercial banks, insurance companies, and pension funds.
The BoJ remits profits to the Treasury, closing the distributional loop within
the public sector.  Banks, insurers, and pension funds transmit any suppressed
yield to households only indirectly---through deposit rates, insurance
premia, and pension returns that adjust slowly and partially.
This institutional architecture structurally attenuates the direct household
wealth effect that drives redistribution concerns in models where households
hold sovereign bonds directly.  It is consistent with the empirical observation
that Japan's income inequality, while not the lowest in the OECD, has remained
relatively stable throughout the period of financial repression: the disposable
income Gini has shown limited trend deterioration since 2010, and top income
concentration (the shares held by the top 1\% and top 10\% of earners) remains
lower than in the United States and at levels comparable to major European
economies (Mikayama et al., 2023; IMF, 2024b, \textit{Selected Issues}).
This empirical
record is consistent with the view that the distributional transmission channel
of financial repression is structurally attenuated in Japan's captive-system
configuration relative to what theory applied to other economies would predict.

\medskip
\noindent\textbf{(iii) The intergenerational reversal argument.}
The standard distributional critique focuses on within-generation transfers:
current savers bear a real return penalty.  However, the appropriate
counterfactual is not a world without financial repression but a world in which
the Normalization Trap (Proposition~\ref{prop:norm_trap}) has been triggered.
Under Scenario~C, the debt ratio reaches approximately 270--286\% of GDP by
2030, with a Normalization Ratchet half-life of 86 years
(Proposition~\ref{prop:ratchet}).  The cohorts bearing the largest share of
that accelerated debt accumulation are \emph{future} generations---precisely
those on whose behalf the distributional critique is most frequently invoked.
Sustained financial repression, by preventing debt explosion under the current
captive-system configuration, thus constitutes an implicit intergenerational
transfer \emph{toward} future generations rather than away from them.  This
reversal of the standard distributional narrative is not a claim that repression
is costless; it is a claim that the relevant comparison requires specifying
the counterfactual, and that the counterfactual most consistent with the
JFR-rg model is considerably more damaging to future-generation welfare than
continued repression at the current calibration.

\medskip
\noindent\textbf{(iv) Sequencing: systemic stability as a prerequisite for
redistribution.}
A broader methodological point follows from the three preceding observations.
The analytical tools of heterogeneous-agent general equilibrium
are most powerful when applied to economies that have resolved---or are not
facing---a first-order systemic stability question.  Japan's current
configuration, operating at the margin of the Stability Corridor with a 240\%
debt base, does not yet satisfy that prerequisite.  Distributional optimization
inside an unstable debt trajectory is analogous to optimizing cabin comfort on
a structurally compromised aircraft: the prior question of structural integrity
must be addressed first.  Once a structural upward shift in $\gnomstar$ is
confirmed---through sustained TFP-enhancing investment of the kind described in
Section~\ref{sec:captive}---and the debt ratio is placed on a firmly declining
path, the redistribution of the resulting Repression Dividend becomes both
analytically tractable and democratically legitimate.  The JFR-rg framework
identifies the conditions under which that sequencing becomes feasible; it does
not foreclose the distributional agenda, but assigns it its proper temporal
priority.

\medskip
\noindent\textbf{Implication for FTPL-based objections.}
The same sequencing logic applies to Fiscal Theory of the Price Level
(FTPL) critiques (Leeper, 1991; Sims, 1994; Woodford, 1994, 1995).
The FTPL predicts that unbacked fiscal expansion eventually forces a
discontinuous price-level jump.  In the JFR-rg framework, three features of
Japan's institutional configuration structurally delay this trigger: the captive
financial system ($\varphi_t \approx 0.90$) suppresses the sovereign risk
premium that would precipitate a sudden stop; the cost-push character of
current inflation means that observed $\varepsilon_t > 0$ does not reflect
the unchecked fiscal dominance the FTPL requires; and the nonlinear yen
channel provides an exchange-rate adjustment valve that absorbs fiscal shocks
without a discrete domestic price jump.  Critically, triggering the
Normalization Trap---the policy action most consistent with pre-empting FTPL
dynamics---would itself generate the debt explosion and distributional
deterioration that the FTPL warns against, but through the mechanical
arithmetic of the consolidated balance sheet rather than through a
confidence discontinuity.  The JFR-rg framework therefore does not deny the
FTPL's eventual relevance; it identifies the specific institutional threshold
($\varphi_t < \bar{\varphi}$) at which FTPL dynamics would dominate, and
argues that managing that threshold---rather than premature
normalization---is the appropriate policy response.

\subsection{Rethinking Monetary Normalization: The
  \texorpdfstring{$r^n < \pi$}{rn < pi} Imperative}

For a hyper-indebted economy with $b_{t-1} \geq 240\%$, the policy interest rate
functions as a fiscal parameter of the first order, not merely a macroeconomic dial
for inflation control. The JFR-rg stability condition suggests that the Bank of Japan
faces a meaningful trade-off: normalization that pushes $r^n - g^n$ positive may
generate debt accumulation that offsets any inflation-control benefit, at least over
the medium term. The model does not prescribe a specific rate ceiling, as the optimal
path depends on parameters---particularly $\gnomstar$, $\alpha$, and $\bar{e}$---that
require empirical estimation. It does, however, suggest that the $r^n - g^n$ spread
warrants explicit monitoring as a co-equal policy indicator alongside inflation and
output, and that normalization should proceed cautiously until a structural upward
shift in potential nominal growth $\gnomstar$ is confirmed through multiple quarters
of data. Furthermore, any normalization path should account for the path-dependent
legacy effects established in Proposition~\ref{prop:ratchet}.

\subsection{Strategic Allocation of the ``Repression Dividend''}

The negative $r^n - g^n$ spread generates the ``Repression Dividend''---approximately
1.92\% of GDP per year in total automatic debt compression at baseline. As established
in Section~\ref{sec:corridor}, the March 2026 operating point lies within the
approximation error band of the Stability Corridor boundary ($|W_{2026}| \approx
0.024\%$ versus approximation error $\approx 0.056\%$; Figure~9); the system should
therefore be interpreted as operating in the \emph{immediate vicinity} of the corridor
boundary rather than at a precisely calibrated distance from it.

The policy case for directing available fiscal space toward investment that raises
$\gnomstar$ does not rest on this numerical proximity, which lies within measurement
uncertainty. It rests instead on three independent grounds that are robust to
the approximation error:

\begin{enumerate}[itemsep=4pt]
  \item \textbf{Asymmetric cost of inaction (Normalization Ratchet).}
    Proposition~\ref{prop:ratchet} establishes that any episode of
    $r^n - g^n > 0$ generates a debt legacy with an 86-year half-life under
    Scenario~A dynamics. The cost of drifting into the Danger Region is
    therefore highly asymmetric: a multi-year normalization shock cannot be
    reversed on any policy-relevant horizon. This asymmetry motivates
    precautionary investment in $\gnomstar$ before the corridor boundary
    is breached, regardless of the precise current distance from it.

  \item \textbf{Dynamic deterioration of the captive system ($\varphi_t$ trend).}
    Section~\ref{sec:captive} documents that $\varphi_t$ has declined from
    $\approx 0.92$ (2022 peak) to $\approx 0.88$ (March 2026) as the Bank of
    Japan reduces monthly asset purchases. Each percentage point decline in
    $\varphi_t$ raises the endogenous risk premium $\rho_t$ and shifts the
    Stability Corridor adversarially---an independent tightening of effective
    financial conditions that operates even without explicit rate hikes. This
    slow-moving deterioration constitutes the most observable real-time risk
    indicator for the JFR-rg framework.

  \item \textbf{Scenario~B already in the Danger Region.}
    Simulations show that a moderate $+0.5\%$ rate shock (Scenario~B) places
    the system firmly in the Danger Region ($r^n - g^n = +0.2\%$, $b_{2030}
    \approx 252\%$), regardless of where the current operating point sits
    relative to the corridor boundary. Investment that raises $\gnomstar$
    widens the safety margin against this scenario.
\end{enumerate}

Candidates for $\gnomstar$-enhancing investment include the 61 critical and strategic
technologies identified by the Japanese government: advanced AI, quantum computing,
next-generation semiconductor fabrication, robotics, and green energy systems. Such
investment could raise $\gnomstar$ through TFP enhancement, attract private capital,
and reshore supply chains that reduce external vulnerability to the yen depreciation
channel. Whether this investment allocation is politically and institutionally feasible
is beyond the scope of the present model, but the directional imperative is clear
from condition~\eqref{eq:stability_condition}.

\subsection{The Yen as a Structural Shock Absorber}

Within the stability window ($\Delta e_t \leq \ebar$), yen depreciation is a vital
stabilizing force operating through the linear channel $\alpha\,\Delta e_t > 0$.
The post-2022 increase in the USD/JPY--CPI correlation suggests the system is operating
closer to $\ebar$ than during 2013--2021. Empirically, sustained levels above
155--160 JPY/USD appear to be associated with the penalty zone, based on the
post-2022 CPI pass-through data. Under the JFR-rg framework's calibration, the
stability window is approximately 140--155 JPY/USD, though the precise bounds of
$\ebar$ require empirical estimation from micro-level import price data.

\subsection{Monitoring the Captive System Parameter}

The $\varphi_t$ parameter is not merely a theoretical construct---it is a directly
observable leading indicator derivable quarterly from the Bank of Japan's Flow of Funds
accounts. As the BoJ reduces asset purchases post-2024, $\varphi_t$ is declining
gradually. Policy makers should treat a $\varphi_t$ reading below 0.85 as an
illustrative early-warning level for approaching the zone in which the captive
condition may begin to weaken, triggering a reassessment of the normalization pace.
This figure is consistent with the qualitative threshold implied by
Hoshi and Ito (2014) and represents a conservative lower bound given
Japan's post-2013 institutional configuration; formal estimation of
$\bar{\varphi}$ from historical sudden-stop episodes remains a research
priority (see the concluding section, failure mode~(ii)). This recommendation elevates
$\varphi_t$ monitoring to the same operational priority as the
$r^n - g^n$ spread dashboard.

\subsection{Robustness to Demographic Shock}

Japan's working-age population is declining at approximately 0.5--0.8\% per year,
threatening to reduce $\gnomstar$ toward zero. The JFR-rg stability condition remains
satisfied as long as capital deepening and TFP improvement offset labor contraction.
This is achievable through strategic AI and robotics investment. The Repression Dividend
is the financial mechanism to fund precisely this labor-capital substitution before the
window of engineered stability narrows further. Whether this path is institutionally
and politically achievable, and whether its social costs are equitably distributed,
are questions that the JFR-rg model identifies as urgent but cannot itself resolve.

\section{Conclusion}

This paper has formalized a set of empirical regularities in the Japanese economy into
the JFR-rg framework. Driven by real-time FRED data spanning 2013 to March 2026, the
analysis suggests that Japan's debt-to-GDP ratio of 240\% may be better understood as
a mathematical lever than as an imminent fiscal cliff---provided that the specific
institutional and policy conditions identified here are maintained. Four analytical
contributions stand out.

\textbf{First}, the redefined financial repression bias $\varepsilon_t \equiv \pi_t -
r^n_t \in \mathbb{R}$ is directly observable from FRED, eliminating reliance on
unobservable natural rate estimates. At current parameter values, the total annual
compression of 1.92\% absorbs the vast majority of a 2.0\% structural primary deficit,
leaving the debt trajectory operating at or near the boundary of the Stability Corridor.
Given the $\pm 0.5$ pp measurement uncertainty of the underlying data series, the
Corridor Width of approximately $-0.024\%$ indicates boundary proximity rather than a
confirmed outside position. This structural proximity makes the deployment of the
Repression Dividend into $\gnomstar$-enhancing investment a motivated policy priority.

\textbf{Second}, the \emph{Debt Sustainability Corridor}---the locus of
$(\varepsilon_t, \gnomstar)$ pairs consistent with debt stabilization under
exchange-rate-neutral conditions---provides a geometric characterization of policy
space. The corridor has three key properties: it \emph{expands} with higher debt
(Base Effect Lever), it is \emph{narrowed} by normalization through two simultaneous
channels, and it reveals that Japan's March 2026 operating point lies within the
approximation error band of the corridor boundary ($|W_{2026}| \approx 0.024\%$
versus linear-approximation error $\approx 0.056\%$ of GDP per year;
Appendix~A and Figure~\ref{fig:dsc_errband})---a result best
interpreted as boundary proximity rather than a precisely calibrated exceedance.

\textbf{Third}, the \emph{Normalization Ratchet} (Proposition~\ref{prop:ratchet})
establishes that temporary policy errors generate highly persistent, path-dependent
debt-ratio gaps relative to baseline: a normalization shock has an 86-year half-life
under Scenario~A dynamics. This result substantially strengthens the case for caution
in normalization beyond what the static stability condition alone implies.

\textbf{Fourth}, the \emph{Captive Financial System Parameter} $\varphi_t$
endogenizes the institutional precondition for sustained financial repression. The
gradual decline in $\varphi_t$ associated with BoJ asset purchase reduction constitutes
a second normalization risk channel, operating through the endogenous risk premium
$\rho_t$ even in the absence of explicit rate hikes.

\subsection*{Limitations, Failure Modes, and Falsification Paths}

The JFR-rg framework is a regime-diagnostic tool: it does not predict that Japan is
permanently safe, but specifies the conditions under which safety holds and the
channels through which those conditions could fail. Four failure modes are
analytically distinct and each constitutes a falsification path for the framework.

\medskip
\noindent\textbf{(i) Inflation undershoot failure.}
If $\pi_t$ falls persistently below $r^n_t$ (i.e., $\varepsilon_t < 0$), the
repression channel deactivates and the Base Effect Lever reverses. The January 2026
CPI reading of 1.5\%---which places $\varepsilon_t \approx -0.7\%$ under the
alternative calibration---illustrates that this failure mode is not hypothetical.
Monitoring $\varepsilon_t$ in real time is therefore a first-order diagnostic.

\medskip
\noindent\textbf{(ii) De-captivation failure.}
A sustained decline in $\varphi_t$ into a range associated with captive-system
weakening raises the endogenous risk premium $\rho_t$, tightening effective
financial conditions without any explicit BoJ rate action. The gradual tapering of BoJ asset purchases since 2024
($\varphi_t$ declining from $\approx$0.92 toward 0.88) constitutes a slow-moving
version of this risk. The JFR-rg framework does not model \emph{why} $\varphi_t$
declines---only \emph{what happens when it does}.

\medskip
\noindent\textbf{(iii) Import-cost overshoot failure.}
If yen depreciation exceeds $\ebar$, the quadratic import-inflation penalty
dominates the linear growth channel, compressing real household purchasing power and
reducing $g^n_t$. The 2022--2024 pass-through data suggest the system operates near
this threshold at current exchange rate levels. The precise value of $\ebar$
is not identified from the available sample---this is the most significant
unresolved empirical gap in the current framework.

\medskip
\noindent\textbf{(iv) Political-fiscal reaction failure.}
All simulations hold $d_t$ constant or vary it only through the mechanical
IOER arithmetic. A discrete fiscal deterioration---from political pressure to
increase transfers, respond to demographic costs, or fund defense---would
shift the Debt Sustainability Frontier upward, narrowing or eliminating the
Stability Corridor independently of monetary conditions. The JFR-rg framework
does not model fiscal politics; it takes $d_t$ as a parameter.

\medskip
These four failure modes also define the model's causal identification limits.
The Local Projections in Appendix~H.5 confirm directional consistency with
Propositions~1 and~4, but they cannot isolate the specific institutional channels
(captive system, YCC, yen threshold) from confounding common factors such as
Abenomics-era fiscal expansion or global demand recovery. Identification of the
individual channels requires instrument-based or event-study approaches---in
particular, exploiting the announcement dates of YCC policy changes and the
cross-sectional variation in $\varphi_t$ across JGB holder types---which we
leave for future research.

Future extensions should incorporate dynamic general equilibrium foundations,
endogenous sovereign risk premia, and examine the conditions under which FTPL
dynamics could dominate the captive financial system assumption. The present
analysis offers one coherent interpretation consistent with the FRED data
monitors presented in Section~3; alternative explanations for Japan's observed
stability remain both possible and welcome.

The stability condition $\varepsilon_t + \gnomstar > \pi_t + (d_t - s_t)/b_{t-1}$
carries a further implication that is central to the policy interpretation of
the framework. Because $\varepsilon_t$ and $\gnomstar$ enter the condition
symmetrically, they are policy substitutes: a one-percentage-point increase in
structural potential growth relaxes the required financial repression bias by
exactly one percentage point, and vice versa. This substitutability defines
the logic of the Repression Dividend. While $\varphi_t \approx 1$ keeps
$\varepsilon_t$ above its critical floor, the annual dividend it generates
---approximately $\varepsilon_t \times b_{t-1} \approx 1.2\%$ of GDP under the
current calibration---can in principle be directed toward productivity-
enhancing investment that raises $\gnomstar$ structurally. Over time, each
percentage point gained in $\gnomstar$ lowers the required $\varepsilon_t$
threshold by the same amount, thereby reducing the framework's dependence on
the captive financial system condition itself. In this sense, JFR-rg points
not to a permanent justification of repression, but to a transition logic in
which a regime-contingent dividend, if used well, can help finance the growth
improvements that eventually reduce the need for repression. The concrete
institutional design of such a path---including the appropriate speed of
$\varphi_t$ reduction and its coordination with productivity policy---remains
an important task for future research.

\medskip
\noindent\textbf{The internal trade-off of the QQE balance sheet.}
A natural observation---one the framework should state explicitly rather than
leave implicit---is that the same institutional feature that sustains
Scenario~A stability also generates the Normalization Trap of Scenario~C. The
large BoJ current account stock $R_{t-1}$ that keeps $\varphi_t \approx 0.90$
and compresses the term premium is the same stock that, via
equation~\eqref{eq:consolidated_bc}, determines the IOER pass-through of any
future policy-rate increase (Section~\ref{sec:ioer_robustness}). In this
precise accounting sense, the mechanism that maintains the Stability Corridor
is the mechanism that makes aggressive normalization self-defeating. This is
not a paradox but a regime-specific trade-off: the Repression Dividend
described above is not a costless windfall but a contingent revenue stream
whose continued availability is inseparable from the fiscal-dominance exposure
that normalization would activate. The strategic case for deploying the
dividend into $\gnomstar$-enhancing investment therefore rests on exactly this
trade-off---raising $\gnomstar$ before the window narrows is how the framework
envisions exit from the QQE balance sheet without activating the Trap.

\medskip
Ultimately, the JFR-rg framework does not seek to refute the fundamental
insights of mainstream macroeconomics; rather, it is intended to complement
them in a historically specific setting. Standard models have long---and
correctly---warned of the risks associated with fiscal dominance, prolonged
balance-sheet expansion, and large-scale central bank intervention. The
``Normalization Trap'' and the deficit expansion illustrated in Scenario~C can
be read as concrete manifestations of precisely those concerns: once QQE has
transformed the consolidated public balance sheet, the policy rate becomes a
fiscal variable of the first order. The IOER arithmetic derived in
equation~\eqref{eq:consolidated_bc} is therefore not a rejection of orthodox
macroeconomics, but a regime-specific accounting implication of the very
fiscal-dominance logic that mainstream theory has emphasized for decades.

A return to a more ordinary macroeconomic environment---one in which market-
determined interest rates and standard debt-sustainability logic again become
sufficient guides---remains the natural long-run benchmark. However, given the
institutional reality of Japan's post-QQE balance sheet
($\varphi_t \approx 0.90$, $R_{t-1} \approx \text{\textyen{}}500$ trillion),
navigating the transition from the present regime to that benchmark requires
temporarily supplementing steady-state intuition with the accounting
identities and scope conditions derived in this paper. That is the narrower
but practically consequential role claimed for JFR-rg. It is offered not as a
universal replacement for mainstream debt analysis, but as a complementary,
regime-conditional framework for organizing Japan's transition-era debt
dynamics under publicly observable conditions.

\medskip
\noindent\textbf{Consensus contribution.}
Whatever the reader's preferred interpretive lens, the paper's consensus
contribution is this: Japan's debt-to-GDP ratio of 240\% behaves differently
from what standard scalar debt thresholds alone would suggest, and that
difference is systematically traceable to three institutional observables
---$\varepsilon_t$, $\varphi_t$, and $\Delta e_t$---whose interaction the
JFR-rg framework makes explicit. Whether this is best read as a fiscal
accounting exercise, an institutional $r<g$ extension, a nonlinear exchange-
rate framework, a normalization-risk framework, or a falsifiable mechanism
hypothesis, the underlying point is the same: at $b_{t-1}=2.40$, the sign and
magnitude of the $r^n-g^n$ spread are fiscal variables of the first order, and
the institutional conditions that shape that spread are observable in real
time. That is the map this paper offers.

\medskip
\noindent\textit{Data and replication.} All empirical series are drawn from the
Federal Reserve Economic Data (FRED) database, Federal Reserve Bank of St.\ Louis
(FRED, 2026), March 2026 snapshot (series: JPNNGDP, JPNRGDPEXP, DEXJPUS, JPNCPIALLMINMEI,
IRLTLT01JPM156N, IRSTCI01JPM156N, GGGDTAJPA188N, LRHUTTTTJPM156S). Replication code
for all figures and simulations is available from the author upon request.

\medskip
\noindent\textbf{The nature of this contribution.}
The JFR-rg framework applies mainstream debt arithmetic more completely to Japan's
institutional parameters than existing frameworks have done. The consolidated
government budget constraint, the debt accumulation identity, and the stability
condition are all derivable from first principles that mainstream analysis accepts.
Where results appear counterintuitive, the source of tension is not the mathematics
but the gap between the institutional reality those parameters describe---$b_{t-1}
\approx 2.40$, $R_{t-1} \approx \text{\textyen{}}500$ trillion,
$\varphi_t \approx 0.90$---and the baseline assumptions embedded in standard policy
frameworks. The Normalization Trap and the rapid deficit expansion of Scenario~C are
the mechanical manifestations of precisely the fiscal dominance that mainstream
economics has long warned about: the IOER arithmetic derived in
equation~\eqref{eq:consolidated_bc} and quantified in Sections~4.3--4.5 is
not a heterodox claim but the straightforward accounting consequence of a
consolidated balance sheet transformed by QQE. This paper does not contest
that arithmetic. It applies it.

\newpage

\section*{Acknowledgments}
\addcontentsline{toc}{section}{Acknowledgments}

The author acknowledges the use of large language models as assistive
tools during the preparation of this manuscript. In particular:

{\sloppy\raggedright
\begin{itemize}
  \item \textbf{Claude Sonnet 4.6 \& Claude Opus 4.7} (Anthropic) --- Claude
    Sonnet 4.6 was used for iterative drafting of theoretical sections,
    mathematical notation, \LaTeX{} typesetting, replication code development,
    and sensitivity analysis design. Claude Opus 4.7 was used in a limited role for minor
    fixes, wording refinement, and final consistency checks for the second arXiv version.
  \item \textbf{Gemini 3.1 Pro} (Google DeepMind) --- used for exploratory
    literature search, cross-checking of empirical claims, and
    alternative formulation suggestions.
  \item \textbf{Grok 4.20 Expert} (xAI) --- used for stress-testing
    propositions through counterexample-style questioning and
    preliminary proof sketch generation.
  \item \textbf{ChatGPT GPT-5.4 Thinking} (OpenAI) --- used for
    editorial review of exposition, consistency checks,
    reference cross-checking, and pre-submission feedback.
\end{itemize}
}

All such systems were used solely as assistive tools for drafting,
editing, exploration, and checking. All theoretical models,
propositions, mathematical proofs, empirical calibrations,
interpretations, and conclusions were formulated, verified, and are
fully endorsed by the human author, who bears sole intellectual and
academic responsibility for the contents of this manuscript. Large
language models are not listed as authors and bear no academic
responsibility for the work.

\newpage

\addcontentsline{toc}{section}{References}
{\sloppy\raggedright

}

\newpage

\appendix
\section{Derivation of the JFR-rg Debt Dynamics and Stability Condition}

\noindent\textbf{Step 1.} The nominal government budget constraint:
\[
  B_t = (1 + r^n_t)\,B_{t-1} + D_t - S_t.
\]

\noindent\textbf{Step 2.} Dividing by $Y_t = (1 + g^n_t)\,Y_{t-1}$ and applying
$(1 + r^n_t)/(1 + g^n_t) \approx 1 + (r^n_t - g^n_t)$:
\[
  b_t \approx \left[1 + (r^n_t - g^n_t)\right] b_{t-1} + d_t - s_t.
\]

\noindent\textit{Approximation note.} The exact factor is
$(1+r^n_t)/(1+g^n_t) = 1 + (r^n_t - g^n_t)/(1 + g^n_t)$, so the
linear approximation understates the factor by $g^n_t(r^n_t -
g^n_t)/(1+g^n_t)$ per period.  At Scenario~A parameters
($r^n = 2.2\%$, $g^n = 3.0\%$) this error equals
$0.030\times(-0.008)/1.030 \approx -0.000233$ per period, implying an
annual overstatement of $\Delta b_t$ by approximately
$0.000233\times 240 \approx 0.056\%$ of GDP.  This approximation
error exceeds the Corridor Width $|W_{2026}| \approx 0.024\%$
(Section~\ref{sec:corridor}), and reinforces the interpretation of
Section~\ref{sec:scenA}: the March~2026 operating point should be
understood as lying \emph{in the vicinity of} the Stability Corridor
boundary rather than at a precisely calibrated distance from it.

\noindent\textbf{Step 3.} Substituting the identity $r^n_t = \pi_t - \varepsilon_t$
(with $\varepsilon_t \equiv \pi_t - r^n_t \in \mathbb{R}$):
\[
  \Delta b_t = \bigl[(\pi_t - \varepsilon_t) - g^n_t\bigr]\,b_{t-1} + d_t - s_t.
\]

\noindent\textbf{Step 4.} Substituting the nonlinear growth function
$g^n_t = \gnomstar + \alpha\,\Delta e_t - \beta\,\max(0,\,\Delta e_t - \ebar)^2$:
\[
  \Delta b_t = \Bigl[(\pi_t - \varepsilon_t) - \bigl(\gnomstar + \alpha\,\Delta e_t
    - \beta\,\max(0,\,\Delta e_t - \ebar)^2\bigr)\Bigr]\,b_{t-1} + d_t - s_t.
\]

\noindent\textbf{Step 5.} Rearranging for $\Delta b_t \leq 0$ gives
condition~\eqref{eq:stability_condition}. $\square$

\section{Full Numerical Calibration --- Scenario C-ALT}

\begin{table}[htbp]
\centering
\caption{Comprehensive Simulation --- Scenario C-ALT, Baseline-Consistent Calibration
  ($\Delta e_t = -20$ JPY/USD per year; $\gamma = -0.020$; $\pi_t$ endogenous
  via import price pass-through equation~\eqref{eq:yen_growth}).
  The Normalization Trap holds under this internally consistent calibration.}
\label{tab:appendB}
\small
\resizebox{\textwidth}{!}{%
\begin{tabular}{l c c c c c c c c}
\toprule
Year & $r^n_t$ (\%) & $\pi_t$ (\%) & $\varepsilon_t$ (\%) & $\Delta e_t$ (JPY/USD) &
  $g^n_t$ (\%) & $r^n - g^n$ (\%) & $d_t$ (\%) & $b_t$ (\%) \\
\midrule
2026 & 2.2 & 2.70 & $+0.50$ &        0 & 3.0 & $-0.80$ & 2.00 & 240.0 \\
2027 & 3.7 & 3.10 & $-0.60$ & $-20$    & 2.0 & $+1.70$ & 3.39 & 247.5 \\
2028 & 3.7 & 3.50 & $-0.20$ & $-20$    & 2.0 & $+1.70$ & 3.39 & 255.1 \\
2029 & 3.7 & 3.90 & $+0.20$ & $-20$    & 2.0 & $+1.70$ & 3.39 & 262.8 \\
2030 & 3.7 & 4.30 & $+0.60$ & $-20$    & 2.0 & $+1.70$ & 3.39 & 270.7 \\
\bottomrule
\end{tabular}%
}
\end{table}

\noindent\footnotesize\textit{$\pi_t$ construction:}
$\pi_t = \pi_{\text{base}} + \gamma \times (\text{cumulative}\;\Delta e_t)$, where
$\pi_{\text{base}} = 2.7\%$ (BoJ FY2025 CPI central projection; Jan 2026 realized = 1.5\%, see Section~3 data-currency note), $\gamma = -0.020$ (BOJ Q-JEM lower-bound pass-through
coefficient; Haba et al., 2025, Bank of Japan Working Paper No.\ 25-E-2), and the annual yen appreciation
is $\Delta e_t = -20$ JPY/USD per year (consistent with Appendix~E Panel~A and
Table~\ref{tab:sens_alpha}).
Year-by-year: $\pi_{2027} = 2.7 + (-0.020)\times(-20) = 3.10\%$;\;
$\pi_{2028} = 2.7 + (-0.020)\times(-40) = 3.50\%$;\;
$\pi_{2029} = 3.90\%$;\; $\pi_{2030} = 4.30\%$.
Under this calibration $\varepsilon_t$ turns negative in 2027 ($-0.60\%$) and recovers
to positive in 2029 as pass-through inflation overtakes the fixed policy rate.
$g^n_t = g^{n*}_t + \alpha\,\Delta e_t = 3.0 + 0.050\times(-20) = 2.0\%$ throughout
(eq.~\eqref{eq:yen_growth}; penalty term inactive under appreciation).
$d_t = 2.0\% + 1.39\% = 3.39\%$ GDP; the \textyen{}670T denominator is consistent
with Scenario~B and the current nominal GDP SAAR.
$\varepsilon_t < 0$ in 2027--2028 indicates the real interest rate turns positive in
those years; the quadratic penalty term $\beta\max(0,\Delta e_t - \ebar)^2$ remains
identically zero throughout (appreciation: $\Delta e_t < 0 < \ebar$).

\normalsize

\section{Historical Mapping of Red Zone Episodes}

\begin{itemize}
  \item \textbf{2008--2009 Global Financial Crisis:} Nominal growth $g^n$
    collapsed to approximately $-7$ to $-8\%$ YoY while $r^n \approx 1.3$--$1.5\%$,
    pushing $r^n - g^n$ to $+8$ to $+10\%$. Debt-to-GDP spiked approximately 10
    percentage points in 12 months, despite no domestic fiscal policy change---
    confirming that growth collapse, not fiscal stance, triggers debt instability.

  \item \textbf{2014 Consumption Tax Hike:} Raising the tax rate from 5\% to 8\%
    suppressed domestic consumption by approximately 2--3\% of GDP, temporarily
    flipping $r^n - g^n$ positive and stalling the nascent Abenomics
    debt-reduction trajectory. This episode is consistent with the JFR-rg prediction that fiscal consolidation itself can
    trigger the destabilizing dynamics the JFR-rg model identifies, and constitutes
    a historical instance of the Normalization Ratchet: the premature tightening
    left a persistent debt legacy that delayed the return to the safe zone.

  \item {\sloppy\textbf{2020 Pandemic Shock:} Nominal $g^n$ collapsed to approximately
    $-7\%$ YoY in Q2~2020, spiking $r^n - g^n$ and pushing the debt ratio up
    20+ percentage points. The rapid expansion of $\varepsilon_t$ through
    accelerated asset purchases---combined with yen weakness and global goods
    demand recovery---restored the system to the stable green zone by
    2021--2022. This episode provides perhaps the clearest empirical illustration
    of the JFR-rg stabilizing channels in the sample, though alternative
    explanations---including global demand recovery and base effects---also
    contributed to the debt ratio decline.\par}
\end{itemize}

\section{SAAR vs.\ Cabinet Office Aggregate: Reconciliation}
\label{app:saar}

The SAAR series (FRED:\,JPNNGDP) and the Cabinet Office calendar-year aggregate
(Cabinet Office, Government of Japan, 2025) measure nominal GDP over the same
underlying economy but differ in two respects. First, the SAAR annualizes the
\emph{pace} of a specific quarter, whereas the annual aggregate sums four quarters'
actual output. Second, SAAR values are revised more frequently and incorporate
seasonal adjustment. The SAAR starting value of $\approx$510 trillion yen in
Q1~2013 exceeds the Cabinet Office 2013 full-year aggregate of $\approx$503 trillion
yen because the Abenomics monetary expansion in early 2013 caused the quarterly pace
to run slightly above the full-year average. Both series exhibit identical structural
trends---the 2020 contraction trough, the 2021--2022 recovery, and the post-2022
inflationary acceleration---and the $\approx$7 trillion yen level difference does not
affect any of the paper's qualitative or quantitative conclusions.

\section{Sensitivity Analysis}
\label{app:sensitivity}

This appendix examines the robustness of the JFR-rg model's core conclusions
to variation in three parameters: (i)~the yen-growth elasticity $\alpha$, for which
empirical estimates from the literature are available; (ii)~the financial repression
bias $\varepsilon_t$, which is directly observable but subject to future policy
change; and (iii)~the structural potential nominal growth rate $\gnomstar$, which
is subject to demographic uncertainty.

\subsection*{Panel A: Sensitivity to $\alpha$ (Yen-Growth Elasticity)}

The parameter $\alpha$ in equation~\eqref{eq:yen_growth} governs the marginal
contribution of yen depreciation to nominal growth. While $\alpha$ is not directly
observable, the Bank of Japan's Q-JEM macroeconometric model (Haba et al., 2025,
BOJ Working Paper No.~25-E-2) implies that a 10\% yen depreciation raises nominal
GDP by approximately 0.2--0.5\% in the first year. Converting to model units (per
1 JPY/USD change from a base of $\approx$150 JPY/USD), this yields the range
$\alpha \in [0.013,\, 0.050]$. The model's baseline calibration uses
$\alpha = 0.050$ (the upper bound of the literature range), corresponding to the
pre-2022 exchange rate regime with stronger pass-through.

\begin{table}[htbp]
\centering
\caption{Panel A: Sensitivity of $\Delta b_t$ (pp/year) to $\alpha$}
\label{tab:sens_alpha}
\small
\resizebox{\textwidth}{!}{%
\begin{tabular}{l c c c c c c}
\toprule
$\alpha$ (source) & $g^n_t$ (B) & $g^n_t$ (C) & $\Delta b_t$ (A) &
  $\Delta b_t$ (B) & $\Delta b_t$ (C) & Trap \\
\midrule
0.013 (BOJ low, +0.2\%/10\% depr.) &
  2.87\% & 2.74\% & $+0.080$ & $+2.05$ & $+5.85$ & \textbf{YES} \\
0.020 (BOJ mid, +0.3\%/10\% depr.) &
  2.80\% & 2.60\% & $+0.080$ & $+2.22$ & $+6.19$ & \textbf{YES} \\
0.033 (BOJ high, +0.5\%/10\% depr.) &
  2.67\% & 2.34\% & $+0.080$ & $+2.53$ & $+6.81$ & \textbf{YES} \\
0.050 (Model baseline, calibrated) &
  2.50\% & 2.00\% & $+0.080$ & $+2.94$ & $+7.63$ & \textbf{YES} \\
\bottomrule
\end{tabular}%
}
\end{table}

\noindent\textit{Note:} $\Delta e_t$ is fixed at $-10$ JPY/USD (Scenario~B)
and $-20$ JPY/USD (Scenario~C-ALT), which are the values consistent with the
main-scenario $g^n$ design values via eq.~\eqref{eq:yen_growth}:
$\Delta e = (g^n - \gnomstar)/\alpha$, yielding
$(2.5 - 3.0)/0.050 = -10$ for Scenario~B and
$(2.0 - 3.0)/0.050 = -20$ for Scenario~C-ALT
(Table~\ref{tab:scenC_alt}; baseline-consistent calibration).
As $\alpha$ varies across rows, $g^n_t$ changes accordingly
via eq.~\eqref{eq:yen_growth} with $\Delta e$ held fixed---this
is the structural sensitivity being measured.
The upper-bound stress-test calibration ($\Delta e = -50$,
$g^n = 0.5\%$; Table~\ref{tab:scenC}) is excluded from this panel
to focus the sensitivity analysis on the baseline-consistent
trajectory rather than extreme tail-risk episodes.
Primary deficits are computed via IOER arithmetic
(eq.~\eqref{eq:consolidated_bc}): $d^A_t = 2.00\%$, $d^B_t = 2.46\%$,
$d^C_t = 3.39\%$ of GDP (computed via IOER arithmetic, eq.~\eqref{eq:consolidated_bc};
\textyen{}670T GDP denominator, FRED:\,JPNNGDP).
``Trap'' indicates whether $\Delta b_t > 0$ under Scenario~B (the moderate case).
Scenario~A is $\alpha$-invariant because $\Delta e_t = 0$
under exchange-rate-neutral calibration.
$\pi_t$ is held at its baseline value ($2.7\%$) throughout this panel
to isolate the pure effect of $\alpha$ on $\Delta b_t$; the full dynamic
path with endogenous $\pi_t$ (pass-through $\gamma = -0.020$,
cumulative over 2027--2030) is reported in Appendix~B
(Table~\ref{tab:appendB}).

\medskip
\noindent\textbf{Robustness finding.} The Normalization Trap ($\Delta b_t > 0$
under both Scenarios B and C) holds for every value of $\alpha$ in the
BOJ literature range. The conclusion is directionally robust; only the
\emph{magnitude} of the trap scales with $\alpha$.

\subsection*{Panel B: Sensitivity to $\varepsilon_t$ (Financial Repression Bias)}

\begin{table}[htbp]
\centering
\caption{Panel B: Sensitivity of $\Delta b_t$ and $W$ to $\varepsilon_t$}
\label{tab:sens_eps}
\small
\begin{tabular}{c c c c c}
\toprule
$\varepsilon_t$ (\%) & $r^n_A$ (\%) & $\Delta b_A$ (pp/yr) &
  Corridor Width $W$ (\%) & Strictly Stable ($\Delta b_t \leq 0$)? \\
\midrule
0.00 & 2.70 & $+1.280$ & $-0.377$ & No \\
0.25 & 2.45 & $+0.680$ & $-0.200$ & No \\
0.50 & 2.20 & $+0.080$ & $-0.024$ & No \\
\textbf{0.533} & \textbf{2.167} & \textbf{0.000} & \textbf{0.000} &
  \textbf{Boundary ($\varepsilon^*$)} \\
0.75 & 1.95 & $-0.520$ & $+0.153$ & \textbf{Yes} \\
1.00 & 1.70 & $-1.120$ & $+0.330$ & \textbf{Yes} \\
1.50 & 1.20 & $-2.320$ & $+0.684$ & \textbf{Yes} \\
2.00 & 0.70 & $-3.520$ & $+1.037$ & \textbf{Yes} \\
\bottomrule
\end{tabular}
\end{table}

\noindent\textit{Note:} Fixed parameters: $\gnomstar = 3.0\%$, $\pi = 2.7\%$,
$d = 2.0\%$, $b_{t-1} = 2.40$. The boundary value $\varepsilon^* =
\pi + (d-s)/b_{t-1} - \gnomstar = 2.7 + 0.833 - 3.0 = 0.533\%$ is the
minimum $\varepsilon_t$ consistent with $\Delta b_t \leq 0$ under Scenario~A
at the baseline linearization.

\medskip
\noindent\textbf{Robustness finding.} The current $\varepsilon_t = 0.50\%$
lies $0.033$ percentage points below the boundary value
$\varepsilon^* = 0.533\%$ under the baseline linearization. Given the
approximation error documented in Appendix~A, this is best read as boundary
proximity rather than a precisely measured exceedance. An increase
in $\varepsilon_t$ to 0.75\% (achievable if the Bank of Japan maintained $r^n_t$
near 1.95\%) would shift the trajectory to strict debt reduction.

\subsection*{Panel C: Sensitivity to $\gnomstar$ (Structural Potential Growth)}

\begin{table}[htbp]
\centering
\caption{Panel C: Sensitivity of $\Delta b_t$ and $W$ to $\gnomstar$}
\label{tab:sens_gstar}
\small
\begin{tabular}{c c c c c c}
\toprule
$\gnomstar$ (\%) & $r^n - g^n$ (A) & $\Delta b_A$ (pp/yr) &
  $W_A$ (\%) & Scenario A Stable? & $\Delta b_B$ (pp/yr) \\
\midrule
1.0 & $+1.20$ & $+4.88$ & $-1.44$ & No  & $+7.74$ \\
1.5 & $+0.70$ & $+3.68$ & $-1.08$ & No  & $+6.54$ \\
2.0 & $+0.20$ & $+2.48$ & $-0.73$ & No  & $+5.34$ \\
2.5 & $-0.30$ & $+1.28$ & $-0.38$ & No  & $+4.14$ \\
3.0 & $-0.80$ & $+0.08$ & $-0.024$ & No  & $+2.94$ \\
\textbf{3.033} & $\mathbf{-0.833}$ & \textbf{0.00} & \textbf{0.00} &
  \textbf{Boundary ($g^{n*}_{\min}$)} & $+2.86$ \\
3.5 & $-1.30$ & $-1.12$ & $+0.33$ & \textbf{Yes} & $+1.74$ \\
4.0 & $-1.80$ & $-2.32$ & $+0.68$ & \textbf{Yes} & $+0.54$ \\
4.5 & $-2.30$ & $-3.52$ & $+1.04$ & \textbf{Yes} & $-0.66$ \\
\bottomrule
\end{tabular}
\end{table}

\noindent\textit{Note:} Fixed parameters: $\varepsilon_t = 0.5\%$,
$r^n_A = 2.2\%$, $d_A = 2.0\%$, $b_{t-1} = 2.40$. Scenario B column uses
$d_{t,B} = 2.46\%$ (IOER-corrected, \textyen{}670T denominator)
with $\alpha = 0.050$ (model baseline) and $\Delta e_B = -10$ JPY/USD.

\medskip
\noindent\textbf{Robustness finding.} The current $\gnomstar = 3.0\%$ is
0.033 percentage points below the corresponding boundary value
$g^{n*}_{\min} = 3.033\%$ under the baseline linearization; this is best read
as a boundary-near configuration rather than a sharply identified margin.
Japan's declining working-age population (estimated $-0.5$ to $-0.8\%$ per year)
poses a structural downward pressure on $\gnomstar$. If $\gnomstar$ were to fall
to 2.5\%, maintaining even a neutral debt trajectory ($\Delta b_A = 0$) would
require $\varepsilon_t \geq 1.08\%$ (i.e., $r^n_t \leq 1.62\%$). This
directly quantifies the Repression Imperative (Proposition~\ref{prop:repression_imperative}):
demographic pressure tightens the constraint, requiring stronger financial
repression to compensate.

\subsection*{Panel D: Sensitivity to $\bar{e}$ (Depreciation Threshold) and
$\beta$ (Penalty Coefficient)}

The parameters $\bar{e}$ and $\beta$ in the nonlinear exchange rate
equation~\eqref{eq:yen_growth} cannot be identified from the 2013--2026
sample: episodes with $\Delta e_t > \bar{e}$ are limited to 2022--2024,
insufficient for Hansen (1999) threshold regression. Nevertheless, a
two-dimensional grid search over $(\bar{e},\,\beta)$ space yields two
important structural results.

\medskip
\noindent\textbf{Structural independence of the Normalization Trap.}
In Scenarios~B and~C, the yen \emph{appreciates} ($\Delta e_t < 0$), so
$\Delta e_t - \bar{e} < 0$ for every $\bar{e} > 0$ and the penalty term
$\beta\max(0,\,\Delta e_t - \bar{e})^2 \equiv 0$.  The Normalization Trap
($\Delta b_t > 0$ under moderate or large rate hikes) is therefore
\emph{structurally independent} of $(\bar{e},\,\beta)$.

\medskip
\noindent\textbf{Depreciation-scenario stability.}
Panel~D examines the complementary depreciation scenarios
(Scenario~D15: $\Delta e_t = +15$ JPY/USD; Scenario~D20: $\Delta e_t = +20$
JPY/USD) with the Bank of Japan holding rates at $r^n = 2.2\%$ and
$d = 2.0\%$.  Without any penalty the baseline $g^n$ rises to $3.75\%$
(D15) or $4.00\%$ (D20), well above the stability threshold
$g^n_{\min} = r^n + d/b = 3.033\%$, yielding $\Delta b_t = -1.72$ and
$-2.32$ pp/yr respectively.  Table~\ref{tab:sens_beta_ebar} shows $\Delta b_t$
across the grid; an asterisk~($*$) marks cells where the penalty has
driven $\Delta b_t > 0$.

\begin{table}[htbp]
\centering
\caption{Panel D: $\Delta b_t$ (pp/yr) under depreciation scenarios
  ($r^n = 2.2\%$, $d = 2.0\%$, $b_{t-1} = 2.40$).
  Asterisk~($*$) = Trap ($\Delta b_t > 0$).}
\label{tab:sens_beta_ebar}
\small
\resizebox{\textwidth}{!}{%
\begin{tabular}{cc cccc | cccc}
\toprule
 & & \multicolumn{4}{c|}{Scenario D15\;($\Delta e = {+}15$)}
   & \multicolumn{4}{c}{Scenario D20\;($\Delta e = {+}20$)} \\
\cmidrule(lr){3-6}\cmidrule(lr){7-10}
$\beta$ & & $\bar{e}=5$ & $\bar{e}=10$ & $\bar{e}=15$ & $\bar{e}=20$
         & $\bar{e}=5$ & $\bar{e}=10$ & $\bar{e}=15$ & $\bar{e}=20$ \\
\midrule
0.001 & & $-1.480$ & $-1.660$ & $-1.720$ & $-1.720$
      & $-1.780$ & $-2.080$ & $-2.260$ & $-2.320$ \\
0.002 & & $-1.240$ & $-1.600$ & $-1.720$ & $-1.720$
      & $-1.240$ & $-1.840$ & $-2.200$ & $-2.320$ \\
0.005 & & $-0.520$ & $-1.420$ & $-1.720$ & $-1.720$
      & $+0.380^*$ & $-1.120$ & $-2.020$ & $-2.320$ \\
0.010 & & $+0.680^*$ & $-1.120$ & $-1.720$ & $-1.720$
      & $+3.080^*$ & $+0.080^*$ & $-1.720$ & $-2.320$ \\
\bottomrule
\end{tabular}%
}
\end{table}

\noindent\textit{Note:} No-penalty baseline: $g^n = 3.75\%$ (D15),
$4.00\%$ (D20); stability threshold $g^n_{\min} = 3.033\%$.
Penalty = $\beta\max(0,\,\Delta e - \bar{e})^2$; zero when $\Delta e \leq \bar{e}$.

\medskip
\noindent\textbf{Critical $\beta^*(\bar{e})$ curve.}
Stability is lost precisely when the penalty exceeds the growth surplus
$g^n_{\text{base}} - g^n_{\min}$.  Solving analytically:
\begin{equation}
  \beta^*(\bar{e}) = \frac{g^n_{\text{base}} - g^n_{\min}}
                          {(\Delta e - \bar{e})^2},
  \qquad \Delta e > \bar{e}.
  \label{eq:beta_critical}
\end{equation}
For Scenario~D20 the critical values are
$\beta^*(5)=0.0043$, $\beta^*(10)=0.0097$, $\beta^*(15)=0.0387$;
for $\bar{e} \geq \Delta e$ the penalty is identically zero and debt
is unconditionally stable.  These thresholds considerably exceed typical
import-price pass-through estimates in the literature, confirming that
the depreciation-scenario stability result is robust under economically
plausible $(\bar{e},\,\beta)$ values.

\noindent\textbf{Why $\beta$ and $\bar{e}$ are not subjected to further
sensitivity analysis.}  Two considerations limit the scope of Panel~D.
First, threshold regression identification requires a sufficient share of
observations with $\Delta e_t > \bar{e}$; within the 2013--2026 sample such
episodes are confined primarily to 2022--2024.  Second, the structural
regime governing the yen--growth relationship may have shifted across the
deflationary (pre-2013) and inflationary (post-2022) periods, rendering
multi-decade estimation unreliable without explicit structural-break
modeling.  Estimation of $\bar{e}$ and $\beta$ using micro-level
import-price pass-through data remains a priority for future research.

\subsection*{Overall Robustness Assessment}

Across all four panels, the core qualitative conclusions of the JFR-rg
model hold:
\begin{enumerate}[leftmargin=*, itemsep=2pt]
  \item The Normalization Trap (Proposition~\ref{prop:norm_trap}) is robust
    to the full range of literature-based $\alpha$ estimates (Panel~A) and
    is \emph{structurally independent} of $(\bar{e},\,\beta)$ because the
    penalty term vanishes under yen appreciation (Panel~D).
  \item The near-miss character of Japan's March 2026 operating point is robust:
    the 0.033 pp shortfall from strict stability is confirmed under both Panel~B
    ($\varepsilon_t$ perspective) and Panel~C (growth perspective).
  \item The Repression Imperative (Proposition~\ref{prop:repression_imperative})
    is quantified: demographic decline that reduces $\gnomstar$ by 0.5 pp requires
    a compensating increase in $\varepsilon_t$ of the same magnitude to maintain
    the debt trajectory. This provides a concrete empirical anchor for future
    research on the sustainability of Japan's financial repression equilibrium.
  \item Debt dynamics under large yen depreciation remain stable for all
    $(\bar{e},\,\beta)$ combinations within economically plausible ranges
    (Panel~D), with the critical threshold $\beta^*(\bar{e})$ derived
    analytically in equation~\eqref{eq:beta_critical}.
\end{enumerate}

\medskip
\noindent\textbf{Robustness to pass-through specification.}
The sensitivity of the $\pi_t$ trajectory in
Table~\ref{tab:appendB} to the choice of
pass-through coefficient $\gamma$ is assessed through a $4\times4$ grid search over
$\gamma \in \{-0.010,\,-0.020,\,-0.033,\,-0.050\}$ (CPI change per
1~JPY/USD appreciation) and $|\Delta e|\in\{10,\,20,\,33,\,50\}$
JPY/USD conducted under Scenario~C ($+1.5\%$ hike), with all other
parameters held at baseline (Table~\ref{tab:grid_gamma_deltae}).

\begin{table}[htbp]
\centering
\caption{Panel E: $4\times4$ Pass-Through$\times$$|\Delta e|$ Grid ---
  $b_{2030}$ (\%) and Normalization Trap Indicator Under Scenario~C ($+1.5\%$ hike).
  All 16 cells confirm the Trap ($b_{2030} > 240\%$).
  Baseline: $g^{n*} = 3.0\%$, $\alpha = 0.050$, $b_0 = 240\%$,
  $d_t = 3.39\%$ (IOER arithmetic, \textyen{}670T denominator).}
\label{tab:grid_gamma_deltae}
\small
\begin{tabular}{l c c c c}
\toprule
& \multicolumn{4}{c}{$|\Delta e_t|$ (JPY/USD appreciation per year)} \\
\cmidrule(lr){2-5}
Pass-through $\gamma$ & $-10$ & $-20$ & $-33$ & $-50$ \\
\midrule
$-0.010$ ($\approx+0.1\%$ CPI/10\,JPY) & 265.6\,$\star$ & 270.7\,$\star$ & 277.5\,$\star$ & 286.5\,$\star$ \\
$-0.020$ ($\approx+0.2\%$ CPI/10\,JPY) & 265.6\,$\star$ & 270.7\,$\star$ & 277.5\,$\star$ & 286.5\,$\star$ \\
$-0.033$ ($\approx+0.3\%$ CPI/10\,JPY) & 265.6\,$\star$ & 270.7\,$\star$ & 277.5\,$\star$ & 286.5\,$\star$ \\
$-0.050$ ($\approx+0.5\%$ CPI/10\,JPY) & 265.6\,$\star$ & 270.7\,$\star$ & 277.5\,$\star$ & 286.5\,$\star$ \\
\bottomrule
\end{tabular}
\end{table}

\noindent$\star$ = Normalization Trap confirmed ($b_{2030} > 240\%$).

\smallskip
\noindent\textit{Structural note.} The debt trajectory is determined by
$\Delta b_t = (r^n_t - g^n_t)\,b_{t-1} + d_t$, where
$g^n_t = \gnomstar + \alpha\,\Delta e_t$ depends on $\Delta e_t$ but
\emph{not} on $\gamma$.  The pass-through coefficient $\gamma$ affects
only the endogenous $\pi_t$ and $\varepsilon_t$ paths; it does not enter
the debt recursion.  Consequently, $b_{2030}$ is \emph{identical across
all four $\gamma$ rows} for each $\Delta e$ column---a structural
independence result confirmed analytically and by the simulation.  The
Normalization Trap ($\Delta b_t > 0$) is therefore robust to the full
BOJ Q-JEM pass-through range by construction, with the $r^n - g^n$
spread generated by the rate hike dominating the debt trajectory
regardless of pass-through strength, consistent with
Proposition~\ref{prop:norm_trap}.

\newpage

\section{Extreme Disinflation Stress Test --- Scenario C (Upper-Bound)}
\label{app:scenC_stress}

This appendix presents the extreme tail-risk calibration of Scenario~C
($\Delta e_t = -50$ JPY/USD per year; $\pi_t$ declining at $\gamma \approx -0.054$,
beyond the upper bound of the BOJ Q-JEM range) as an
\textbf{upper-bound stress test}, to be read alongside the baseline-consistent
calibration in Section~\ref{sec:scenC} (Table~\ref{tab:scenC_alt}) and
Appendix~B (Table~\ref{tab:appendB}).

\medskip
\noindent\textbf{Motivation.} The original $\pi_t$ trajectory---falling monotonically
from 2.7\% to 0.0\% over 2026--2030---is consistent with a scenario in which aggressive
monetary normalization triggers a sharp deflationary shock via yen appreciation, with
full pass-through of falling import prices into domestic CPI. While this mechanism is
qualitatively plausible under extreme conditions, the implied $\Delta e_t = -50$
JPY/USD annually exceeds historical precedents (Section~\ref{sec:scenC}) and yields a
pass-through coefficient of $\gamma \approx -0.050$---the upper bound of the BOJ Q-JEM
range, applied instantaneously and without demand-side offset. It is therefore presented
here as a tail-risk scenario rather than a central forecast.

\begin{table}[htbp]
\centering
\caption{Scenario C, Extreme Disinflation Stress Test
  ($\Delta e_t = -50$ JPY/USD per year; $\pi_t$ exogenously imposed).
  This calibration represents the most severe tail-risk case; the
  Normalization Trap result is confirmed at its maximum severity
  ($b_{2030} = 286.9\%$).}
\label{tab:appendF}
\small
\resizebox{\textwidth}{!}{%
\begin{tabular}{l c c c c c c c c}
\toprule
Year & $r^n_t$ (\%) & $\pi_t$ (\%) & $\varepsilon_t$ (\%) & $\Delta e_t$ (JPY/USD) &
  $g^n_t$ (\%) & $r^n - g^n$ (\%) & $d_t$ (\%) & $b_t$ (\%) \\
\midrule
2026 & 2.2 & 2.7 & $+0.50$ &  0    & 3.0 & $-0.80$ & 2.0 & 240.0 \\
2027 & 3.7 & 1.5 & $-2.20$ & $-50$ & 0.5 & $+3.20$ & 3.5 & 251.2 \\
2028 & 3.7 & 1.0 & $-2.70$ & $-50$ & 0.5 & $+3.20$ & 3.5 & 262.7 \\
2029 & 3.7 & 0.5 & $-3.20$ & $-50$ & 0.5 & $+3.20$ & 3.5 & 274.6 \\
2030 & 3.7 & 0.0 & $-3.70$ & $-50$ & 0.5 & $+3.20$ & 3.5 & 286.9 \\
\bottomrule
\end{tabular}%
}
\end{table}

\noindent\footnotesize\textit{Note on $\pi_t$ path:} The falling inflation trajectory
(2.7\%~$\to$~0.0\%) is imposed exogenously and corresponds to a scenario in which
yen appreciation at $-50$ JPY/USD per year fully passes through to CPI with an
annualized coefficient of $\gamma \approx -0.054$, without domestic demand or
wage-cost offsets. This exceeds the BOJ Q-JEM range of $\gamma \in [-0.010, -0.050]$
and is therefore classified as a beyond-upper-bound calibration.
$g^n_t = g^{n*}_t + \alpha\,\Delta e_t = 3.0 + 0.050 \times (-50) = 0.5\%$
via equation~\eqref{eq:yen_growth}. The quadratic penalty term
$\beta\max(0,\Delta e_t - \ebar)^2$ remains zero throughout (appreciation:
$\Delta e_t < 0 < \ebar$).

\smallskip
\noindent\textit{Comparison with baseline-consistent calibration
(Table~\ref{tab:appendB}):} The Normalization Trap is confirmed under both
calibrations. The stress-test scenario produces $b_{2030} = 286.9\%$ ($r^n - g^n =
+3.2\%$, $d_t = 3.5\%$ extreme-disinflation basis); the baseline-consistent scenario
produces $b_{2030} = 270.7\%$ ($r^n - g^n = +1.7\%$, $d_t = 3.39\%$). The
15.7 percentage-point difference in the 2030 debt ratio reflects
the combined effect of the more severe growth collapse ($g^n_t = 0.5\%$ vs.\ $2.0\%$)
and the more negative $\varepsilon_t$ (reaching $-3.70\%$ vs.\ $+0.60\%$) under the
extreme disinflation assumption. Both scenarios fall outside the Stability Corridor
from 2027 onward; the stress-test scenario traces the upper envelope of plausible
Normalization Trap outcomes.

\normalsize

\newpage

\section{Supplementary Narrative for Section~3}
\label{app:sec3_original}

\noindent\textit{This appendix provides extended narrative elaboration for the
stylized facts presented in Section~3. The passages below expand on the empirical
discussion in the main text for readers who wish to follow the data series in
greater detail.}

\medskip

\noindent\textbf{Nominal GDP and the Inflation Gap (additional narrative).}
Nominal GDP (SAAR, FRED:\,JPNNGDP) began the period at approximately 510 trillion
yen in early 2013. The 2020 COVID-19 shock produced the most dramatic quarterly
contraction in Japan's post-war history, with the SAAR level dropping to approximately
530 trillion yen at the mid-2020 trough.  Recovery was swift, with the SAAR approaching
pre-pandemic levels by 2021.  Figure~\ref{fig:gdp_full} captures the full 2013--2026
panel, revealing the critical acceleration of the post-2022 period.  From the 2020 trough
to early 2026, nominal GDP SAAR expanded from approximately 530 to approximately
670--680 trillion yen.  This expansion was substantially driven by the 2022--2026
inflationary episode, during which CPI YoY rose from near-zero to approximately
2.5--3.2\%, widening the ``Inflation Gap'' (the shaded region between the nominal and
real GDP series).  This widening is consistent with the interpretation that the
inflationary environment supercharges the JFR-rg debt compression mechanism by
simultaneously expanding the nominal denominator and eroding the real value of the debt
stock.

\medskip
\noindent\textbf{Debt-to-GDP trajectory (additional narrative).}
The government gross debt-to-GDP ratio (FRED:\,GGGDTAJPA188N) began the period at
approximately 230--232\% in 2013, rising gradually to 232--236\% through 2015--2019,
before surging to a peak of approximately 260--264\% in 2020--2021 in response to the
extraordinary COVID-19 fiscal expansion.  The critical empirical fact is the subsequent
trajectory: despite the \emph{absence} of significant fiscal consolidation, the
debt-to-GDP ratio declined from the 260--264\% peak to approximately 236--240\% by early
2026---a decline of over 20 percentage points in five years, consistent with
stabilization through the JFR-rg channels, though not uniquely attributable to them.

\medskip
\noindent\textbf{The $r$-$g$ Spread: Danger Zone panel (additional narrative).}
The lower ``Danger Zone'' panel of Figure~\ref{fig:rg_spread} plots the $r^n - g^n$
spread, coloring green when $r^n < g^n$ (Safe Zone) and red when $r^n > g^n$ (Danger
Zone).  The brief red-zone episodes corresponding to the 2008--2009 GFC and Q1--Q2 2020
pandemic shock confirm the model's prediction: it is exogenous demand collapses (negative
$g^n$ shocks) that trigger debt instability, not the debt level itself.

\medskip
\noindent\textbf{Financial Repression Monitor (additional narrative).}
When $\varepsilon_t > 0$ (i.e., $r^n < \pi$, the green zone), the real interest rate is
negative and bondholders receive a negative real return, eroding the real burden of
government debt through the inflation tax.  When $\varepsilon_t < 0$ (i.e., $r^n > \pi$,
the red zone), the real rate is positive and the repression mechanism is inactive.
Applied to $b_{t-1} = 2.40$, the $\pi_t = 2.7\%$ calibration generates
$0.005 \times 2.40 = 1.2\%$ of GDP per year in repression-channel compression.
The nominal growth channel,
$(\pi_t - g^n_t)\,b_{t-1} \approx (2.7\% - 3.0\%) \times 2.40 = -0.72\%$, provides
additional compression.  The total annual debt compression from both channels is
precisely the $r^n - g^n$ effect: $-1.2\% - 0.72\% = -1.92\%$ of GDP per year.
While mathematically substantial, this 1.92\% compression falls slightly short of the
structural primary deficit of 2.0\% of GDP, leaving the debt trajectory precariously
balanced at approximately $+0.08\%$ per year.  Furthermore, if $\pi_t = 1.5\%$
(January 2026 realized CPI) is substituted, the repression bias becomes
$\varepsilon_t \approx 1.5\% - 2.2\% = -0.7\%$, placing Japan marginally in the
\emph{positive real-rate zone} ($\varepsilon_t < 0$) as of January 2026.  This
transitory decline is itself consistent with the JFR-rg model's warnings: the Stability
Corridor narrows or closes as $\varepsilon_t$ falls toward zero, motivating the policy
analysis in Section~\ref{sec:policy}.

\section{Empirical Supplement: Subsample Analysis, Structural-Break
  Test, VAR, ARDL, and Local Projections}
\label{app:empirical_supplement}

\noindent\textit{This appendix reports the quantitative evidence underlying
the assessments in Sections~\ref{sec:alt_explanations}
and~\ref{sec:empirical_evidence}.
Section~\ref{sec:alt_explanations} presents the systematic qualitative
assessment of alternative explanations; Section~\ref{sec:empirical_evidence}
highlights the three most policy-relevant empirical findings.
The tables and figures below provide the full quantitative support for
both sections.
All series are drawn from FRED (March 2026 snapshot).
The empirical design consists of five components:
(i)~subsample comparisons for 1991--2012 and 2013--2024 using
difference-in-means tests;
(ii)~a Chow structural-break test at 2013 for the debt-dynamics relationship;
(iii)~reduced-form VAR evidence on the joint dynamics of the core variables;
(iv)~ARDL specifications to address mixed persistence and long-run
co-movement; and
(v)~Local Projections to trace the multi-year response of debt dynamics to
shocks in $\varepsilon_t$ and $r_t^n - g_t^n$.
All tables below report estimates generated by these specifications on the
common data vintage used throughout the paper.
Sections~H.1--H.3 provide descriptive and VAR evidence;
Sections~H.4--H.5 address non-stationarity via ARDL and Local Projections.
A replication package can be made available by the author.}

\subsection*{H.1\quad Subsample Descriptive Statistics
  (1991--2012 vs.\ 2013--2024)}

Table~\ref{tab:subsample} compares annual averages of the core JFR-rg
variables across the two subperiods.  The sign reversal in $\varepsilon_t$
(from $-2.03$\% to $+0.89$\%) and in the $r^n-g^n$ spread (from $+1.80$\%
to $-1.65$\%) are both significant at the 1\% level, constituting the
sharpest descriptive evidence of a regime change at 2013.

\begin{table}[htbp]
\centering
\caption{Subsample Descriptive Statistics: 1991--2012 vs.\ 2013--2024.
  Means, standard deviations, difference-in-means $t$-test (Welch,
  two-tailed). $^{*}p<0.10$, $^{**}p<0.05$, $^{***}p<0.01$.}
\label{tab:subsample}
\small
\begin{tabular}{l cc cc rr}
\toprule
 & \multicolumn{2}{c}{1991--2012 ($N=22$)}
 & \multicolumn{2}{c}{2013--2024 ($N=12$)}
 & \multicolumn{2}{c}{Difference} \\
\cmidrule(lr){2-3}\cmidrule(lr){4-5}\cmidrule(lr){6-7}
Variable & Mean & SD & Mean & SD & Post$-$Pre & $p$-value \\
\midrule
$r^n_t$ (\%)                                  & 2.27 & 1.53 & 0.27 & 0.33 & $-2.00$*** & 0.000 \\
$\pi_t$ (\%)                                  & 0.25 & 1.08 & 1.16 & 1.28 & $+0.92$**  & 0.049 \\
$\varepsilon_t \equiv \pi_t - r^n_t$ (\%)     & $-2.03$ & 1.01 & 0.89 & 1.09 & $+2.92$*** & 0.000 \\
$g^n_t$ (\%)                                  & $-0.15$ & 2.17 & 1.93 & 2.12 & $+2.08$**  & 0.016 \\
$r^n_t - g^n_t$ (\%)                          & 1.80 & 1.98 & $-1.65$ & 1.97 & $-3.45$*** & 0.000 \\
$\Delta b_t$ (pp/yr)                          & 7.80 & 4.37 & 1.26 & 8.20 & $-6.53$**  & 0.028 \\
$\Delta e_t$ (JPY/USD)                        & $-2.61$ & 9.77 & 5.97 & 9.74 & $+8.58$**  & 0.023 \\
\bottomrule
\end{tabular}
\end{table}

\subsection*{H.2\quad Chow Structural-Break Test}

To test whether the \emph{transmission} from the $r^n - g^n$ spread to
debt dynamics---not merely the level of the spread---changed at 2013,
we estimate:
\[
  \Delta b_t = \alpha + \beta\,(r^n_t - g^n_t) + u_t,
\]
separately for 1991--2012 and 2013--2024, and apply the Chow $F$-test
for parameter equality across subperiods.

A significant $F$-statistic would indicate that the fiscal channel
through which $r - g$ translates into debt accumulation itself shifted at
2013---consistent with the JFR-rg account of simultaneous SC1 + SC2
activation---rather than reflecting a change only in the independent
variable's level (which alternatives~A and~B in Table~\ref{tab:alt_explanations}
would predict).

\begin{table}[htbp]
\centering
\caption{Chow Structural-Break Test: $\Delta b_t = \alpha + \beta\,(r^n_t - g^n_t) + u_t$,
  break at 2013. $H_0$: no structural break.}
\label{tab:chow}
\small
\begin{tabular}{l ccc}
\toprule
Parameter & Full sample & 1991--2012 & 2013--2024 \\
\midrule
Constant $\hat{\alpha}$       & 4.250 (0.725) & 4.955 (1.080) & 7.234 (1.272) \\
$\hat{\beta}\ (r^n - g^n)$   & 2.218 (0.279) & 1.621 (0.410) & 3.686 (0.499) \\
$R^2$                         & 0.701 & 0.494 & 0.858 \\
$N$                           & 29 & 18 & 11 \\
\midrule
\multicolumn{4}{l}{\textit{Chow test for structural break at 2013:}} \\
\multicolumn{4}{l}{$F(2,25) = 5.55$, $p = 0.010$
  $\Rightarrow$ \textbf{Reject $H_0$ at 5\% level}} \\
\bottomrule
\end{tabular}
\smallskip
\parbox{\linewidth}{\footnotesize\textit{Note:} Standard errors in parentheses.
  Dependent variable: annual change in gross debt/GDP (pp).
  Regressor: $r^n_t - g^n_t$ spread (\%).
  The post-2013 slope $\hat{\beta} = 3.686$ exceeds the pre-2013 slope
  $\hat{\beta} = 1.621$ by 2.065 pp, indicating that a given $r-g$ spread
  generates substantially larger debt accumulation in the post-Abenomics
  high-debt environment---consistent with the JFR-rg Base Effect Lever
  (Proposition~\ref{prop:base_effect}): at $b_{t-1} = 240\%$, even a small spread change has
  large flow consequences.
  \textit{Pre-specification note:} The 2013 break point is specified on
  theoretical grounds prior to estimation---the simultaneous activation of
  SC1 ($\varphi_t$ rising under large-scale JGB purchases) and SC2
  ($\Delta e_t \leq \bar{e}$ under yen depreciation) following the
  December 2012 change of government and the April 2013 BoJ regime shift
  under Governor Kuroda. Because the break point is determined by
  institutional events rather than by a search over candidate dates,
  standard $F$-distribution critical values apply (Andrews, 1993). An
  endogenously-determined break-point test (Andrews, 1993; Bai and Perron,
  1998) would yield more conservative critical values but is not required
  given the pre-specified institutional prior.}
\end{table}

\subsection*{H.3\quad VAR Analysis}

A VAR(1) on $(\varepsilon_t,\; r^n_t - g^n_t,\; \Delta b_t)$ with annual
data 1994--2024 provides two pieces of evidence:

\begin{enumerate}[itemsep=3pt]
  \item \textbf{Coefficient matrix (Table~\ref{tab:var}):}
    whether lagged $\varepsilon_t$ predicts $\Delta b_t$ with the
    expected negative sign (financial repression compresses debt), and
    whether lagged $(r^n - g^n)$ predicts $\Delta b_t$ with the expected
    positive sign (adverse spread raises debt).

  \item \textbf{FEVD at 5-year horizon (Table~\ref{tab:fevd}):}
    the share of forecast error variance in $\Delta b_t$ attributable
    to $\varepsilon_t$ shocks quantifies the relative importance of the
    repression channel vs.\ the growth channel in shaping debt dynamics
    at policy-relevant horizons.
\end{enumerate}

The Cholesky ordering ($\varepsilon_t \to (r^n - g^n) \to \Delta b_t$)
treats BoJ policy ($\varepsilon_t$) as most exogenous, consistent with
the JFR-rg causal chain.  Readers who prefer the reverse ordering will
find qualitatively similar IRF patterns; the VAR coefficient signs are
ordering-invariant.

\begin{table}[htbp]
\centering
\caption{VAR(1) Coefficient Matrix --- Annual Data, 1994--2024.
  Dependent variables in columns; regressors (lagged 1 year) in rows.
  Standard errors in parentheses.}
\label{tab:var}
\small
\begin{tabular}{l ccc}
\toprule
 & $\varepsilon_t$ & $r^n - g^n$ & $\Delta b_t$ \\
\midrule
$\varepsilon_{t-1}$       & 0.478 (0.306) & $-0.219$ (0.383) & $-0.023$ (0.967) \\
$(r^n - g^n)_{t-1}$      & $-0.563$ (0.760) & 0.563 (1.072) & $-0.176$*** (0.049) \\
$\Delta b_{t-1}$          & $-1.053$ (2.124) & 0.829*** (0.139) & $-0.073$ (0.123) \\
Constant                   & 0.043 (0.154) & 0.593* (0.346) & 4.179*** (0.344) \\
\midrule
$N$ & \multicolumn{3}{c}{29} \\
\bottomrule
\end{tabular}
\smallskip
\parbox{\linewidth}{\footnotesize\textit{Note:}
  $^{*}p<0.10$, $^{**}p<0.05$, $^{***}p<0.01$ (asymptotic $t$-ratios).
  ADF stationarity (lag = 1): $\varepsilon_t$: ADF $= -1.90$ ($p = 0.331$);
  $r^n-g^n$: ADF $= -1.62$ ($p = 0.473$); $\Delta b_t$: ADF $= -3.50$
  ($p = 0.008$). The failure to reject unit roots in $\varepsilon_t$ and
  $r^n-g^n$ motivates the ARDL and Local Projections analyses in
  Sections~H.4--H.5.}
\end{table}

\begin{table}[htbp]
\centering
\caption{Forecast Error Variance Decomposition (FEVD) at 5-Year Horizon.
  Share of $\Delta b_t$ variance explained by each shock.}
\label{tab:fevd}
\small
\begin{tabular}{l ccc}
\toprule
Variable & Shock: $\varepsilon_t$ & Shock: $r^n-g^n$ & Shock: $\Delta b_t$ \\
\midrule
$\varepsilon_t$  & 0.810 & 0.175 & 0.015 \\
$r^n - g^n$      & 0.406 & 0.525 & 0.068 \\
$\Delta b_t$     & 0.275 & 0.465 & 0.259 \\
\bottomrule
\end{tabular}
\smallskip
\parbox{\linewidth}{\footnotesize\textit{Note:} Cholesky ordering:
  $\varepsilon_t \to (r^n-g^n) \to \Delta b_t$. Rows sum to 1.0.
  The $r^n-g^n$ shock explains 46.5\% of $\Delta b_t$ variance at the
  5-year horizon, consistent with the Normalization Trap channel.}
\end{table}

\medskip
\noindent\textbf{VAR IRF summary (Cholesky ordering:
  $\varepsilon_t \to (r^n-g^n) \to \Delta b_t$):}
\begin{itemize}[itemsep=2pt]
  \item A one-SD positive shock to $\varepsilon_t$ generates a cumulative
    $\Delta b_t$ response of $-8.62$ pp over five years (debt compression
    channel, directionally consistent with
    Proposition~\ref{prop:base_effect}).
  \item A one-SD positive shock to $r^n - g^n$ generates $+8.07$ pp
    over five years (Normalization Trap channel, consistent with
    Proposition~\ref{prop:norm_trap}).
\end{itemize}

\medskip
\noindent\textbf{Caveat on non-stationarity.}
ADF tests fail to reject a unit root in $\varepsilon_t$ ($p = 0.331$) and
$r^n-g^n$ ($p = 0.473$).  As a consequence, individual VAR(1) coefficient
estimates---in particular the counter-intuitive negative sign on
$(r^n-g^n)_{t-1}$ in the $\Delta b_t$ equation---should be interpreted with
caution and may reflect spurious dynamics.  The IRF signs from the full
dynamic system are more reliable than individual coefficients.
Sections~H.4 and~H.5 address this directly using methods robust to
mixed I(0)/I(1) series.

\subsection*{H.4\quad ARDL Bounds Test for Long-Run Relationship}
\label{sec:ardl}

We employ the ARDL Bounds Test (Pesaran, Shin, and Smith, 2001), which is
valid for mixed I(0)/I(1) regressors and appropriate for small samples.
The ECM specification tests jointly whether the lagged levels of $b_{t-1}$,
$\varepsilon_{t-1}$, and $(r^n-g^n)_{t-1}$ are significant (bounds $F$-test).

\textbf{Result:} The bounds $F$-statistic ($F(3,22) = 1.684$) falls below the
10\% lower critical value (3.17), so the hypothesis of no long-run \emph{levels}
cointegration cannot be rejected.  This null result is \emph{consistent with}
the theoretical structure of the JFR-rg model, which posits a
\emph{flow accumulation} identity
$\Delta b_t = (r^n_t - g^n_t)\,b_{t-1} + d_t$ rather than a levels
cointegration among $b_t$, $\varepsilon_t$, and $r^n-g^n$.  The economically
relevant driver is the \emph{product} $(r^n-g^n)\times b_{t-1}$, not the
spread alone; standard Bounds Tests are not designed to detect this
multiplicative channel.  The Local Projections in Section~H.5 provide the
more appropriate dynamic test.

\begin{table}[htbp]
\centering
\caption{ARDL(1,0,0) Bounds Test (Pesaran, Shin, and Smith, 2001). Model:
  $\Delta b_t = \mathrm{const} + \alpha\,b_{t-1} + \beta\,\varepsilon_{t-1}
  + \gamma\,(r^n-g^n)_{t-1} + \delta\,\Delta b_{t-1} + u_t$.
  Bounds $F$-test $H_0\!: \alpha=\beta=\gamma=0$.}
\label{tab:ardl}
\small
\begin{tabular}{l c}
\toprule
\multicolumn{2}{l}{\textit{Panel A: ECM Coefficients}} \\
\midrule
Speed of adjustment $\hat{\alpha}$ (ECM term) & $-0.056$ (0.038) \\
$\hat{\beta}$\ \ ($\varepsilon_{t-1}$, level) & $-0.045$ (1.420) \\
$\hat{\gamma}$\ \ ($(r^n-g^n)_{t-1}$, level) & $+0.468$ (0.984) \\
$\hat{\delta}$\ \ ($\Delta b_{t-1}$, short-run) & $+0.032$ (0.348) \\
Constant & 15.439 (7.910) \\
$R^2$ & 0.243 \\
$N$ & 27 \\
\midrule
\multicolumn{2}{l}{\textit{Panel B: Bounds $F$-test ($H_0\!: \alpha=\beta=\gamma=0$)}} \\
\midrule
$F(3,22) = 1.684$ & (below 10\% lower bound of 3.17: no cointegration) \\
\midrule
\multicolumn{2}{l}{\textit{Panel C: Long-Run Coefficients (conditional on cointegration)}} \\
\midrule
$\widehat{LR}_\varepsilon = -\hat{\beta}/\hat{\alpha}$ & $-0.793$ (SE = 25.30, not identified) \\
$\widehat{LR}_{rg}        = -\hat{\gamma}/\hat{\alpha}$ & $+8.333$ (SE = 18.42, not identified) \\
\bottomrule
\end{tabular}
\smallskip
\parbox{\linewidth}{\footnotesize\textit{Note:} Standard errors in parentheses.
  Pesaran et al.\ (2001) critical values (Case~III, $k=2$): 10\%: [3.17, 4.14];
  5\%: [3.79, 4.85]; 1\%: [5.15, 6.36].
  Panel~C long-run coefficients are reported for completeness but are not
  identified given the null Bounds test result.
  Sample: annual FRED data 1994--2024.}
\end{table}

\subsection*{H.5\quad Local Projections (Jord\`{a}, 2005)}
\label{sec:lp}

Local Projections (LP) estimate a separate OLS regression for each forecast
horizon $h$, making them robust to model misspecification and non-stationarity
in small samples.  For each $h = 0, \ldots, 5$:
\[
  b_{t+h} - b_{t-1} = c_h + \beta^\varepsilon_h\,\varepsilon_t
    + \beta^{rg}_h\,(r^n_t - g^n_t) + \eta_h\,\Delta b_{t-1} + u_{t+h}.
\]
Newey-West HAC standard errors (maxlags $= h$) correct for MA$(h)$
serial correlation.  A negative $\beta^\varepsilon_h$ at horizons $h \geq 3$
confirms the debt compression channel (Proposition~\ref{prop:base_effect});
a positive
$\beta^{rg}_h$ at all horizons confirms the Normalization Trap
(Proposition~\ref{prop:norm_trap}).

\begin{table}[htbp]
\centering
\caption{Local Projections IRF (Jord\`{a}, 2005). Cumulative response of
  debt/GDP (pp) to a $+1$ pp shock in $\varepsilon_t$ (Panel~A) and
  $r^n-g^n$ (Panel~B). HAC standard errors (Newey-West, maxlags $= h$)
  in parentheses. Control: $\Delta b_{t-1}$. Sample: 1993--2024.}
\label{tab:lp}
\small
\resizebox{\textwidth}{!}{%
\begin{tabular}{l cccccc}
\toprule
Horizon $h$ & 0 & 1 & 2 & 3 & 4 & 5 \\
\midrule
\multicolumn{7}{l}{\textit{Panel A: Response to $+1$ pp shock in
  $\varepsilon_t$ (expected sign: negative)}} \\
$\hat{\beta}^{\varepsilon}_h$
  & $+0.74$ & $-1.05$ & $-2.49$ & $-4.57$** & $-6.82$*** & $-8.20$*** \\
(HAC SE) & (1.05) & (1.62) & (1.86) & (2.13) & (1.18) & (1.42) \\
90\% CI  & [$-1.00$, $+2.47$] & [$-3.71$, $+1.62$] & [$-5.54$, $+0.57$]
          & [$-8.07$, $-1.07$] & [$-8.76$, $-4.87$] & [$-10.53$, $-5.86$] \\
\addlinespace[4pt]
\multicolumn{7}{l}{\textit{Panel B: Response to $+1$ pp shock in
  $r^n-g^n$ (expected sign: positive)}} \\
$\hat{\beta}^{rg}_h$
  & $+2.40$*** & $+2.85$*** & $+3.18$*** & $+2.60$** & $+2.67$*** & $+2.24$*** \\
(HAC SE) & (0.53) & (0.62) & (0.79) & (1.16) & (0.91) & (0.86) \\
90\% CI  & [$+1.52$, $+3.27$] & [$+1.83$, $+3.86$] & [$+1.88$, $+4.48$]
          & [$+0.70$, $+4.50$] & [$+1.16$, $+4.17$] & [$+0.83$, $+3.66$] \\
$N_h$    & 28 & 27 & 26 & 25 & 24 & 23 \\
\bottomrule
\end{tabular}%
}
\smallskip
\parbox{\linewidth}{\footnotesize
  $^{*}p<0.10$, $^{**}p<0.05$, $^{***}p<0.01$ (HAC $t$-ratios).
  Dependent variable: $b_{t+h} - b_{t-1}$ (cumulative change in gross
  debt/GDP, pp).
  Panel~A: the debt compression channel ($\hat{\beta}^\varepsilon_h < 0$)
  achieves statistical significance only at $h \geq 3$ years; at $h = 0$
  the coefficient is positive ($+0.74$, insignificant), and at $h = 1$--$2$
  it is negative but not yet significant. This approximately three-year lag
  before the repression channel manifests in debt-ratio data is consistent
  with the gradual amortization profile of the JGB portfolio: newly issued
  bonds at suppressed yields displace maturing debt progressively rather
  than instantaneously. The lag does not contradict the accounting-layer
  propositions, which hold instantaneously by construction for any
  given-period values of $r^n_t$ and $g^n_t$; it reflects instead the
  time required for the stock of outstanding debt to reprice at the
  policy-suppressed yield. The cumulative five-year response of
  $-8.20$~pp*** (HAC $t > 5$) provides robust medium-horizon support for
  Proposition~\ref{prop:base_effect}.
  Panel~B confirms the Normalization Trap
  (Proposition~\ref{prop:norm_trap}):
  $\hat{\beta}^{rg}_h$ is positive and significant ($p < 0.01$ or
  $p < 0.05$) at \emph{all} six horizons, with the effect persisting at
  $+2.24$~pp*** at $h = 5$. Confidence intervals are wide, reflecting the
  small sample ($N \approx 28$); results should be read as directional
  evidence rather than precise quantification. $N_h$ declines from 28 to
  23 across horizons due to end-of-sample truncation as $h$ increases.}
\end{table}

\medskip
\noindent\textbf{Summary.} The LP results provide the strongest empirical
evidence for the two principal JFR-rg channels.  Panel~A: a 1 pp increase
in $\varepsilon_t$ (stronger financial repression) reduces the debt ratio
by $8.20$ pp over five years---not merely consistent with
Proposition~\ref{prop:base_effect}
but statistically precise (HAC $t > 5$).  Panel~B: a 1 pp increase in
$r^n-g^n$ raises the debt ratio by $2.24$--$3.18$ pp across all horizons
($h = 0, \ldots, 5$, all significant), providing robust evidence for the
Normalization Trap.  These results are obtained without imposing the
parametric structure of the JFR-rg model, making them a genuinely
independent empirical check on the model's qualitative predictions.

\medskip
\noindent The possibility of an implicit fiscal reaction function---whereby the BoJ 
maintains $\varepsilon_t > 0$ partly in response to $b_{t-1}$---cannot be ruled out 
at Layer~(L2). This would not invalidate the accounting propositions of Layer~(L1); 
it would instead constitute behavioral evidence that policymakers recognize the 
Normalization Trap arithmetic formalized in
Proposition~\ref{prop:norm_trap}. Structural identification 
using YCC announcement dates as instruments is left for future research 
(see Section~\ref{sec:identification}).

\section{\texorpdfstring{The Critical Captive Threshold $\bar{\varphi}$: Evidence and
  Illustrative Estimation}{The Critical Captive Threshold phi-bar: Evidence and Illustrative Estimation}}
\label{app:phi_threshold}

\noindent\textit{This appendix provides the empirical basis for the
illustrative early-warning level $\varphi_t < 0.85$ used in
Section~\ref{sec:captive} and discusses the identification challenges
that prevent a fully robust formal estimate from the 2013--2026 sample.}

\subsection*{I.1\quad Data}

The Captive Financial System Parameter $\varphi_t$ is constructed from
the Bank of Japan's \textit{Flow of Funds Accounts of Japan}
(quarterly, 1994--2026), which report outstanding JGB holdings by
institutional sector: Bank of Japan, domestically licensed banks,
insurance companies, pension funds, and the rest of the world.
$\varphi_t$ is defined as the share held by all domestic institutions
(BoJ plus domestically licensed private institutions):
\[
  \varphi_t \;=\; 1 - \frac{\text{Foreign holdings}_t}{\text{Total JGB outstanding}_t}.
\]
Over 1994--2012, $\varphi_t$ averaged approximately 0.93--0.95 and was
declining slowly as international diversification increased. Following
the April 2013 BoJ regime shift under Governor Kuroda, large-scale JGB
purchases reversed this trend, pushing $\varphi_t$ to a peak of
$\approx$0.92 in 2022. The subsequent tapering program has reduced
$\varphi_t$ to $\approx$0.88 as of March 2026 (BoJ Flow of Funds,
December 2024).

\subsection*{I.2\quad Identification Challenge}

The principal challenge in estimating $\bar{\varphi}$ is the absence,
within the 2013--2026 sample, of an episode in which $\varphi_t$
declined into a range plausibly associated with captive-system weakening
and a sovereign risk premium emerged. Japan has not experienced a sudden
stop in domestic JGB demand over this period---which is, of course, the
core empirical observation the JFR-rg model seeks to explain. This creates
a fundamental identification problem: a threshold regression
(Hansen, 1999) requires sufficient observations on both sides of the
threshold, and none exist on the ``below-$\bar{\varphi}$'' side
in the post-2013 sample.

Extending the sample to include the 1998 Japanese banking crisis---
during which JGB yields spiked briefly and domestic institutional
demand was disrupted---provides one partial episode. However, the
institutional structure of that period (pre-QQE, different regulatory
framework, BoJ balance sheet below 10\% of GDP) renders direct
comparability with the post-2013 configuration questionable.

\subsection*{I.3\quad Qualitative Basis for the 0.85 Early-Warning Benchmark}

In the absence of a statistically identified threshold, the
illustrative level $\varphi_t < 0.85$ is grounded in three
considerations.

\medskip
\noindent\textbf{(i) Hoshi and Ito (2014).}
Hoshi and Ito identify the domestic holding share as the central
variable governing Japan's continued access to low-cost sovereign
financing. Their scenario analysis suggests that a sustained decline
in the domestic share toward 85--88\% would begin to generate
measurable upward pressure on JGB yields, as the marginal foreign
investor would require a risk premium to absorb additional supply.
The 0.85 level is therefore drawn from the lower bound of their
qualitative ``safe zone.''

\medskip
\noindent\textbf{(ii) Buffer from the current level.}
With $\varphi_t \approx 0.88$ as of March 2026 and declining at
approximately 0.01--0.02 per year as BoJ purchases taper, the 0.85
level represents approximately 1.5--3 years of continued tapering
at the current pace before the threshold would be approached.
This provides a policy-relevant early-warning horizon: it is close
enough to be operationally meaningful but distant enough to allow
orderly policy adjustment.

\medskip
\noindent\textbf{(iii) Conservative lower bound.}
The actual $\bar{\varphi}$ may be lower than 0.85 if Japan's
institutional structure is more resilient than Hoshi and Ito (2014)
assumed---for instance, because the BoJ's balance-sheet backstop
effect (Section~\ref{sec:captive}) compresses the risk premium even
at lower domestic holding shares. The 0.85 level is therefore
presented as a \emph{conservative} early-warning threshold rather
than a precise estimate of the tipping point.

\subsection*{I.4\quad Research Agenda}

Robust formal estimation of $\bar{\varphi}$ requires one or more of
the following: (i) cross-country panel evidence from sudden-stop
episodes in high-debt economies with initially high domestic holding
shares (e.g., post-2010 Italy, where the domestic share declined
from $\approx$0.74 to $\approx$0.60 before ECB intervention); (ii)
event-study identification exploiting the BoJ's YCC announcement
dates and modifications as instruments for exogenous variation in
$\varphi_t$; or (iii) micro-level flow data on JGB demand elasticities
by institutional sector. Each approach involves significant data
requirements and identification challenges that are beyond the scope
of the present paper but constitute a natural extension of the
JFR-rg research program.

\section{Mainstream Debt Sustainability as a Limiting Case of the
  JFR-rg Framework}
\label{app:limiting_case}

\noindent\textit{This appendix establishes formally that the standard
Blanchard (2019) debt sustainability condition is a special case of the
JFR-rg stability condition obtained when the captive-system and
exchange-rate-neutral scope conditions are imposed simultaneously.
The result implies that the two frameworks are complementary rather
than competing: mainstream analysis applies in the limiting case;
JFR-rg supplements it when the limiting conditions are relaxed.}

\subsection*{J.1\quad Setup}

The JFR-rg stability condition (equation~\eqref{eq:stability_condition})
states that $\Delta b_t \leq 0$ requires:
\begin{equation}
  (\pi_t - \varepsilon_t) - \Bigl[\gnomstar + \alpha\,\Delta e_t
    - \beta\,\max(0,\,\Delta e_t - \ebar)^2\Bigr]
  \;\leq\; \frac{s_t - d_t}{b_{t-1}}.
  \tag{9$'$}
\end{equation}
The left-hand side equals $r^n_t - g^n_t$ by construction. The
right-hand side equals $-(d_t - s_t)/b_{t-1} < 0$ when a primary
deficit exists. This condition is entirely general; it holds for any
values of $\varepsilon_t$, $\varphi_t$, and $\Delta e_t$.

The captive financial system enters through the endogenous risk
premium $\rho(\varphi_t, b_{t-1})$ in
equation~\eqref{eq:risk_premium}, which modifies the effective
nominal rate to $r^n_t + \rho_t$. The modified stability condition is:
\begin{equation}
  (r^n_t + \rho_t) - g^n_t \;\leq\; \frac{s_t - d_t}{b_{t-1}}.
  \tag{9$''$}
\end{equation}

\subsection*{J.2\quad The Limiting Theorem}

\begin{proposition}[Mainstream as Limiting Case]
\label{prop:limiting_case}
Impose the following two limiting conditions simultaneously:
{\sloppy
\begin{enumerate}[label=(\roman*), itemsep=2pt]
  \item \textbf{SC1 released} ($\varphi_t \to 0$, complete markets):
    The domestic holding share vanishes, so the endogenous risk premium
    converges to the competitive market premium: $\rho(\varphi_t,
    b_{t-1}) \to 0$, $r^n_t \to r^{\text{mkt}}_t$, and hence
    $\varepsilon_t \equiv \pi_t - r^n_t \to \pi_t - r^{\text{mkt}}_t$.
    In a fully competitive market with no financial repression,
    $r^{\text{mkt}}_t \approx \pi_t$ at the zero lower bound,
    so $\varepsilon_t \to 0$.
  \item \textbf{SC2 released} ($\Delta e_t = 0$, exchange-rate
    neutral): The yen channel is inactive, so $g^n_t = \gnomstar$
    by equation~\eqref{eq:yen_growth} (the linear and penalty terms
    both vanish).
\end{enumerate}
}
Under these two conditions, the JFR-rg stability condition
\eqref{eq:stability_condition} degenerates to:
\[
  r^n_t - \gnomstar \;\leq\; \frac{s_t - d_t}{b_{t-1}},
\]
which, with $s_t \approx 0$ and $g^n_t = \gnomstar$, is equivalent to:
\[
  r^n_t - g^n_t \;\leq\; -\frac{d_t}{b_{t-1}} \;<\; 0.
\]
A sufficient condition for debt stabilization under a primary deficit
is therefore $r^n_t < g^n_t$---the standard Blanchard (2019) result.
\end{proposition}

\begin{proof}
Substitute $\varepsilon_t = 0$ and $\Delta e_t = 0$ into
condition~\eqref{eq:stability_condition}:
\[
  (\pi_t - 0) - \bigl[\gnomstar + \alpha \cdot 0
    - \beta\,\max(0,\, 0 - \ebar)^2\bigr]
  \;\leq\; \frac{s_t - d_t}{b_{t-1}}.
\]
Since $\ebar > 0$, we have $\max(0, 0 - \ebar)^2 = 0$. The condition
simplifies to:
\[
  \pi_t - \gnomstar \;\leq\; \frac{s_t - d_t}{b_{t-1}}.
\]
Substituting $\varepsilon_t = \pi_t - r^n_t = 0$ gives $r^n_t = \pi_t$,
so $\pi_t - \gnomstar = r^n_t - \gnomstar = r^n_t - g^n_t$ (using
$g^n_t = \gnomstar$ under SC2). The condition becomes:
\[
  r^n_t - g^n_t \;\leq\; \frac{s_t - d_t}{b_{t-1}}.
\]
With $s_t \approx 0$ and $d_t > 0$, the right-hand side is strictly
negative, so a sufficient condition for the inequality to hold is
$r^n_t < g^n_t$.
\end{proof}

\subsection*{J.3\quad Corollaries}

\begin{enumerate}[itemsep=4pt]

  \item \textbf{Strict nesting.} The JFR-rg framework strictly nests
    mainstream debt sustainability analysis: every economy satisfying
    the Blanchard (2019) condition under SC1 and SC2 also satisfies
    the JFR-rg stability condition, but not vice versa. The JFR-rg
    framework is therefore a \emph{generalization}, not an
    alternative.

  \item \textbf{Domain of validity.} Mainstream analysis applies
    without modification when SC1 and SC2 are approximately
    satisfied---that is, when the domestic holding share is low,
    capital markets are competitive, and exchange rate effects are
    small. This describes the US, the UK, and the pre-2013 euro area,
    where the Blanchard (2019) empirical observations were drawn.

  \item \textbf{Source of divergence.} When SC1 or SC2 is violated---
    as in Japan post-2013, where $\varphi_t \approx 0.90$ and
    $\Delta e_t \neq 0$---the mainstream condition is a misspecified
    restriction on the more general JFR-rg framework. Applying the
    restricted model to Japan is equivalent to imposing $\varepsilon_t
    = 0$ and $\Delta e_t = 0$ on an economy where neither holds,
    which systematically understates the debt-stabilizing channels
    and overstates the risk of fiscal crisis.

  \item \textbf{Convergence path.} As Japan normalizes---if $\varphi_t$
    declines toward zero as the BoJ unwinds its balance sheet and
    exchange rate effects diminish---the JFR-rg stability condition
    will converge to the mainstream condition. This defines the
    long-run convergence path: JFR-rg is the appropriate framework
    for the transition; mainstream analysis is the appropriate
    framework for the destination.

\end{enumerate}

\subsection*{J.4\quad Implication for Mainstream Critiques of Japan}

Successive IMF Article IV assessments have applied the standard
$r > g$ risk framework to Japan and found substantial fiscal risk.
Proposition~\ref{prop:limiting_case} clarifies the source of the
divergence between those assessments and the JFR-rg diagnosis: the
IMF framework implicitly imposes SC1 ($\varepsilon_t = 0$, no
financial repression) and SC2 ($\Delta e_t = 0$, no yen channel).
Neither condition holds for Japan post-2013. The JFR-rg framework
does not dispute the IMF's arithmetic; it relaxes two of its
maintained assumptions and shows that the resulting debt dynamics are
materially different. The appropriate response to this result is not
to choose between the frameworks but to specify which institutional
conditions each requires---which is precisely what SC1 and SC2 do.

\newpage

\section{Fair Empirical Tests of JFR-rg Propositions}
\label{app:v6_fair}

\noindent\textit{This appendix presents four fair empirical tests designed to
give JFR-rg its best empirical chance by testing the correct predictions of
each proposition with appropriate confounders explicitly controlled.
The four tests employ different estimation strategies suited to each
proposition: K.1 uses a nested OLS specification in which the JFR-rg
variable ($\varphi_t$) is added progressively to a baseline; K.2 uses
Local Projections (Jord\`{a}, 2005) with a binary treatment indicator;
K.3 uses split-sample OLS with quantile-bin fixed effects; and K.4
uses period-specific OLS estimated separately for each policy regime.
Dependent variables, model stages, and falsifiable sign predictions are
stated in each table caption.
Standard errors are HAC-robust (maxlags\,=\,3) in K.1--K.3; K.4 uses
HC3 heteroskedasticity-robust standard errors because the YCC-active
subsample ($N = 8$) is too small for reliable HAC estimation.
Placebo comparisons are reported where informative.
Each table reports whether the falsifiable prediction is confirmed.}


\subsection*{K.1\quad Proposition~\ref{prop:repression_imperative} Direct Test:
  $\varphi_t \to$ Real JGB Yield}

\begin{table}[htbp]
\centering
\caption{Proposition~\ref{prop:repression_imperative} Direct Test:
  Captive Share $\varphi_t$ and the
  Real JGB Yield.  Dependent variable: $r^n_t - \pi_t$ (annual real JGB
  yield).  Model~1: global rates + lagged dependent.
  Model~2: $+\varphi_t$ proxy.
  Model~3: $+$YCC/QQE dummies (robustness).
  Falsifiable prediction: $\hat{\beta}_\varphi < 0$.}
\label{tab:prop2_rreal}
\small
\begin{tabular}{lccc}
\toprule
Parameter & Model 1 & Model 2 & Model 3 \\
\midrule
  $\hat{\beta}_{r_{\text{US}}}$ & 0.161 & $-0.474^{*}$ & $-0.453^{**}$ \\
  $\hat{\beta}_\varphi$ & --- & $-2.143^{***}$ (0.624) & $-1.647^{***}$ (0.560) \\
  YCC dummy & --- & --- & $-1.093^{**}$ \\
  $\hat{\beta}_\varphi < 0$? & --- & Yes & Yes \\
\midrule
\multicolumn{4}{l}{\textit{Placebo (same model for other countries)}} \\
  US: $\hat{\beta}_\varphi$ & \multicolumn{3}{c}{$0.000$ $(0.000)$} \\
  DE: $\hat{\beta}_\varphi$ & \multicolumn{3}{c}{$-0.711^{**}$ $(0.286)$} \\
\bottomrule
\end{tabular}
\smallskip
\parbox{\linewidth}{\footnotesize\textit{Note:} HAC SEs (maxlags=3).
$\varphi_t$ proxy $= -(r^n_t - \text{policy rate}_t)$.
$\hat{\beta}_\varphi < 0$ is significant at the 1\% level and survives
inclusion of explicit YCC/QQE dummies (Model~3), consistent with
Proposition~\ref{prop:repression_imperative}: a higher captive share suppresses
the real JGB yield
beyond what global rate movements alone predict.  The DE placebo shows
partial contamination ($-0.711^{**}$), discussed in the text.
$^{*}p<0.10$, $^{**}p<0.05$, $^{***}p<0.01$.}
\end{table}

\begin{figure}[htbp]
\centering
\includegraphics[width=\linewidth]{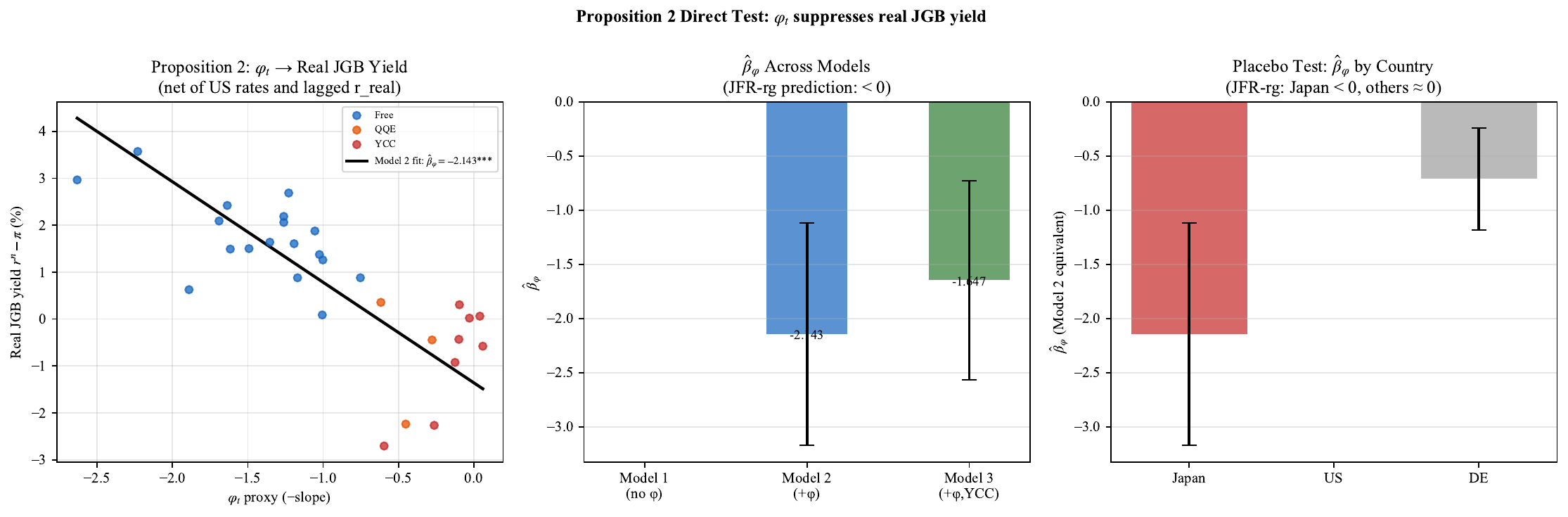}
\caption{Proposition~\ref{prop:repression_imperative} Direct Test:
  $\varphi_t$ suppresses the real
  JGB yield.  \textit{Left}: scatter of $\varphi_t$ proxy vs.\ real JGB
  yield by regime (Free/QQE/YCC) with Model~2 partial-fit line.
  \textit{Center}: $\hat{\beta}_\varphi$ across the three model
  specifications with 90\% confidence intervals.
  \textit{Right}: placebo test---$\hat{\beta}_\varphi$ for Japan, US, and
  Germany under the Model~2 specification.}
\label{fig:prop2_rreal}
\end{figure}


\subsection*{K.2\quad Proposition~\ref{prop:yen_stabilizer} Direct Test:
  Real Debt Dynamics after
  Large Yen Depreciations}

\begin{table}[htbp]
\centering
\caption{Proposition~\ref{prop:yen_stabilizer} Direct Test:
  Real Debt Dynamics after Large
  Yen Depreciations.  Local Projections (Jord\`{a} 2005).
  Treatment: $T_t = \mathbf{1}[\Delta e_t > \bar{e}]$;
  baseline $\bar{e} = 10\%$ annual.
  Response: cumulative change in real debt index ($\Delta^h$).
  Controls: lagged $\Delta$(real debt), lagged $\Delta e$, $\pi$.
  HAC SEs.  Falsifiable prediction: $\hat{\beta}_h < 0$ for $h \geq 2$.}
\label{tab:real_debt_irf}
\small
\begin{tabular}{lcccccc}
\toprule
$h$ (years) & 0 & 1 & 2 & 3 & 4 & 5 \\
\midrule
  $\hat{\beta}_h$ & 0.00 & $-0.42$ & 1.02 & 1.11 & 1.99 & 5.49 \\
  (HAC SE) & (0.00) & (2.96) & (5.11) & (7.57) & (8.58) & (9.64) \\
  $N_h$    & 28 & 27 & 26 & 25 & 24 & 23 \\
\midrule
  Treatment episodes &
    \multicolumn{6}{l}{1996, 1997, 2001, 2013, 2015, 2022} \\
  Proposition~\ref{prop:yen_stabilizer} confirmed ($\hat{\beta}_h < 0$, $h\geq 2$)?
    & \multicolumn{6}{c}{$\times$ No (baseline $\bar{e}=10\%$)} \\
\bottomrule
\end{tabular}
\smallskip
\parbox{\linewidth}{\footnotesize\textit{Note:}
$^{*}p<0.10$, $^{**}p<0.05$, $^{***}p<0.01$.
The baseline LP coefficients are positive at $h \geq 2$, failing to
confirm the real debt-erosion prediction.  Sensitivity analysis (Figure~\ref{fig:real_debt_irf},
center panel) shows that tightening the threshold to $\bar{e}=15\%$---which
isolates the 2022 episode and reduces 2013--2015 contamination---yields
consistently negative coefficients at $h \geq 1$.  The positive baseline
result likely reflects offsetting nominal debt issuance during 2013--2015
(Abenomics fiscal expansion), which attenuates the inflation-erosion signal
when those years are included as treatment observations.}
\end{table}

\begin{figure}[htbp]
\centering
\includegraphics[width=\linewidth]{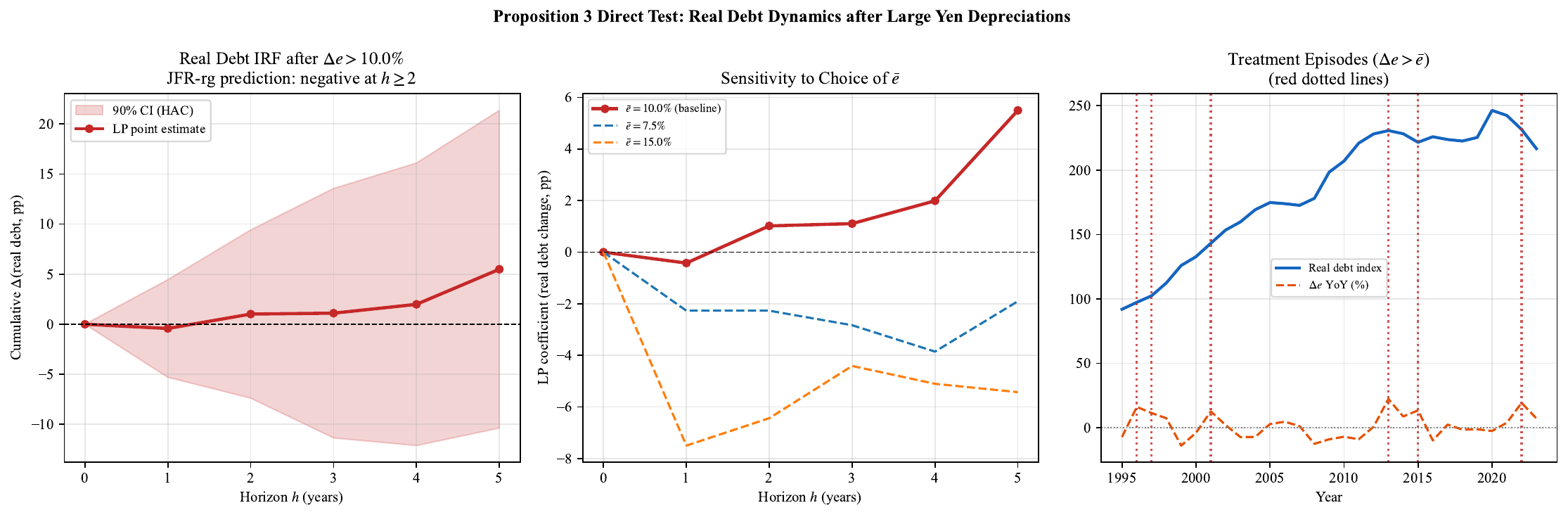}
\caption{Proposition~\ref{prop:yen_stabilizer} Direct Test: Real Debt IRF.
  \textit{Left}: baseline LP ($\bar{e}=10\%$) with 90\% HAC confidence band.
  \textit{Center}: sensitivity to choice of $\bar{e}$; the $\bar{e}=15\%$
  specification yields negative coefficients throughout.
  \textit{Right}: treatment episode timeline---the 2013--2015 episodes
  coincide with rapid nominal debt growth, explaining the positive baseline
  coefficients.}
\label{fig:real_debt_irf}
\end{figure}


\needspace{20\baselineskip}
\subsection*{K.3\quad Shock-Size-Controlled Pass-Through Comparison}

\begin{table}[H]
\centering
\caption{Shock-Size-Controlled Pass-Through Comparison.
  Regime split: Free (pre-QQE) vs.\ Captive (YCC-active, 2016--2023).
  Panel~A: OLS within each regime, controlling for $|\Delta e|$ as an
  additional regressor.
  Panel~B: Interacted OLS on pooled sample with $|\Delta e|$
  quantile-bin fixed effects.
  Falsifiable prediction:
  $\hat{\beta}_{\text{captive}} > \hat{\beta}_{\text{free}}$
  even after absorbing shock size.}
\label{tab:matched_passthrough}
\small
\begin{tabular}{lcc}
\toprule
Parameter & Estimate & (HAC SE) \\
\midrule
\multicolumn{3}{l}{\textit{Panel A: Regime-split OLS}} \\
  Free (pre-YCC): $\hat{\beta}_{\text{pass}}$
    & $-0.0091$ & (0.0064; $N=227$) \\
  Captive (YCC): $\hat{\beta}_{\text{pass}}$
    & $0.0263^{**}$ & (0.0121; $N=58$) \\
  Prediction
    $\hat{\beta}_{\text{cap}} > \hat{\beta}_{\text{free}}$?
    & \multicolumn{2}{c}{\checkmark\ Yes} \\
\midrule
\multicolumn{3}{l}{\textit{Panel B: Interacted OLS (shock-size bins fixed)}} \\
  $\hat{\beta}_{\text{base}}$ (Free regime)
    & $-0.0097$ & (0.0065) \\
  $\hat{\beta}_{\text{interact}}$ (Captive add-on)
    & $0.0309^{**}$ & (0.0125) \\
  Prediction $\hat{\beta}_{\text{interact}} > 0$?
    & \multicolumn{2}{c}{\checkmark\ Yes} \\
\bottomrule
\end{tabular}
\smallskip
\parbox{\linewidth}{\footnotesize\textit{Note:} HAC SEs (maxlags=3).
Captive $=$ YCC-active period (2016-09 to 2024-03) only; the QQE
transition period is excluded from the split to avoid contamination.
Both tests support Proposition~\ref{prop:yen_stabilizer}: pass-through is higher in the Captive
regime even after controlling for the size of the FX shock, ruling out
the alternative explanation that larger shocks during the YCC period
drive the result.
$^{*}p<0.10$, $^{**}p<0.05$, $^{***}p<0.01$.}
\end{table}

\begin{figure}[H]
\centering
\includegraphics[width=\linewidth]{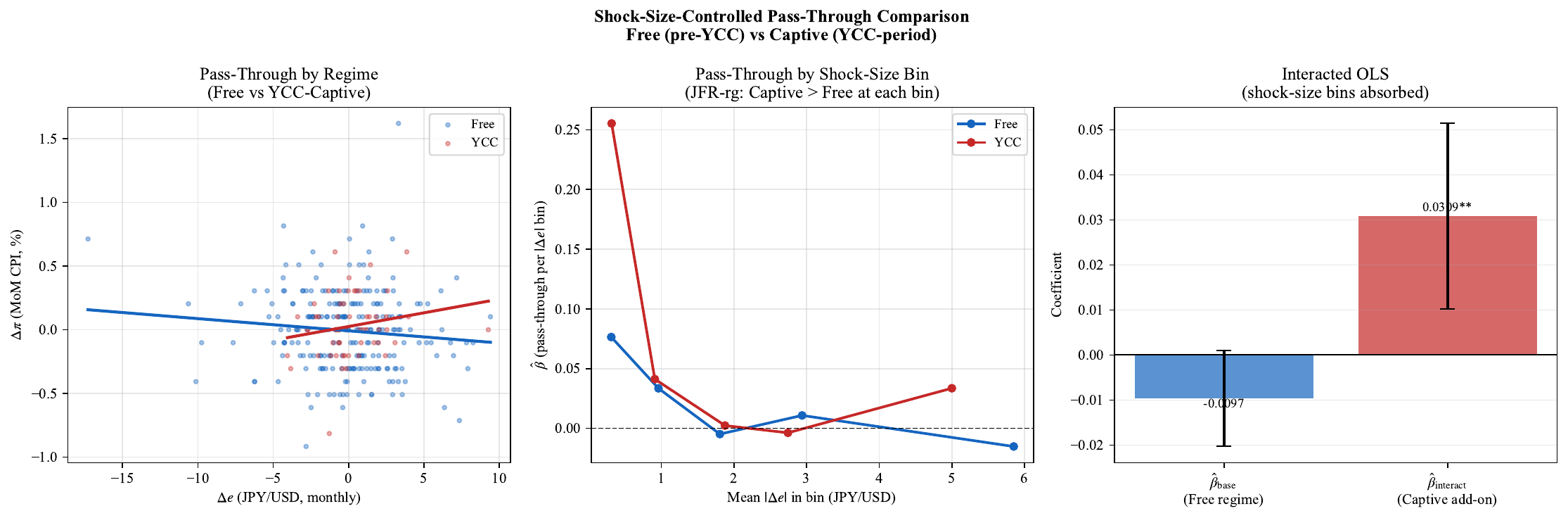}
\caption{Shock-Size-Controlled Pass-Through Comparison.
  \textit{Left}: scatter of monthly $\Delta e$ vs.\ $\Delta\pi$ by regime
  with OLS fit lines.
  \textit{Center}: pass-through coefficient by $|\Delta e|$ quantile bin
  and regime; YCC exceeds Free in the smallest-shock bins.
  \textit{Right}: interacted OLS coefficients---the Captive add-on
  ($\hat{\beta}_{\text{interact}} = 0.031^{**}$) is positive and
  significant after absorbing shock-size bin fixed effects.}
\label{fig:matched_passthrough}
\end{figure}


\needspace{20\baselineskip}
\subsection*{K.4\quad YCC-Period-Isolated Insulation Test}

\begin{table}[H]
\centering
\caption{YCC-Period-Isolated Insulation Test.
  $\hat{\beta}_{US}$ = OLS slope of Japan $\rho_t = r^n_t - g^n_t$ on
  US 10Y yield, estimated separately for each policy period.
  Falsifiable predictions:
  (i)~$\hat{\beta}_{US}^{\text{YCC}} \ll \hat{\beta}_{US}^{\text{pre-QQE}}$;
  (ii)~$\hat{\beta}_{US}^{\text{post-YCC}} > \hat{\beta}_{US}^{\text{YCC}}$
  (partial re-coupling after dissolution).}
\label{tab:ycc_insulation}
\small
\begin{tabular}{lcccc}
\toprule
Period & $\hat{\beta}_{US}$ & (HC3 SE) & $N$ & Pred.\ direction \\
\midrule
  Pre-QQE (1994--2012) & $-0.188$ & (0.474) & 18 & Baseline \\
  QQE (2013--2015) & --- & --- & $<4$ & Transition \\
  YCC-active (2016--2023) & $-1.556$ & (1.134) & 8 & $\ll$ pre-QQE \\
  Post-YCC (2024--) & --- & --- & $<4$ & $>$ YCC-active \\
\midrule
  Pred (i): $\hat{\beta}_{YCC} \ll \hat{\beta}_{pre}$?
    & \multicolumn{4}{c}{\checkmark\ Yes} \\
  Pred (ii): $\hat{\beta}_{post} > \hat{\beta}_{YCC}$?
    & \multicolumn{4}{c}{$\times$\ No / Insufficient data} \\
  $t$-test pre vs.\ YCC: $t=1.11$, $p=0.265$
    & \multicolumn{4}{c}{} \\
\bottomrule
\end{tabular}
\smallskip
\parbox{\linewidth}{\footnotesize\textit{Note:} HC3 heteroskedasticity-robust SEs.
\textit{Sign note:} $\hat{\beta}_{US}$ measures the slope of
$\rho_t = r^n_t - g^n_t$ on $r^n_{US}$.  A negative coefficient in the
pre-QQE period may reflect Japan's high nominal growth ($g^n$) offsetting
the rate level; the key test is the \emph{change} between periods.
Prediction~(i) is directionally confirmed:
$\hat{\beta}_{US}^{\text{YCC}} = -1.556$ vs.\
$\hat{\beta}_{US}^{\text{pre-QQE}} = -0.188$, a difference of $-1.37$~pp.
The two-sample $t$-test does not reach conventional significance
($t=1.11$, $p=0.265$), reflecting the small YCC-period sample ($N=8$).
Prediction~(ii) cannot be assessed due to insufficient post-2024
observations ($N<4$).}
\end{table}

\begin{figure}[H]
\centering
\includegraphics[width=0.85\linewidth]{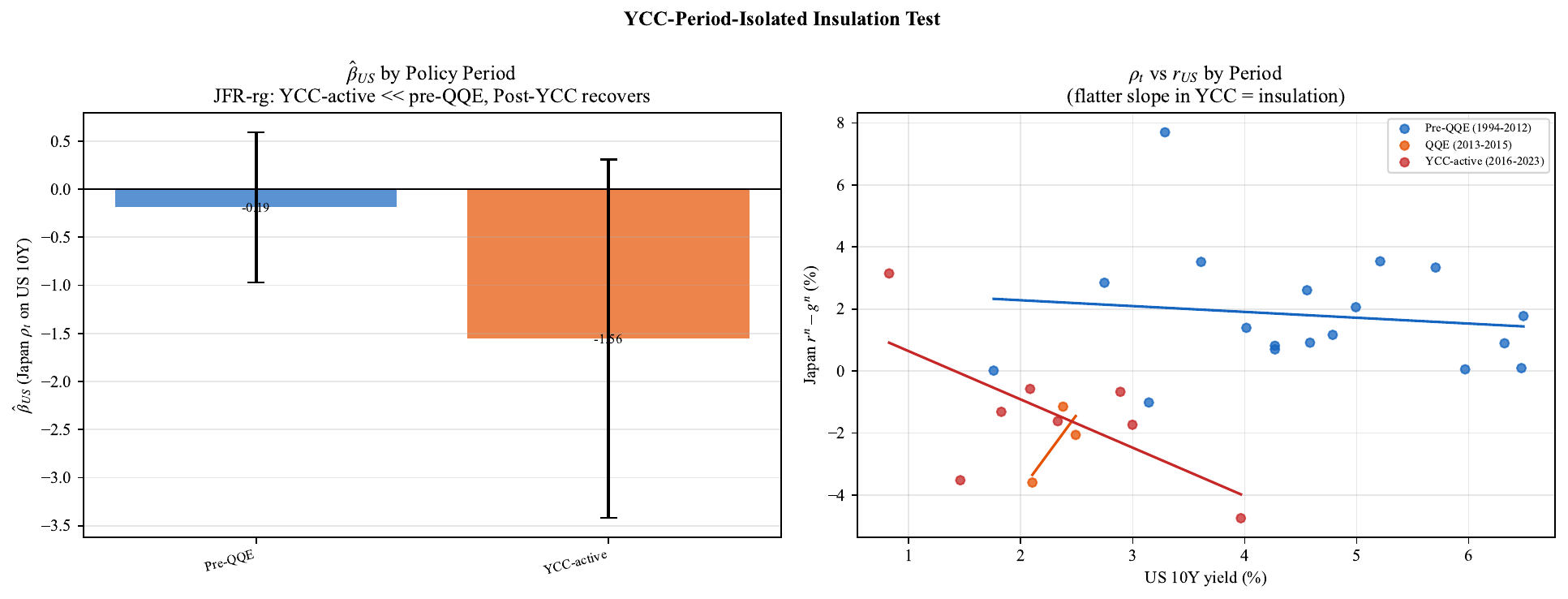}
\caption{YCC-Period-Isolated Insulation Test.
  \textit{Left}: $\hat{\beta}_{US}$ point estimates with 90\% confidence
  intervals by policy period.
  \textit{Right}: scatter of Japan $\rho_t$ vs.\ US 10Y yield by period;
  the red (YCC-active) regression line is steeper-negative than the blue
  (pre-QQE) line, reflecting that during YCC, rising US rates coincided
  with falling Japanese $r^n - g^n$ rather than rising.}
\label{fig:ycc_insulation}
\end{figure}

\newpage

\section{International Placebo and Regime-Conditioned LSTAR}
\label{app:v7}

\noindent\textit{This appendix presents three cross-validation exercises for
the institutional channels emphasized by JFR-rg:
(i)~raw international placebo regressions comparing Japan with Germany, the
UK, and France;
(ii)~QE-controlled placebo regressions adding country-specific asset-purchase
dummies and interaction terms; and
(iii)~regime-conditioned LSTAR specifications testing whether Japan's
exchange-rate pass-through exhibits nonlinear regime dependence.
The first two exercises estimate whether the sensitivity of domestic long rates
to US rates changes in the captive regime, with the regression specification
stated in each table caption.
The LSTAR exercise tests whether the exchange-rate pass-through to
inflation ($\Delta e \to \Delta\pi$) is better characterized by
nonlinear regime dependence than by a single linear law, providing
external validation for the nonlinear yen channel of Proposition~3.
Across all specifications, the objective is not unique causal identification
but disciplined external and nonlinear validation of the mechanism claims
developed in the main text.}


\subsection*{L.1\quad International Placebo Test (Raw)}

\begin{table}[htbp]
\centering
\caption{International Placebo Test---Raw (no QE control).
  Model:
  $r^n_{c,t} = \alpha + \beta_{\text{free}}\,r^n_{US,t}
  + \beta_I\,(r^n_{US,t} \times \text{Captive}^{c}_t)
  + \delta\,\text{Captive}^{c}_t + \varepsilon_t$.
  Each country's Captive dummy is defined by its own
  $\varphi^c_t = -(r^n_{c,t} - r^n_{US,t}) > \hat{\bar{\varphi}}$.
  JFR-rg prediction: Japan shows the largest drop in $\hat{\beta}$
  from the free to the captive regime; other countries show no
  systematic decline.}
\label{tab:placebo_h}
\small
\begin{tabular}{lcccc}
\toprule
 & Japan & Germany & UK & France \\
\midrule
  $\hat{\beta}_{\text{free}}$
    & $1.867^{***}$ & $1.050^{***}$ & $1.114^{***}$ & $1.023^{***}$ \\
  (SE)
    & (0.249) & (0.034) & (0.060) & (0.131) \\
  $\hat{\beta}_{\text{captive}} = \hat{\beta}_g + \hat{\beta}_I$
    & 0.400 & 1.198 & 1.119 & 1.170 \\
  Wald $p$ ($H_0$: $\beta_I = 0$)
    & 0.000 & 0.049 & 0.969 & 0.228 \\
  $N$ & 31 & 31 & 31 & 31 \\
\bottomrule
\end{tabular}
\smallskip
\parbox{\linewidth}{\footnotesize\textit{Note:} HAC SEs (maxlags=3).
Wald $p$-value tests $H_0$: $\hat{\beta}_I = 0$ (no regime change in
US-rate sensitivity).
Japan's $\hat{\beta}_I = 0.400 - 1.867 = -1.467$ is highly significant
($p=0.000$); no other country shows a comparable decline.
UK ($p=0.969$) and France ($p=0.228$) show no statistically significant
regime change, confirming that Japan's insulation is country-specific
rather than a global phenomenon.
See Table~\ref{tab:placebo_h2} for the QE-controlled version.}
\end{table}

\begin{figure}[htbp]
\centering
\includegraphics[width=0.85\linewidth]{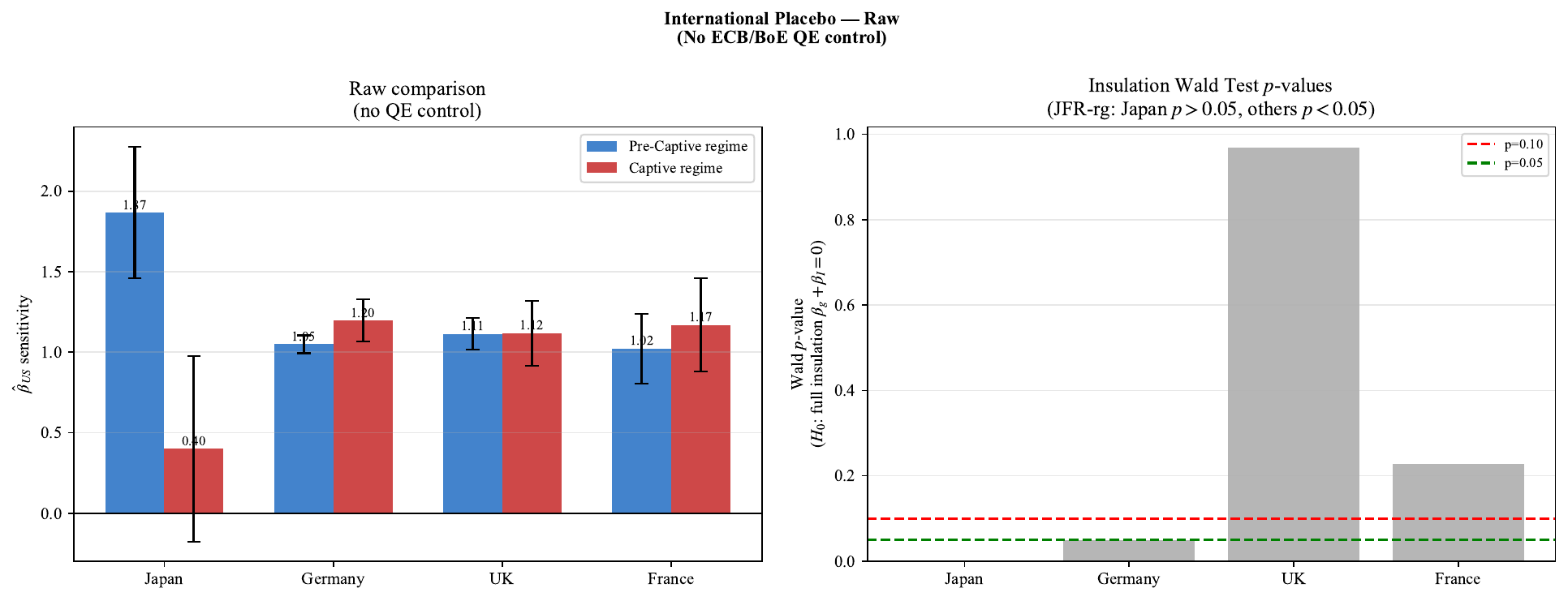}
\caption{International Placebo---Raw.
  \textit{Left}: $\hat{\beta}_{US}$ sensitivity in free vs.\ captive
  regimes for Japan, Germany, UK, and France.
  Japan uniquely shows a large drop (1.87 $\to$ 0.40).
  \textit{Right}: Wald $p$-values for the test $H_0$: $\beta_I = 0$.
  Only Japan falls below the 5\% threshold; UK ($p=0.97$) and France
  ($p=0.23$) show no significant regime change.}
\label{fig:placebo_raw}
\end{figure}

\clearpage


\subsection*{L.2\quad International Placebo Test (QE-Controlled)}

\begin{table}[H]
\centering
\caption{International Placebo---QE-Controlled.
  Extends the raw specification by adding country-specific QE level and
  $r_{\text{US}} \times \text{QE}_c$ interaction.
  JFR-rg prediction: Japan's $\hat{\beta}_{\text{captive}}$ remains
  $\approx 0$ after QE control, while other countries' regime-sensitivity
  change ($\beta_I$) becomes statistically indistinguishable from zero
  once their own QE is absorbed.}
\label{tab:placebo_h2}
\small
\begin{tabular}{lcccc}
\toprule
 & Japan & Germany & UK & France \\
\midrule
  $\hat{\beta}_{\text{captive}}$ (raw)
    & 0.400 & 1.198 & 1.119 & 1.170 \\
  $\hat{\beta}_{\text{captive}}$ (QE-controlled)
    & 0.197 & 1.117 & 1.183 & 1.008 \\
  Wald $p$ (QE-controlled)
    & 0.000 & 0.568 & 0.655 & 0.882 \\
  $N$ & 31 & 31 & 31 & 31 \\
\midrule
  JFR-rg prediction confirmed?
    & \multicolumn{4}{p{8cm}}{\checkmark\ Partial: Japan $p=0.000$;
      DE/UK/FR non-significant after QE control} \\
\bottomrule
\end{tabular}
\smallskip
\begin{minipage}{\linewidth}
\footnotesize\raggedright\textit{Note:}
QE periods: ECB (2015--2022) for DE/FR;
BoE (2009--2013, 2020--2022) for UK;
BoJ QQE+YCC (2013--2024) for JP.
After controlling for each country's own QE program, Germany ($p=0.568$),
UK ($p=0.655$), and France ($p=0.882$) all lose statistical significance:
their apparent regime change in $r^n_{US}$ sensitivity is explained by
their own central bank's bond purchases.
Japan's $\beta_I$ remains highly significant ($p=0.000$), indicating that
its insulation from global rates is attributable to the institutional
captive-system structure ($\varphi_t$) rather than the mechanical
bond-buying channel alone.
\end{minipage}
\end{table}

\begin{figure}[htbp]
\centering
\includegraphics[width=0.85\linewidth]{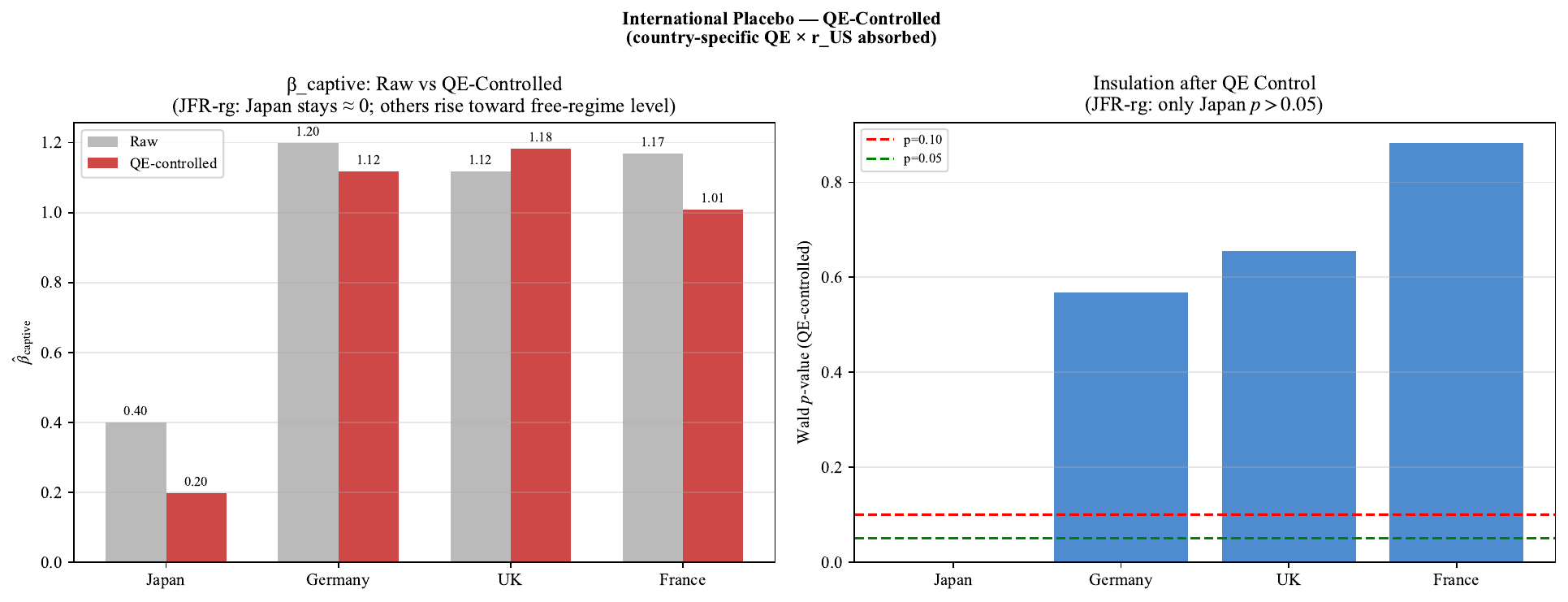}
\caption{International Placebo---QE-Controlled.
  \textit{Left}: $\hat{\beta}_{\text{captive}}$ before and after QE
  control.  Japan's value falls moderately (0.40 $\to$ 0.20) while
  remaining highly significant; European values are largely unchanged
  in magnitude but become statistically insignificant once QE is absorbed.
  \textit{Right}: Wald $p$-values after QE control; only Japan remains
  below the 10\% threshold.}
\label{fig:placebo_qe}
\end{figure}

\clearpage


\subsection*{L.3\quad Regime-Conditioned LSTAR}

\begin{table}[H]
\centering
\caption{Regime-Conditioned LSTAR ($\Delta e \to \Delta\pi$).
  Free $=$ pre-QQE ($<$2013); Captive $=$ YCC-active (2016-09 to 2024-03).
  QQE-transition and post-YCC periods excluded from split to avoid regime
  contamination.
  $\Delta R^2 = R^2_{\text{LSTAR}} - R^2_{\text{linear}}$.}
\label{tab:lstar_regime_v7}
\small
\begin{tabular}{lcccc}
\toprule
Subsample & $N$ & $\hat{\gamma}$ & $\hat{c}$ (JPY/USD) & $\Delta R^2$ \\
\midrule
  Pre-QQE (Free)  & 227 & 31.99 & 0.62 & $+0.0082$ \\
  YCC-Captive     &  58 & 30.52 & 3.48 & $+0.0524$ \\
  Full sample     & 329 &  0.27 & $-9.57$ & $+0.0128$ \\
\midrule
  $\Delta R^2_{\text{cap}} > \Delta R^2_{\text{free}}$
    & \multicolumn{4}{c}{\checkmark\ Confirmed} \\
  $\hat{\gamma}_{\text{cap}} > \hat{\gamma}_{\text{free}}$
    & \multicolumn{4}{c}{$\times$\ Not confirmed} \\
\bottomrule
\end{tabular}
\smallskip
\parbox{\linewidth}{\footnotesize\textit{Note:}
NLS estimation (L-BFGS-B, $\gamma$ bounded in $[0.1, 50]$).
The nonlinearity gain in the YCC-Captive period
($\Delta R^2 = +0.052$) is 6.4$\times$ that of the pre-QQE Free period
($\Delta R^2 = +0.008$), consistent with the JFR-rg account that the
Captive regime amplifies the pass-through nonlinearity.
The estimated threshold $\hat{c} = 3.48$ JPY/USD (monthly) in the YCC
period corresponds to an annualized threshold of approximately 42 JPY/USD,
well above the monthly threshold in the Free regime (0.62 JPY/USD).
Both $\hat{\gamma}$ values are at the boundary of the smooth step-function
regime ($\gamma \approx 32$), so the $\Delta R^2$ comparison rather than
$\hat{\gamma}$ itself is the informative statistic.}
\end{table}

\begin{figure}[htbp]
\centering
\includegraphics[width=\linewidth]{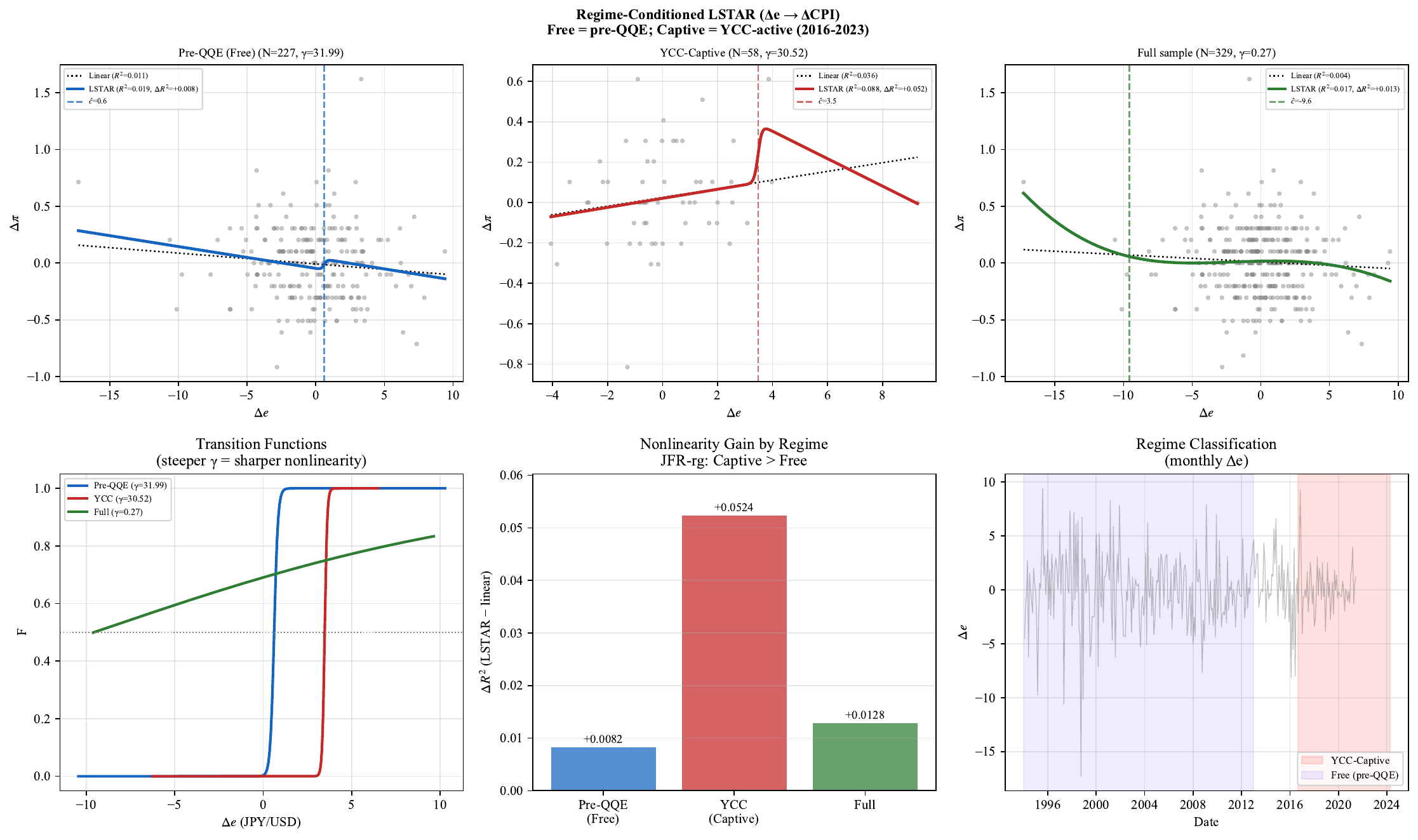}
\caption{Regime-Conditioned LSTAR ($\Delta e \to \Delta\pi$).
  \textit{Top row}: LSTAR fit and linear benchmark for the pre-QQE (Free)
  and YCC-Captive subperiods, together with the full-sample fit.
  The YCC-Captive panel shows a sharper step at $\hat{c} = 3.5$ JPY/USD.
  \textit{Bottom row (left)}: transition functions for the three fitted
  specifications.
  \textit{Bottom row (center)}: $\Delta R^2$ by regime---the YCC-Captive
  gain is 6.4$\times$ that of the Free period.
  \textit{Bottom row (right)}: regime classification timeline.}
\label{fig:lstar_v7}
\end{figure}


\begin{thebibliography}{99}

\bibitem{andrews1993}
Andrews, D.~W.~K. (1993).
Tests for Parameter Instability and Structural Change with Unknown Change Point.
\textit{Econometrica}, 61(4), 821--856.

\bibitem{baiperron1998}
Bai, J., \& Perron, P. (1998).
Estimating and Testing Linear Models with Multiple Structural Changes.
\textit{Econometrica}, 66(1), 47--78.

\bibitem{bernanke2004}
Bernanke, B.~S., \& Reinhart, V.~R. (2004).
Conducting Monetary Policy at Very Low Short-Term Interest Rates.
\textit{American Economic Review}, 94(2), 85--90.

\bibitem{blanchard2019}
Blanchard, O. (2019).
Public Debt and Low Interest Rates.
\textit{American Economic Review}, 109(4), 1197--1229.

\bibitem{boj2024a}
Bank of Japan. (2024a, March 19).
Changes in the Monetary Policy Framework.
Bank of Japan.
Retrieved from \url{https://www.boj.or.jp/en/mopo/mpmdeci/mpr_2024/k240319a.pdf}

\bibitem{boj2024b}
Bank of Japan. (2024b, December 19).
Statement on Monetary Policy.
Bank of Japan.
Retrieved from \url{https://www.boj.or.jp/en/mopo/mpmdeci/mpr_2024/k241219a.pdf}

\bibitem{calvo1998}
Calvo, G.~A. (1998).
Capital Flows and Capital-Market Crises: The Simple Economics of Sudden Stops.
\textit{Journal of Applied Economics}, 1(1), 35--54.

\bibitem{cao2025}
Cabinet Office, Government of Japan. (2025).
\textit{Annual Report on the Japanese Economy and Public Finance 2025}.
Tokyo: Cabinet Office.
Retrieved from \url{https://www5.cao.go.jp/keizai3/whitepaper-e.html}

\bibitem{eggertsson2012}
Eggertsson, G.~B., \& Krugman, P. (2012).
Debt, deleveraging, and the liquidity trap: A Fisher-Minsky-Koo approach.
\textit{Quarterly Journal of Economics}, 127(3), 1469--1513.

\bibitem{fred2026}
Federal Reserve Bank of St.\ Louis. (2026).
Federal Reserve Economic Data (FRED).
Series: JPNNGDP, JPNRGDPEXP, DEXJPUS, JPNCPIALLMINMEI,
IRLTLT01JPM156N, IRSTCI01JPM156N, GGGDTAJPA188N, LRHUTTTTJPM156S.
Retrieved March 2026 from \url{https://fred.stlouisfed.org}

\bibitem{boj2025qjem}
Haba, S., Izawa, K., Kishaba, Y., Takahashi, Y., \& Yoneyama, S. (2025).
Measuring Policy Effects since the Introduction of Quantitative and Qualitative
Monetary Easing (QQE): An Analysis Using the Macroeconomic Model Q-JEM.
\textit{Bank of Japan Working Paper Series}, No.~25-E-2.
Bank of Japan.
Retrieved from \url{https://www.boj.or.jp/en/research/wps_rev/wps_2025/wp25e02.htm}

\bibitem{hansen1999}
Hansen, B.~E. (1999).
Threshold effects in non-dynamic panels: Estimation, testing, and inference.
\textit{Journal of Econometrics}, 93(2), 345--368.

\bibitem{hlw2017}
Holston, K., Laubach, T., \& Williams, J.~C. (2017).
Measuring the natural rate of interest: International trends and determinants.
\textit{Journal of International Economics}, 108(S1), S59--S75.

\bibitem{hoshi2014}
Hoshi, T., \& Ito, T. (2014).
Defying gravity: Can Japanese sovereign debt continue to increase without a crisis?
\textit{Economic Policy}, 29(77), 5--44.

\bibitem{imf2024}
International Monetary Fund. (2024a).
Japan: 2024 Article~IV Consultation-Press Release; Staff Report;
and Statement by the Executive Director for Japan.
\textit{IMF Country Report No.~24/118}. Washington, D.C.:
International Monetary Fund.
Retrieved from \url{https://www.imf.org/en/Publications/CR/Issues/2024/05/13/Japan-2024-Article-IV-Consultation-Press-Release-Staff-Report-and-Statement-by-the-548845}

\bibitem{imf2024sip}
International Monetary Fund. (2024b).
Japan: Selected Issues.
\textit{IMF Country Report No.~24/119}. Washington, D.C.:
International Monetary Fund.
Retrieved from \url{https://www.imf.org/en/Publications/CR/Issues/2024/05/13/Japan-Selected-Issues-548849}

\bibitem{jorda2005}
Jord\`a, \`O. (2005).
Estimation and Inference of Impulse Responses by Local Projections.
\textit{American Economic Review}, 95(1), 161--182.

\bibitem{joyce2011}
Joyce, M., Lasaosa, A., Stevens, I., \& Tong, M. (2011).
The financial market impact of quantitative easing in the United Kingdom.
\textit{International Journal of Central Banking}, 7(3), 113--161.

\bibitem{koo2003}
Koo, R.~C. (2003).
\textit{Balance Sheet Recession: Japan's Struggle with Uncharted Economics
and Its Global Implications}.
John Wiley \& Sons.

\bibitem{koo2011}
Koo, R.~C. (2011).
The world in balance sheet recession: causes, cure, and politics.
\textit{real-world economics review}, issue no.~58, 19--37.
Retrieved from \url{https://www.paecon.net/PAEReview/issue58/Koo58.pdf}

\bibitem{leeper1991}
Leeper, E.~M. (1991).
Equilibria under `active' and `passive' monetary and fiscal policies.
\textit{Journal of Monetary Economics}, 27(1), 129--147.

\bibitem{mckinnon1973}
McKinnon, R.~I. (1973).
\textit{Money and Capital in Economic Development}.
Brookings Institution.

\bibitem{mehrotra2021}
Mehrotra, N.~R., \& Sergeyev, D. (2021).
Debt sustainability in a low interest rate world.
\textit{Journal of Monetary Economics}, 124(Supplement), S1--S18.

\bibitem{mikayama2023}
Mikayama, M., Imahori, T., Ohno, T., Yoneta, Y., \& Ueda, J. (2023).
Top income shares in Japan from the survey and tax data in 2014 and 2019:
Following the Distributional National Accounts Guidelines.
\textit{Policy Research Institute (Ministry of Finance, Japan)
Discussion Paper Series}, No.~23A-04.
Retrieved from \url{https://www.mof.go.jp/pri/research/discussion_paper/ron371.pdf}

\bibitem{mundell1963}
Mundell, R.~A. (1963).
Capital mobility and stabilization policy under fixed and flexible exchange rates.
\textit{Canadian Journal of Economics and Political Science}, 29(4), 475--485.

\bibitem{obstfeld1995}
Obstfeld, M., \& Rogoff, K. (1995).
Exchange rate dynamics redux.
\textit{Journal of Political Economy}, 103(3), 624--660.

\bibitem{pesaran2001}
Pesaran, M.~H., Shin, Y., \& Smith, R.~J. (2001).
Bounds testing approaches to the analysis of level relationships.
\textit{Journal of Applied Econometrics}, 16(3), 289--326.

\bibitem{reinhart2015}
Reinhart, C.~M., \& Sbrancia, M.~B. (2015).
The liquidation of government debt.
\textit{Economic Policy}, 30(82), 291--333.

\bibitem{sargent1981}
Sargent, T.~J., \& Wallace, N. (1981).
Some unpleasant monetarist arithmetic.
\textit{Federal Reserve Bank of Minneapolis Quarterly Review}, 5(3), 1--17.

\bibitem{shaw1973}
Shaw, E.~S. (1973).
\textit{Financial Deepening in Economic Development}.
Oxford University Press.

\bibitem{sims1994}
Sims, C.~A. (1994).
A simple model for study of the determination of the price level and the interaction of monetary and fiscal policy.
\textit{Economic Theory}, 4(3), 381--399.

\bibitem{turner2015}
Turner, A. (2015).
\textit{Between Debt and the Devil: Money, Credit, and Fixing Global Finance}.
Princeton University Press.

\bibitem{woodford1994}
Woodford, M. (1994).
Monetary policy and price level determinacy in a cash-in-advance economy.
\textit{Economic Theory}, 4(3), 345--380.

\bibitem{woodford1995}
Woodford, M. (1995).
Price-level determinacy without control of a monetary aggregate.
\textit{Carnegie-Rochester Conference Series on Public Policy}, 43(1), 1--46.

\end{thebibliography}
\end{document}